\documentclass[12pt]{article}

\usepackage[left=1in,top=1in,right=1in,bottom=1in,head=.1in,nofoot]{geometry}
\usepackage{setspace,url,bm,amsmath} 
\usepackage{titlesec} 
\titlelabel{\thetitle.\quad} 
\usepackage[margin=20pt]{subcaption}
\usepackage{amsmath,amsfonts,amsthm,amssymb}
\usepackage{graphicx}
\usepackage[colorlinks=true,linkcolor=black,citecolor=blue,urlcolor=blue]{hyperref}
\usepackage[table]{xcolor}
\usepackage{multirow}
\usepackage{array}
\usepackage{booktabs}
\usepackage{threeparttable}
\usepackage{comment}
\usepackage{float}
\usepackage{natbib}
\usepackage{indentfirst}

\newtheorem*{theorem*}{Theorem}
\newtheorem{theorem}{Theorem}
\newtheorem{proposition}{Proposition}
\newtheorem{lemma}{Lemma}
\newtheorem{remark}{Remark}
\newtheorem{example}{Example}
\newtheorem{definition}{Definition}
\newtheorem{corollary}{Corollary}
\newtheorem*{corollary*}{Corollary}
\newtheorem{condition}{Condition}

\newcommand{\LXTL}[1]{\textcolor{black}{#1}}

\def\top{{\mathrm{\scriptscriptstyle T}}}

\def\tdx{\tilde{{X}}}
\def\tdxi{{\tilde{x}}_{i\cdot}}

\def\tdyi{\tilde{Y}_{i\cdot}}
\def\tdy{\tilde{Y}}
\def\tdadjy{\tilde{Y}^{\textnormal{adj}}}
\def\tdadjyi{\tilde{Y}^{\textnormal{adj}}_{i\cdot}}
\def\tdadje{\tilde{\varepsilon}^{\textnormal{adj}}}
\def\tdadjei{\tilde{\varepsilon }^{\textnormal{adj}}_{i\cdot}}
\def\tdei{\tilde{\varepsilon }_{i\cdot}}
\def\tde{\tilde{\varepsilon}}

\def\tauht{\hat{\tau}_{ \scalebox{0.6}{$\textnormal{ht}$}}}
\def\tauhtb{\check{\tau}_{\scalebox{0.6}{$\textnormal{ht}$}}^{\scalebox{0.6}{$\textnormal{adj}$}}}

\def\tauhaj{\hat{\tau}_{ \scalebox{0.6}{$\textnormal{haj}$}}}
\def\tauhajb{\check{\tau}_{\scalebox{0.6}{$\textnormal{haj}$}}^{\scalebox{0.6}{$\textnormal{adj}$}}}

\def\taucht{{\hat{\tau}_{ \scalebox{0.6}{\textnormal{ht},${c}$ }}}}
\def\tauvht{{\hat{\tau}_{\textnormal{ht}, v}}}
\def\tauxhaj{{\hat{\tau}_{ \scalebox{0.6}{\textnormal{haj,${x}$}}}}}

\def\tauadjhaj{{\hat{\tau}_{ \scalebox{0.6}{$\textnormal{haj}$}}^{\scalebox{0.6}{$\textnormal{adj}$}}}}

\def\tauadjht{\hat{\tau}_{ \scalebox{0.6}{$\textnormal{ht}$}}^{\scalebox{0.6}{$\textnormal{adj}$}}}

\def\vttht{V_{\textnormal{ht},\tau\tau}}
\def\vtthtb{V_{\textnormal{ht},\tau\tau}^{\scalebox{0.6}{$\textnormal{adj}$}}}

\def\vtcht{V_{\textnormal{ht},\tau c}}
\def\vtchtb{V^{\scalebox{0.6}{$\textnormal{adj}$}}_{\textnormal{ht},\tau c}}

\def\vctht{V_{\textnormal{ht},c\tau}}
\def\vcthtb{V^{\scalebox{0.6}{$\textnormal{adj}$}}_{\textnormal{ht},c\tau }}

\def\vccht{V_{\textnormal{ht},cc}}

\def\vtthaj{V_{\textnormal{haj},\tau\tau}}
\def\vtxhaj{V_{\textnormal{haj},\tau x}}
\def\vxthaj{V_{\textnormal{haj},x\tau}}
\def\vxxhaj{V_{\textnormal{haj},xx}}

\def\vtthajb{V_{\textnormal{haj},\tau\tau}^{\scalebox{0.6}{$\textnormal{adj}$}}}

\def\vtxhajb{V^{\scalebox{0.6}{$\textnormal{adj}$}}_{\textnormal{haj},\tau x}}

\def\vxthajb{V^{\scalebox{0.6}{$\textnormal{adj}$}}_{\textnormal{haj},x\tau }}

\def\hatvtthajb{\hat{V}_{\textnormal{haj},\tau\tau}^{\scalebox{0.6}{$\textnormal{adj}$}}}
\def\hatvtthtb{\hat{V}_{\textnormal{ht},\tau\tau}^{\scalebox{0.6}{$\textnormal{adj}$}}}

\def\mx{\mathcal{M}_{x}}

\def\mc{\mathcal{M}_{c}}

\def\betaw{{{\beta}_{w}}}
\def\betav{{{\beta}_{v}}}

\def\hbetaw{{{\hat \beta}_{w}}}
\def\hbetav{{{\hat \beta}_{v}}}

\def\bary{\bar Y}
\def\dca{\mathcal{D}_{c}(A_c)}
\def\dxa{\mathcal{D}_{x}(A_x)}

\def\var{\operatorname{var}}
\def\cov{\operatorname{cov}}
\def\hcovfz{\hat{\operatorname{cov}}_{\textnormal{f,z}}}
\def\hcovft{\hat{\operatorname{cov}}_{\textnormal{f,1}}}
\def\hcovfc{\hat{\operatorname{cov}}_{\textnormal{f,0}}}
\def\hcovf{\hat{\operatorname{cov}}_{\textnormal{f}}}
\def\cova{\textnormal{cov}_{\textnormal{a}}}
\def\pra{\textnormal{pr}_{\textnormal{a}}}
\def\vara{\textnormal{var}_{\textnormal{a}}}
\def\varf{\operatorname{var}_{\textnormal{f}}}
\def\hvarfz{\hat{\operatorname{var}}_{\textnormal{f,z}}}

\def\hvarf{\hat{\operatorname{var}}_{\textnormal{f}}}
\def\covf{\operatorname{cov}_{\textnormal{f}}}
\def\vhathw{\hat{V}_{\scalebox{.5}{$\textnormal{HW}$}}}
\def\vhatadjhw{\hat{V}^{\textnormal{adj}}_{\scalebox{.5}{$\textnormal{HW}$}}}
\def\vhatlz{\hat{V}_{\scalebox{.5}{$\textnormal{LZ}$}}}
\def\vhatadjlz{\hat{V}^{\textnormal{adj}}_{\scalebox{.5}{$\textnormal{LZ}$}}}

\def\radjx{R^{\textnormal{adj}}_{x}}
\def\radjc{R^{\textnormal{adj}}_{c}}
\def\hradjx{\hat{R}^{\textnormal{adj}}_{x}}
\def\hradjc{\hat{R}^{\textnormal{adj}}_{c}}

\begin{document}
\begin{singlespace}
\title{\bf Design-based theory for cluster rerandomization}

\author{
\small
{
Xin Lu$^{1}$, Tianle Liu$^{2}$, Hanzhong Liu$^{1}$, Peng Ding$^{3}$
}
\\ \\
{\small $^{1}$ Center for Statistical Science, Department of Industrial Engineering, Tsinghua University, Beijing, China}\\
{\small $^{2}$ Department of Statistics, Harvard University, Cambridge, MA, U.S.A.}\\
{\small $^{3}$ Department of Statistics, University of California, Berkeley, California, U.S.A.}
}

\date{}
\maketitle
\end{singlespace}

\thispagestyle{empty}
\vskip -8mm 

\begin{singlespace}
\begin{abstract}

Complete randomization balances covariates on average, but covariate imbalance often exists in finite samples. Rerandomization can ensure covariate balance in the realized experiment by discarding the undesired treatment assignments. Many field experiments in public health and social sciences assign the treatment at the cluster level due to logistical constraints or policy considerations. Moreover, they are frequently combined with rerandomization in the design stage. We refer to cluster rerandomization as a cluster-randomized experiment compounded with rerandomization to balance covariates at the individual  or cluster level. Existing asymptotic theory can only deal with rerandomization with treatments assigned at the individual level, leaving that for cluster rerandomization an open problem. To fill the gap, we provide a design-based theory for cluster rerandomization. Moreover, we compare two cluster rerandomization schemes that use prior information on the importance of the covariates: one based on the weighted Euclidean distance and the other based on the Mahalanobis distance with tiers of covariates. We demonstrate that the former dominates the latter with optimal weights and orthogonalized covariates. Last but not least, we discuss the role of covariate adjustment in the analysis stage and recommend covariate-adjusted procedures that can be conveniently implemented by least squares with the associated robust standard errors.

\vspace{12pt}
\noindent {\bf Key words}: cluster randomization; cluster-robust standard error; constrained randomization; covariate balance
\end{abstract}

\end{singlespace}

\newpage

\clearpage
\setcounter{page}{1}

\allowdisplaybreaks
\baselineskip=24pt

\begin{singlespace}

\section{Introduction}\label{sec:intro}


Cluster randomization has been widely used in public health \citep{donner2000design, turner2017reviewdesign, turner2017reviewanalysis, hayes2017cluster} and social sciences \citep{raudenbush1997statistical, Schochet2013, Athey2017, schochet2020analyzing}. It assigns the treatments at the cluster level, with units within a cluster receiving the same treatment or control condition, which helps to avoid interference within clusters and is applicable when individual-level assignments are logistically infeasible.

Before experiments, researchers often collect covariates at the individual or cluster level. For instance, in clinical trials, individual-level covariates may include gender, race, age, and education of patients while cluster-level covariates may include the capacity of the clinics and whether they are in a metropolitan area \citep{li2016evaluation}. The probability of having imbalance in covariates after treatment assignments is non-negligible and would complicate the interpretation of the experimental results, especially if some particular covariates are predictive to the outcome.

A common approach to tackle covariate imbalance is to perform regression adjustment in the analysis stage of an experiment \citep{Lin2013Agnostic, su2021modelassisted}. In contrast,
rerandomization actively reduces the imbalance in the design stage. Rerandomization, also known as constrained randomization \citep{moulton2004covariate, morgan2012rerandomization, li2016evaluation, li2017evaluation, Li9157}, is a procedure where one discards the undesired assignments that result in highly unbalanced covariates and repeats randomization until a proper assignment appears. The existing design-based theory for rerandomization assumes that the treatments are assigned at the individual level \citep{morgan2012rerandomization, Li9157}, leaving it unclear how rerandomization would theoretically work in cluster experiments. The overarching goal of this paper is to fill this gap.


We study cluster rerandomization with both individual- and cluster-level covariates, and derive an asymptotic theory for estimators either with or without regression adjustment. Our theory precisely quantifies the asymptotic efficiency gain of performing cluster rerandomization at both levels. We further show that under cluster rerandomization, the adjusted estimators based on linear regressions recommended by \cite{su2021modelassisted} still have desired theoretical guarantees. All these inferential procedures can be conveniently realized by standard statistical packages. Moreover, previous asymptotic theory mainly focuses on the rerandomization scheme based on the Mahalanobis distance \citep{morgan2012rerandomization, Li9157, Li2020factorial} while we also discuss a distinct but widely-adopted practical scheme based on the weighted Euclidean distance where one could utilize prior knowledge of the covariate importance \citep{raab2001balance, wight2002limits,althabe2008behavioral, de2012best, li2016evaluation,li2017evaluation,hayes2017cluster,dempsey2018effect}. We show that for orthogonalized covariates with optimal weights, it dominates the cluster rerandomization with tiers of covariates, another scheme accounting for the covariate importance, the non-cluster version of which is proposed in \cite{morgankari2015}.

To facilitate the discussion, we introduce the following notation. Let $\operatorname{det}(\cdot)$ denote the determinant of a matrix. \LXTL{For a real-valued matrix $A$, we write $A>0$ if $A$ is strictly positive definite.} For a finite population $\{a_1, \ldots, a_M\}$, let $\var_\textup{f}(a) = (M-1)^{-1} \sum_{i=1}^M (a_i - \bar{a})^2$ denote its finite-population variance, where $\bar{a}=M^{-1} \sum_{i=1}^M a_i$. Analogously, let $\cov_\textup{f}$ denote the finite-population covariance. For two random sequences $A$ and $B$, let $A\mathrel{\dot{\sim}}B$ indicate that they have the same asymptotic distribution. Let $\pra$, $\vara$ and $\cova$ denote the asymptotic probability, variance, and covariance, respectively. Besides, our asymptotic theory requires mild moment conditions, but we relegate them to the appendix to avoid clutter in the main text.


\section{Notation and review for cluster-randomized experiments}\label{sec:crt}

Consider a cluster-randomized experiment with a finite population of $N$ units grouped into $M$ clusters. The experimenter randomly assigns $M_1$ clusters to the treatment arm and $M_0=M-M_1$ clusters to the control arm, with units in the same cluster receiving the same treatment or control. For cluster $i$, let $Z_i$ be the treatment indicator with $Z_i=1$ when it is assigned to the treatment arm and $Z_i=0$ when it is assigned to the control arm. Let $n_i$ be the size of cluster $i$, with $\sum_{i=1}^M n_i = N$. Let $(i,j)$ or simply $ij$ to index unit $j$ in cluster $i$. Let $Z_{ij}$ be the treatment indicator for unit $(i,j)$ and by design, $Z_{ij} = Z_{i}$ for $j=1,\ldots,n_i$. Let $N_1=\sum_{i=1}^{M}Z_in_i$ and $N_0=N-N_1$  denote the numbers of units assigned to the treatment and control arms, respectively. Importantly, $N_1$ and $N_0$ are random if the $ n_i $'s vary. Let $Y_{ij}(1)$ and $Y_{ij}(0)$  denote the potential outcomes of  unit $(i,j)$ under the treatment and control, respectively. We adopt the design-based framework by conditioning on the potential outcomes with the $Z_i$'s being the only source of randomness. The observed outcome $Y_{ij}= Z_iY_{ij}(1)+(1-Z_i)Y_{ij}(0)$ is random in general due to the random $Z_i$'s. We are interested in estimating the average treatment effect
$$
  \tau = N^{-1}\sum_{i=1}^{M}\sum_{j=1}^{n_i} \{Y_{ij}(1)-Y_{ij}(0)\}
$$
based on the observed data.

The  difference-in-means estimator of $\tau$ is given by
\begin{equation*}
  \tauhaj   = N_1^{-1}  \sum_{i=1}^M\sum_{j=1}^{n_i}Z_{ij}Y_{ij}  - N_0^{-1}  \sum_{i=1}^M\sum_{j=1}^{n_i} (1-Z_{ij})Y_{ij}  ,
\end{equation*}
with the subscript ``haj'' signifying that it is the  Hajek estimator based on the inverse of the treatment probabilities \citep{su2021modelassisted}.
As pointed out by \cite{middleton2015unbiased}, the random denominators in $\tauhaj$ can result in bias when the $n_i$'s vary. The unbiased Horvitz--Thompson estimator replaces the random denominators $N_1$ and $N_0$ with their expectations:
\begin{equation*}
  \tauht = (NM_1/M)^{-1}\sum_{i=1}^M Z_i\sum_{j=1}^{n_i}Y_{ij} -(NM_0/M)^{-1}\sum_{i=1}^M  (1-Z_i)\sum_{j=1}^{n_i}Y_{ij} .
\end{equation*}
Define the scaled cluster total potential outcome as $\tilde{Y}_{i\cdot}(z) = \sum_{j=1}^{n_i}Y_{ij}(z)M/N$ so that 
$
  \tau = M^{-1}\sum_{i=1}^M \{ \tilde{Y}_{i\cdot}(1) - \tilde{Y}_{i\cdot}(0)\}.
$
Based on the observed analog $\tilde{Y}_{i\cdot} = Z_i \tilde{Y}_{i\cdot}(1) + (1-Z_i)\tilde{Y}_{i\cdot}(0)  $, the  Horvitz--Thompson estimator simplifies to
\begin{equation}
  \label{eq:tauhtnew}
  \tauht   = M_1^{-1}\sum_{i=1}^MZ_i\tilde{Y}_{i\cdot}-M_0^{-1}\sum_{i=1}^M (1-Z_i)\tilde{Y}_{i\cdot}. \nonumber
\end{equation}
This is unbiased for $\tau$ because it is the scaled difference in means of the potential outcomes if each cluster is treated as a whole.



Both $\tauhaj$ and $\tauht$ are asymptotically normal under some mild regularity conditions \citep{su2021modelassisted}.
Let $e_z =M_z/M$ be the proportion of clusters assigned to the treatment arm $z$ ($z=0,1$). Denote the normalized cluster size vector by $\tilde{\omega}=(\tilde{\omega}_1,\ldots,\tilde{\omega}_M)^{\top}$, with $\tilde{\omega}_i = n_iM/N$. Let $\bar{Y}(z) = N^{-1}\sum_{i=1}^{M}\sum_{j=1}^{n_i}Y_{ij}(z)$.
Define $\tilde{\varepsilon }(z)=(\tilde{\varepsilon} _{1\cdot}(z),\ldots,\tilde{\varepsilon} _{M\cdot}(z))^{\top}$ and $\tilde{Y}(z)=(\tilde{Y}_{1\cdot}(z),\ldots,\tilde{Y}_{M\cdot}(z))^{\top}$, where
\begin{equation*}
  \tdei(z) = \tilde{Y}_{i\cdot}(z)-\tilde{\omega}_i\bar{Y}(z)= \sum_{j=1}^{n_i}\varepsilon _{ij}(z)M/N, \quad \varepsilon _{ij}(z) = Y_{ij}(z)-\bar{Y}(z).
\end{equation*}
\citet{su2021modelassisted} showed that $M^{1/2} (\hat\tau_{\star}  -\tau)\mathrel{\ \dot{\sim}\ } \mathcal N(0, V_{\star, \tau\tau})$ for $\star = \textnormal{ht}, \textnormal{haj}$, where
\begin{gather}
  \vttht   = e_1 ^{-1}\varf   \{\tilde{Y}(1) \}  +  e_0 ^{-1} \varf   \{\tilde{Y}(0) \} - \varf   \{\tilde{Y}(1)-\tilde{Y}(0) \}, \label{eq::ht-asym-var} \\
  \vtthaj  = e_1 ^{-1} \varf   \{\tilde{{\varepsilon} } (1) \} +  e_0 ^{-1}\varf   \{\tilde{{\varepsilon} } (0) \} -
  \varf   \{\tilde{{\varepsilon} } (1)-\tilde{{\varepsilon} } (0) \}. \label{eq::haj-asym-var}
\end{gather}
Therefore, $\tauht$ is unbiased in finite samples and $\tauhaj$ is unbiased asymptotically. By \cite{middleton2015unbiased} and \cite{su2021modelassisted}, $\tauhaj$ often practically outperforms $\tauht$ in terms of the mean squared error.

\section{Cluster rerandomization using the Mahalanobis distance}\label{sec:clre}

\subsection{Design}\label{sec:design}


In a cluster-randomized experiment, suppose we observe baseline covariates at the individual  or cluster level.
Let $x_{ij}=(x_{ij1},\ldots,x_{ijK})^{\top}\!\!$ denote the individual-level covariates for unit $(i,j)$, and $c_i=(c_{i1},\ldots,c_{iK})^{\top}$ be the cluster-level covariates for cluster $i$. To simplify the presentation, we center them to ensure that $\bar{x} = N^{-1}\sum_{i=1}^{M} \sum_{j=1}^{n_i} x_{ij}=0$ and $M^{-1}\sum_{i=1}^M c_i = 0$, and use the same $K$ to denote their dimensions, although they may differ. Let $\tilde{x}_{i\cdot} = \sum_{j=1}^{n_i}x_{ij}M/N$ be the scaled cluster totals, concatenated as  $\tilde{X} = (\tilde{x}_{1\cdot},\ldots,\tilde{x}_{M\cdot})^{\top}$. Similarly, let $C=(c_1,\ldots,c_M)^{\top}$.


Two types of covariates suggest two schemes of cluster rerandomization. The first one is based on the Mahalanobis distance of covariates at the cluster level. Define
\begin{equation*}
  \taucht   =M_1^{-1} \sum_{i=1}^M Z_i c_i - M_0^{-1}\sum_{i=1}^M (1-Z_i)c_i
\end{equation*}
as the Horvitz--Thompson estimator applied to the cluster-level covariates. For a pre-determined threshold  $a > 0$, cluster rerandomization based on cluster-level covariates accepts the treatment assignment if and only if the following event happens:
$$
  \mc  = \{  \taucht ^{\top} \cov (\taucht) ^{-1}\ \taucht  \leq a \}
  =  \{ e_1 e_0  M\ \taucht ^{\top} \covf ({C}) ^{-1}\ \taucht  \leq a \} .
$$
The second one is based on the Mahalanobis distance of covariates at the individual level. Define
$$
  \tauxhaj   = N_1^{-1} \sum_{i=1}^M Z_i\sum_{j=1}^{n_i}x_{ij} -N_0^{-1}  \sum_{i=1}^M (1-Z_i)\sum_{j=1}^{n_i}x_{ij}
$$
as the difference in means of the individual-level covariates. For a pre-determined threshold  $a > 0$, cluster rerandomization based on individual-level covariates accepts the treatment assignment if and only if the following event happens:
\begin{equation}\label{eq:ddif}
  \mx =  \{  \tauxhaj ^{\top} \cova (\tauxhaj) ^{-1}\ \tauxhaj \leq a \}
  = \{e_1 e_0  M\  \tauxhaj  ^{\top} \covf(\tilde{{X}}) ^{-1} \tauxhaj  \leq a\}.
\end{equation}
We use the same threshold $a$ for two types of cluster rerandomization to simplify the notation. Due to the complicated form of the exact covariance of $\tauxhaj$, we have used its asymptotic analog in \eqref{eq:ddif}, with more details given in  Proposition~\ref{prop:FCLT} in the next section.

\LXTL{The definition of cluster-level or individual-level covariates can be subtle, and different levels of covariates could be used interchangeably. On the one hand, we could take some statistics of aggregated individual-level covariates within a cluster to construct cluster-level covariates. For instance, we could define ${c}_{i} = \tilde{{x}}_{i\cdot}$. On the other hand, it is common that only cluster-level covariates $c_i$ are observed prior to rerandomization, in which case we can define individual-level covariates as $x_{ij}=c_i$ for all $j$ in cluster $i$. Hence, both cluster rerandomization schemes are applicable.}

\subsection{Asymptotic distributions}

We have introduced two estimators $\tauht$ and $\tauhaj  $ for $\tau$, and two cluster rerandomization schemes $\mc $ and $\mx$. Theoretically, they generate four combinations. Nevertheless, it is more natural to consider the design and analysis with the same type. Therefore, we will focus on the asymptotic distributions of $\tauht$ given $\mc$ and $ \tauhaj$ given $\mx$. We start with unconditional joint asymptotic distributions of $(\tauhaj ,  \tauxhaj )$ and $(    \tauht,     \taucht )$ as shown in Proposition~\ref{prop:FCLT} below.


%
%
%
%

\begin{proposition}\label{prop:FCLT}
  Under regularity conditions,
  \begin{gather*}
    M^{1/2}
    \left(
    \begin{array}{c}
        \tauhaj  -\tau \\
        \tauxhaj
      \end{array}
    \right) \mathrel{\ \dot{\sim}\ }
    \mathcal{N} \left( 0 ,
    \begin{bmatrix}
        \vtthaj & \vtxhaj \\
        \vxthaj & \vxxhaj
      \end{bmatrix}
    \right),\\
    M^{1/2}
    \left(
    \begin{array}{c}
        \tauht  -\tau \\
        \taucht
      \end{array}
    \right)
    \mathrel{\ \dot{\sim}\ }  \mathcal{N}  \left( 0 ,
    \begin{bmatrix}
        \vttht & \vtcht \\
        \vctht & \vccht
      \end{bmatrix} \right),
  \end{gather*}
  where $ \vttht $ and  $      \vtthaj $ are defined in \eqref{eq::ht-asym-var} and \eqref{eq::haj-asym-var}, respectively, and
  \begin{eqnarray*}
    \vxthaj &=&(\vtxhaj)^{\top}=  e_1 ^{-1}\covf   \{\tilde{{X}},\tilde{\varepsilon } (1) \} +  e_0 ^{-1}\covf\{\tilde{{X}},\tilde{\varepsilon } (0) \}, \quad \vxxhaj =( e_1  e_0)^{-1}\covf (\tilde{{X}}) , \\
    \vctht  &=&(\vtcht  )^{\top}= e_1 ^{-1}\covf   \{{C},\tilde{Y}(1) \}+  e_0 ^{-1}\covf   \{{C},\tilde{Y}(0) \}, \quad \vccht=( e_1   e_0 ) ^{-1} \covf ({C}).
  \end{eqnarray*}
\end{proposition}

By Proposition~\ref{prop:FCLT}, the Mahalanobis distances based on $\tauxhaj$ and $\taucht$
converge in distribution to  $\chi_K^2$. This provides a guidance for choosing $a$.  For example, we can choose $a$ as the $\alpha$th quantile of $\chi_K^2$ to ensure an asymptotic acceptance rate of $\alpha$. \citet{morgan2012rerandomization} suggested $\alpha=0.001$ when the cluster numbers are moderate or large. \LXTL{However, practitioners need to be aware that with small $M$ and $a$, we do not have enough randomizations and asymptotic approximation can be poor. We suggest the Fisher randomization tests in such situations, where the threshold $a$ is chosen to ensure non-trivial powers \citep{johansson2021optimal}.}


Proposition~\ref{prop:FCLT} provides the basis for deriving the conditional asymptotic distributions, $\tauht$ given $\mc $ and $ \tauhaj$ given $ \mx $, which are the asymptotic distributions of $\tauhaj$ and $\tauht$ under cluster rerandomization. Extending \citet{Li9157}, we introduce two squared multiple correlations: 
\begin{eqnarray*}
R_{ c}^2 & = & \cova   (\tauht  , \taucht  ) \cova    (\taucht  )^{-1}\cova   (\taucht ,\tauht   )/\vara    (\tauht   ), \\
R_{ x}^2 & = &  \cova   (\tauhaj  , \tauxhaj   )\cova    ( \tauxhaj   )^{-1}\cova   ( \tauxhaj  ,\tauhaj   )/\vara    ( \tauhaj   ). 
\end{eqnarray*}
Moreover, let  $L_{K, a} \sim D_{1} \mid {D}^{\top} {D} \leq a$ where ${D}=(D_{1}, \ldots, D_{K})^{\top}$ is a $K$-dimensional standard normal random vector, and let $\epsilon$ be a standard normal random variable independent of $L_{K,a}$. These quantities play important roles in Theorem \ref{thm:asymp2} below.

\begin{theorem}\label{thm:asymp2}
  Under regularity conditions,
  \begin{eqnarray*}
    M^{1/2}(\tauhaj  -\tau) \mid \mx &\mathrel{\ \dot{\sim}\ } & ( \vtthaj )^{1/2}\{(1-R_{ x}^{2})^{1/2}  \epsilon+R_{ x}  L_{K, a}\}, \\
    M^{1/2}(\tauht  -\tau) \mid \mc  &\mathrel{\ \dot{\sim}\ } & (\vttht )^{1/2}\{(1-R_{ c}^{2})^{1/2}  \epsilon+R_{ c} L_{K, a}\}.
  \end{eqnarray*}
\end{theorem}

Theorem~\ref{thm:asymp2} extends \citet[][Theorem 1]{Li9157} for rerandomization with treatments assigned at the individual level. We comment on the result on $\tauht$ in Theorem~\ref{thm:asymp2}, and analogous discussion also applies to $\tauhaj$.
It demonstrates that the asymptotic distribution of $\tauht$ is a linear combination of two independent random variables: a standard normal $\epsilon $ and a truncated normal $L_{k, a}$. The truncated normal random variable is due to the projection of $\tauht$ onto $\taucht $, which is affected by cluster rerandomization. The untruncated normal random variable is due to the residual from the projection, which is unaffected by cluster rerandomization. \LXTL{By Theorem~\ref{thm:asymp2}, rerandomization narrows the quantile ranges of $\tauhaj$ and $\tauht$, and with proper moment estimators of the asymptotic distributions, we can construct narrower confidence intervals for $\tau$.}


%

Theorem~\ref{thm:asymp2} also allows for a direct comparison of the two cluster rerandomization schemes. \LXTL{We construct cluster-level covariates from individual-level ones by defining ${c}_{i} = \tilde{{x}}_{i\cdot}$ for $\mc$, while we use the individual-level covariates ${x}_{ij}$ for $\mx$ to enable a fair comparison.} Corollary~\ref{coro:clusterbetter} below demonstrates the superiority of the cluster rerandomization scheme $\mc $.

\begin{corollary}\label{coro:clusterbetter}
 Under regularity conditions, if ${c}_{i}=( n_i ,\tilde{{x}}_{i \cdot }^{\top})^{\top}$, then
  %
  %
  $
    \vtthaj   (1-R_{ x}^{2} ) \geq  \vttht   (1-R_{ c}^{2}).
  $
\end{corollary}

With a small threshold $a>0$, the variance of $L_{K,a}$ is small, and therefore, $    \vtthaj   (1-R_{ x}^{2} )$ and $\vttht   (1-R_{ c}^{2})$ are the leading terms of the asymptotic variances of $\tauhaj$ and $\tauht$, respectively.  Corollary \ref{coro:clusterbetter} demonstrates that based on the asymptotic variances, $\mc $ with cluster-level covariates dominates $\mx$ with individual-level covariates if the cluster-level covariates include the cluster size and scaled cluster totals of the individual-level covariates. A key requirement of Corollary \ref{coro:clusterbetter} is that the cluster size must be used as a cluster-level covariate. This parallels the theory for regression adjustment in the analysis of cluster-randomized experiments that the regression-adjusted estimator based on scaled cluster totals dominates the regression-adjusted estimator based on individual-level data with properly defined covariates \citep{su2021modelassisted}.

\section{Cluster rerandomization with prior knowledge on covariate importance}
\label{sec:quadrageneral}

\subsection{Weighted Euclidean distance criterion}
\label{sec:4.1}

The cluster rerandomization schemes in Section \ref{sec:clre} view all covariates as equally important. Although they have the advantage of being invariant to non-degenerate linear transformation of the covariates, they are not ideal in experiments with prior knowledge about the relative importance of the covariates. In those cases, a better choice is cluster rerandomization with the weighted Euclidean distance, which has been frequently used in practice \citep{wight2002limits,althabe2008behavioral,li2016evaluation, li2017evaluation,hayes2017cluster, dempsey2018effect}. We will study its design-based properties in this subsection.

Consider cluster rerandomization schemes defined by general quadratic forms of the measures of covariate imbalance:
\[
  \mathcal{D}_{x}(A_x)=\{M\tauxhaj  ^{\top}A_x\ \tauxhaj\leq a\}, \quad \mathcal{D}_{c}(A_c) = \{M\taucht ^{\top}A_c\ \taucht \leq a\},
\]
where $A_x > 0 $ and $A_c > 0$. They reduce to cluster rerandomization schemes in Section \ref{sec:clre} if $A_x = M^{-1} \cova ( \tauxhaj )^{-1} $ and $A_c = M^{-1} \cov ( \taucht )^{-1} $, respectively. 
The weighted Euclidean distances correspond to diagonal $A_x$ and $A_c$.
Theorem~\ref{thm:asymquadra} below provides the basis for inference under $\dxa$ and $\dca$.
\begin{theorem}
  \label{thm:asymquadra}
  Under regularity conditions,
  \begin{align*}
    M^{1/2}(\tauhaj-\tau)\mid \mathcal{D}_{x}(A_x)                      & \mathrel{\ \dot{\sim}\ }  V_{\textnormal{haj},\tau \tau}^{1/2} \big\{  (1-R_{x}^2 )^{1/2}\epsilon +   R_{x}\mu_x^{\top}{\eta} \mid
    {\eta}^{\top}V_{\textnormal{haj},xx}^{1/2}A_{x}V^{1/2}_{\textnormal{haj},xx}{\eta}\leq a \big\},                                                                                                                       \\
    M^{1/2}(\hat{\tau}_{\textnormal{ht}}-\tau)\mid \mathcal{D}_{c}(A_c) & \mathrel{\ \dot{\sim}\ }  V_{\textnormal{ht},\tau \tau}^{1/2} \big\{  (1-R_{c}^2 )^{1/2}\epsilon +   R_{c}\mu_c^{\top}{\eta} \mid
    {\eta}^{\top}{V_{\textnormal{ht},cc}^{1/2}A_{c}V^{1/2}_{\textnormal{ht},cc}}{\eta}\leq a\big\},
  \end{align*}
  where ${\eta}=(\eta_{1},\ldots,\eta_{K})^{\top}$, $\epsilon,\eta_1, \ldots,\eta_K$ are independent standard normal random variables,  and
  $${\mu}_x^{\top}=(V_{\textnormal{haj},\tau x}V^{-1}_{\textnormal{haj},xx}V_{\textnormal{haj},x \tau})^{-1/2}V_{\textnormal{haj},\tau x}V_{\textnormal{haj},xx}^{-1/2},\quad {\mu}_c^{\top}=(V_{\textnormal{ht},\tau c}V^{-1}_{\textnormal{ht},cc}V_{\textnormal{ht},c \tau})^{-1/2}V_{\textnormal{ht},\tau c}V_{\textnormal{ht},cc}^{-1/2}.$$
\end{theorem}

Theorem \ref{thm:asymquadra} is similar to Theorem \ref{thm:asymp2} in that the asymptotic distributions have two components: one is normal, and the other is truncated normal. It allows us to derive properties of $\tauhaj$ and $\hat{\tau}_{\textnormal{ht}}$ under cluster rerandomization.

\begin{proposition}\label{prop:concentration-symmetric-unimodality}
  Under regularity conditions, (i) the asymptotic distributions in Theorem \ref{thm:asymquadra} are symmetric around zero and unimodal, and (ii)  $\pra \{ M^{1/2}|\tauhaj-\tau|<\delta\mid  \mathcal{D}_{x}(A_x) \}$ is a non-decreasing function of $R^2_{x}$ and $\pra  \{M^{1/2}| \hat{\tau}_{\textnormal{ht}}-\tau | <\delta  \mid  \mathcal{D}_{c}(A_c) \}$ is a non-decreasing function of ${R}_{c}^2$ for any fixed $\delta>0$.
\end{proposition}

Proposition \ref{prop:concentration-symmetric-unimodality} (i) ensures that the asymptotic distributions in Theorem \ref{thm:asymquadra} are both bell-shaped, although they differ from normal distributions.
Proposition \ref{prop:concentration-symmetric-unimodality} (ii) ensures that
the asymptotic distributions in Theorem \ref{thm:asymquadra}  are  more concentrated at zero than those under standard cluster randomization.


To compare the asymptotic efficiency of the two cluster rerandomization schemes, we can compare their variance reductions given the same acceptance rate. Let $\alpha$ denote the common asymptotic acceptance rate:
$$
  \alpha = \textup{pr}_\textup{a} \{ \mathcal{D}_{x}(A_x) \}
  =  \textup{pr}_\textup{a}  (M\hat{\tau}_{\textnormal{haj},x}^{\top}A_x\hat{\tau}_{\textnormal{haj},x}\leq a ),\quad
  \alpha = \textup{pr}_\textup{a} \{ \mathcal{D}_{c}(A_c) \}
  =  \textup{pr}_\textup{a}  (M\taucht^{\top}A_c\taucht\leq a ) .
$$
Let $\Gamma (\cdot)$ be the Gamma function and 
\begin{equation} \label{eqn::pk}
    p_{K}=\frac{2\pi}{K+2}\left\{ \frac{2\pi^{K/2}}{K\Gamma (K/2)} \right\}^{-2/K}. \nonumber
\end{equation}
The expansions in Theorem~\ref{thm:varalpha} below provide the basis for comparing asymptotic efficiency.


\begin{theorem}\label{thm:varalpha}
  Under regularity conditions, 
  \begin{gather*}
    \vara\{M^{1/2}(\tauhaj-\tau)\mid  \mathcal{D}_{x}(A_x) \} = {V}_{\textnormal{haj},\tau\tau} \{(1-R^2_{x})+{R}^2_{x}p_K {\nu}_x (A_x)\alpha^{2/K} +o(\alpha^{2/K})\},\\
    \vara\{M^{1/2}(\hat{\tau}_{\textnormal{ht}}-\tau)\mid  \mathcal{D}_{c}(A_c) \} = {V}_{\textnormal{ht},\tau\tau} \{(1-R^2_{c})+{R}^2_{c}p_K {\nu}_c (A_c)\alpha^{2/K} +o(\alpha^{2/K})\},
  \end{gather*}
  for a small $\alpha$,  where
  \begin{eqnarray*}
    \nu_x (A_x) &=& \frac{V_{\textnormal{haj},\tau x}V_{\textnormal{haj},xx}^{-1}A_x^{-1}V_{\textnormal{haj},xx}^{-1}V_{\textnormal{haj},x\tau} \det(A_x)^{1/K} \det(V_{\textnormal{haj},xx}) ^{1/K}}{V_{\textnormal{haj},\tau x}V_{\textnormal{haj},xx}^{-1}V_{\textnormal{haj},x\tau}},\\
    \nu_c (A_c) &=& \frac{V_{\textnormal{ht},\tau c}V_{\textnormal{ht},cc}^{-1}A_c^{-1}V_{\textnormal{ht},cc}^{-1}V_{\textnormal{ht},c\tau} \det(A_c)^{1/K} \det(V_{\textnormal{ht},cc}) ^{1/K}}{V_{\textnormal{ht},\tau c}V_{\textnormal{ht},cc}^{-1}V_{\textnormal{ht},c\tau}}.
  \end{eqnarray*}
\end{theorem}

As a sanity check, the cluster rerandomization schemes using the Mahalanobis distances lead to ${\nu}_x(A_x)=\nu_x (V_{\textnormal{haj},xx}^{-1})=1$ and ${\nu}_c(A_c)=\nu_c (V_{\textnormal{ht},cc}^{-1})=1$. Given the same $R^2_x$ or $R^2_c$ in Theorem \ref{thm:varalpha}, the asymptotic efficiency of using  $\mathcal{D}_{x}(A_x)$ or $\mathcal{D}_{c}(A_c)$ depends only on $\nu_x (A_x)$ or $\nu_c (A_c)$, respectively. Cluster rerandomization schemes using the weighted Euclidean distances have diagonal $A_x$ and $A_c$, which allows us to derive the optimal weighting matrices as shown in Theorem~\ref{thm:optimdiag} below.   Let $\xi_k $ denote a $K$-dimensional vector with $1$ at the $k$th dimension and $0$ at other dimensions.

\begin{theorem}
  \label{thm:optimdiag}
  Under regularity conditions, if ${V}_{\textnormal{haj},\tau x}{V}_{\textnormal{haj},xx}^{-1}\xi_k $ and ${V}_{\textnormal{ht},\tau c}{V}_{\textnormal{ht},cc}^{-1}\xi_k $ are nonzero for all $k =1,\ldots,K$, then $ \nu_x \{{\operatorname{diag}}(w_1,\ldots,w_K)\}$ reaches minimum if $w_k \propto (V_{\textnormal{haj},\tau x}V_{\textnormal{haj},xx}^{-1}\xi_k )^2$ for $k =1,\ldots,K$, and $ \nu_c \{{\operatorname{diag}}(w_1,\ldots,w_K)\}$ reaches minimum if $w_k \propto (V_{\textnormal{ht},\tau c}V_{\textnormal{ht},cc}^{-1}\xi_k )^2$ for $k =1,\ldots,K$.
\end{theorem}

Based on Theorem \ref{thm:optimdiag}, let $A^{\textnormal{opt}}_{x}$ and $A^{\textnormal{opt}}_{c}$ denote the optimal weighting matrices at the individual level and cluster level, respectively. Recall Proposition~\ref{prop:FCLT} to obtain that
\begin{gather*}
    \vtxhaj \vxxhaj^{-1}  = \covf   \{ e_1  \tilde{\varepsilon} (0)+ e_0  \tilde{\varepsilon} (1),\tilde{{X}} \}\{\covf (\tilde{{X}})\}^{-1}, \\
    \vtcht \vccht ^{-1}  = \covf   \{ e_1  \tilde{Y} (0)+ e_0  \tilde{Y} (1),{C} \}\{\covf ({C})\}^{-1} .
\end{gather*}
Therefore, the optimal weights in $A^{\textnormal{opt}}_{x}$ and $A^{\textnormal{opt}}_{c}$ correspond to the squared coefficients obtained from regressing $ e_1 \tde(0)+ e_0 \tde(1)$ on $\tdx$ and regressing $ e_1 \tilde{Y} (0)+ e_0 \tilde{Y} (1)$ on $\tilde{{C}}$, respectively.

\begin{remark}\label{remark::weighted-maha}
  Orthogonalizing the covariates ensures that $V_{\textnormal{haj},xx}^{-1}$ and $V_{\textnormal{ht},cc}^{-1}$ are both diagonal. Then the optimality of $A_x^{\textnormal{opt}}$ or $A_c^{\textnormal{opt}}$ implies $\nu_x (A_x^{\textnormal{opt}}) \leq \nu_x (V_{\textnormal{haj},xx}^{-1})=1$ or $\nu_c (A_c^{\textnormal{opt}}) \leq \nu_c (V_{\textnormal{ht},cc}^{-1})=1$, respectively. Therefore, with orthogonalized covariates, the cluster rerandomization schemes based on the optimal weighted Euclidean distances are superior to those based on the Mahalanobis distances. However, this conclusion does not hold without orthogonalizing the covariates. We give a counterexample in the appendix.
\end{remark}


\subsection{Comparison with cluster rerandomization with tiers of covariates}
\label{sec:comparison-tier}

With prior knowledge on the importance of the covariates, \citet{morgankari2015} proposed rerandomization with tiers of covariates as an alternative to rerandomization with the weighted Euclidean distance. \citet{morgankari2015} and \citet{Li9157} derived its properties under the assumption that treatments are assigned at the individual level. The theory for the cluster analog is still missing. Moreover, no comparison has been made between these two rerandomization schemes. We fill these gaps in this subsection.

Similar to Remark \ref{remark::weighted-maha}, the comparison does not yield definite answers with general covariates. However, we can demonstrate that  rerandomization  with the optimal weighted Euclidean distance is superior to  rerandomization with tiers of covariates when the covariates are orthogonalized. In this subsection, we focus on cluster rerandomization with individual-level covariates; in the appendix, we give the corresponding results for cluster rerandomization with cluster-level covariates. 

Orthogonalization is applied on the scaled cluster totals. The procedure proceeds as follows: first, find the upper triangular matrix $U$ by applying Gram--Schmidt orthogonalization to $\tdxi$ such that $U\tdxi$ has orthogonal components, and then transform the individual-level covariates to $Ux_{ij}$. To simplify the presentation, we assume that $x_{ij} = (x_{ij1}, \ldots, x_{ijK})^\top$ are already orthogonalized by the aforementioned procedure and are ordered in decreasing importance.


We first re-stated the results for cluster rerandomization with orthogonalized covariates. Let
$$
  V_{\textnormal{haj},xx} = \operatorname{diag}(V_{\textnormal{haj},x_1x_1},\ldots,V_{\textnormal{haj},x_Kx_K}),
  \quad
  V_{\textnormal{haj},\tau x} = (V_{\textnormal{haj},\tau x_{1}}, \ldots, V_{\textnormal{haj},\tau x_K}).
$$
Define the squared multiple correlation for the $k$th covariate as
$$
  R_{x_k}^2 =  V_{\textnormal{haj},\tau x_k }^2/  ( V_{\textnormal{haj},\tau\tau}V_{\textnormal{haj},x_k  x_k }), \quad k = 1,\ldots,K.
$$
Then by Theorem~\ref{thm:optimdiag}, the cluster rerandomization scheme based on the optimal weighted Euclidean distance is
$
  \mathcal{D}_x(A_{x}^{\textnormal{opt}})= \{M{\hat{\tau}_{\textnormal{haj}, x}}^\top A_{x}^{\textnormal{opt}} {\hat{\tau}_{\textnormal{haj}, x}} \leq a\}
$
with the optimal diagonal matrix
$$
  A_{x}^{\textnormal{opt}} = \operatorname{diag}\{(V_{\textnormal{haj},\tau x_{1}}V_{\textnormal{haj},x_{1}x_{1}}^{-1})^2,\ldots,(V_{\textnormal{haj},\tau x_{k}}V_{\textnormal{haj},x_{K}x_{K}}^{-1})^2\} .
$$
Corollary  \ref{cor:varalphaortho} below is a special case of Theorem~\ref{thm:varalpha}.

\begin{corollary}
  \label{cor:varalphaortho}
  Under regularity conditions with orthogonalized covariates and  optimal weighted Euclidean distance, if ${V}_{\textnormal{haj},\tau x_k }{V}_{\textnormal{haj},x_k x_k }^{-1}$ are nonzero for all $k =1,\ldots,K$, then
  $$
    \vara    \{M^{1/2}(\tauhaj-\tau) \mid \mathcal{D}_{x}(A_{x}^{\textnormal{opt}}) \} =  V_{\textnormal{haj},\tau \tau} \Big\{(1- R_{x}^2)+ K \Big(\prod_{k =1}^K R_{x_{k}}^2\Big)^{1/K}p_K \alpha^{2/K} +o(\alpha^{2/K}) \Big\},
  $$
  for a small $\alpha$. 
\end{corollary}

By Theorem \ref{thm:varalpha}, we can obtain that
$$
  {\vara}   \{M^{1/2} (\tauhaj-\tau) \mid \mathcal{M}_x  \} = V_{\textnormal{haj},\tau \tau} \{(1- R_{x}^2)+ R_{x}^2 p_{K}\alpha^{2/K}+o(\alpha^{2/K})  \}.
$$
Comparing it with Corollary \ref{cor:varalphaortho} above, we can  see that cluster rerandomization with the optimal weighted Euclidean distance is indeed superior to cluster rerandomization with the Mahalanobis distance due to $K \big(\prod_{k=1}^K R_{x_k}^2\big)^{1/K} \leq  R_{x}^2$ by the inequality of arithmetic and geometric means, which echoes Remark \ref{remark::weighted-maha}.

In rerandomization with tiers of covariates,  \citet{morgankari2015} partitioned the covariates into  $L$ tiers by their importance with $K_l$ covariates in tier $l$, for $l=1,\ldots,L$ and $\sum_{l=1}^L K_l = K$. Let ${x}_{[l]}$ be the covariates in the $l$th tier. 
Then ${x} =({x}_{[1]},\ldots,{x}_{[L]})$. \cite{morgankari2015} originally used the block-wise Gram--Schmidt orthogonalization, but we use the above element-wise orthogonalization because it is notationally simpler and  leads to the same asymptotic distribution. Within tier $l$, define $\hat{\tau}_{\textnormal{haj},x_{[l]}}$ as the difference in means of covariates, $V_{\textnormal{haj},x_{[l]}x_{[l]}}$ as its covariance, $\mathcal{M}_{[l]}$ as the event that the corresponding Mahalanobis distance is smaller than or equal to the threshold $a_{[l]}$, $R_{x_{[l]}}^2$ as the corresponding $R^2$, and
$\alpha_{[l]}  $ as the asymptotic acceptance rate. The overall asymptotic acceptance rate satisfies $\alpha = \prod_{l=1}^L \alpha_{[l]}$ due to the orthogonalization of the covariates. 

\begin{corollary}
  \label{cor:varalphatier}
  Under regularity conditions,
  $$
    \vara  \{M^{1/2} (\tauhaj-\tau)\mid \mathcal{M}_{[1]},\ldots,\mathcal{M}_{[L]}  \} =   V_{\textnormal{haj},\tau \tau} \Big\{(1- R_{x}^2)+ \sum_{l=1}^L R_{x_{[l]}}^2 p_{K_{l}}\alpha_{[l]}^{2/K_{l}}  +o(\alpha_{[l]}^{2/K_l}) \Big\},
  $$
  for small $\alpha_{[l]}$, $l=1,\ldots,L$.
\end{corollary}

Given $\alpha = \prod_{l=1}^L \alpha_{[l]}>0$,  to minimize the second term in the asymptotic variance in Corollary \ref{cor:varalphatier}, we must choose
$$
  \alpha_{[l]}=\left({c_0 {R}_{x_{[l]}}^2 p_{K_{l}}}/{K_{l}}\right)^{-K_{l}/2} \quad (l=1,\ldots,L),
$$
for some positive constant $c_0$. Theorem \ref{thm:wtedbetter} below compares the second terms in the asymptotic variances in Corollaries \ref{cor:varalphaortho} and \ref{cor:varalphatier}.

\begin{theorem}
  \label{thm:wtedbetter}
  Under regularity conditions with orthogonalized covariates, the following inequality holds for any $\alpha_{[1]},\ldots,\alpha_{[L]}$ satisfying $\alpha = \prod_{l=1}^L \alpha_{[l]}$:
  \begin{align*}
    \sum_{l=1}^L {R}_{x_{[l]}}^2 p_{K_{l}}\alpha_{[l]}^{2/K_l} \geq  K\Big(\prod_{k=1}^K R_{x_{k}}^2\Big)^{1/K}p_K \alpha^{2/K}.
  \end{align*}
\end{theorem}

Based on the comparison of the asymptotic variances,  Theorem \ref{thm:wtedbetter} quantifies  the superiority of cluster rerandomization with the optimal weighted Euclidean distance when the covariates are orthogonalized.



\section{Rerandomization and regression adjustment}\label{sec:readj}


Rerandomization uses covariates in the design stage \citep{morgan2012rerandomization}, and regression adjustment uses covariates in the analysis stage \citep{Lin2013Agnostic}. \citet{Li9157} pointed out that they are dual in the design and analysis of randomized experiments with treatments assigned at the individual level. Moreover, \citet{2020Rerandomization} showed that they could be used simultaneously. In this section, we will show that  analogous results hold under cluster rerandomization and point out some subtle differences. 

Under $\mathcal{D}_x(A_x)$, we can use the coefficient of $Z_{ij}  $ in the least squares fit of $Y_{ij}$ on $(1,Z_{ij}, x_{ij}, Z_{ij} x_{ij})$ to estimate $\tau$, and use the cluster-robust standard error to approximate the true asymptotic standard error \citep{Liang1986}; under $\mathcal{D}_c(A_c)$, we can use the coefficient of $Z_i$ in the least squares fit of $\tilde Y_{i\cdot}$ on $(1,Z_i, c_i, Z_i c_i)$ to estimate $\tau$, and use the heteroskedasticity-robust standard error to approximate the true asymptotic standard error \citep{Huber1967,White1980}. Denote the two sets of regression coefficient and variance estimator by $(\tauadjhaj,  \vhatadjlz)$ and $(\tauadjht, \vhatadjhw )$, respectively. Let $z_{\varsigma}$ ($0 < \varsigma < 1$) be the $\varsigma$th quantile of a standard normal distribution. Theorem~\ref{thm::regression-adjustment-under--cluster-rerandomization} below summarizes the important results about cluster rerandomization combined with regression adjustment. 


\begin{theorem}
  \label{thm::regression-adjustment-under--cluster-rerandomization}
  Assume regularity conditions hold.
  \begin{itemize}
    \item[(i)] Under  $\mathcal{D}_c(A_c)$, the estimator $\tauadjht$ is consistent for $\tau$ and asymptotically normal, the probability limit of $M \vhatadjhw$ is larger than or equal to the true asymptotic variance of $ M^{1/2} \tauadjht$, and the $1-\varsigma$ confidence interval 
      \[
    \left[\tauadjht+ (\vhatadjhw)^{1/2}z_{\varsigma/2}, \tauadjht+(\vhatadjhw)^{1/2}z_{1-\varsigma/2}\right]
  \]
  has asymptotic coverage rate $\geq 1-\varsigma$;
    \item[(ii)] Under $\mathcal{D}_x(A_x)$, the estimator $\tauadjhaj$ is consistent for $\tau$ and its asymptotic distribution is a convolution of normal and truncated normal, the probability limit of $M \vhatadjlz$ is larger than or equal to the true asymptotic variance of $ M^{1/2}  \tauadjhaj$, and the $1-\varsigma$ confidence interval 
      \[
    \left[\tauadjhaj+(\vhatadjlz)^{1/2}z_{\varsigma/2}, \tauadjhaj+(\vhatadjlz)^{1/2}z_{1-\varsigma/2}\right]
  \]
  has asymptotic coverage rate $\geq 1-\varsigma$;
    \item[(iii)]  If  ${c}_{i}=( n_i ,\tilde{{x}}_{i \cdot }^{\top})^{\top}$, the asymptotic distribution of $\tauadjht \mid\mathcal{D}_c(A_c)$ is more concentrated at $\tau$ than $\tauadjhaj\mid\mathcal{D}_x(A_x)$, in the sense that for any $\delta>0$, we have
      \begin{align*}
        \pra\{M^{1/2}|\tauadjhaj-\tau|<\delta\mid  \mathcal{D}_{x}(A_x) \}\leq \pra\{M^{1/2}| \tauadjht-\tau | <\delta  \mid  \mathcal{D}_{c}(A_c) \}.
      \end{align*}
  \end{itemize}

\end{theorem}

To avoid notation complexity, we relegate the details of the asymptotic distributions of $\tauadjht$ given $\mathcal{D}_c(A_c)$ and $\tauadjhaj$ given $\mathcal{D}_x(A_x)$ to Theorem~\ref{thm:tauhtadjasymp} in the appendix. Theorem~\ref{thm::regression-adjustment-under--cluster-rerandomization} parallels the key results in \citet{su2021modelassisted} without rerandomization. It justifies the standard Wald-type inference based on appropriate regressions and standard errors, which can be conveniently obtained by standard statistical software packages. It also extends the results in \cite{2020Rerandomization} from classic rerandomization to cluster rerandomization. However, some asymptotic results differ from those under rerandomization with treatments assigned at the individual level. In particular, with cluster rerandomization and regression adjustment using individual-level covariates, the asymptotic distribution is not normal anymore; see Theorem~\ref{thm:tauhtadjasymp} in the appendix for more details. Moreover, Theorem~\ref{thm::regression-adjustment-under--cluster-rerandomization} (iii) extends the theory in \cite{su2021modelassisted} that regression adjustment using cluster size and scaled cluster totals outperforms regression adjustment using individual-level covariates.

\begin{remark}\label{remark::discrepancy}
  \citet{2020Rerandomization} also discussed the scenarios that the designer and analyzer do not communicate well, allowing the covariates used in the analysis stage be different from those used in the design stage.  One important special case is that the set of covariates used in the analysis stage is empty. To be specific, under $\mathcal{D}_x(A_x)$, we can use the coefficient of $Z_{ij}  $ in the least squares fit of $Y_{ij}$ on $(1,Z_{ij})$ to estimate $\tau$, and use the cluster-robust standard error to approximate the true asymptotic standard error; under $\mathcal{D}_c(A_c)$, we can use the coefficient of $Z_i$ in the least squares fit of $\tilde Y_{i\cdot}$ on $(1,Z_i)$ to estimate $\tau$, and use the heteroskedasticity-robust standard error to approximate the true asymptotic standard error. These point estimators correspond to the Hajek and Horvitz--Thompson estimators discussed in Section \ref{sec:crt}, with standard errors $(\vhatadjlz)^{1/2}$ and $(\vhatadjhw)^{1/2}$, respectively. Analogous to Theorem \ref{thm::regression-adjustment-under--cluster-rerandomization}, these point estimators are still consistent for $\tau$ and converge in distribution to convolutions of normal and truncated normal under cluster rerandomization. Moreover, the normal-based confidence intervals are still asymptotically conservative. However, different from the results in \citet{su2021modelassisted}, the standard errors $(\vhatadjlz)^{1/2}$ and $(\vhatadjhw)^{1/2}$ are overly conservative because they ignore the truncated normal components in the asymptotic distributions.  Due to the technical complexities, we relegate the theory and improved inference methods to the appendix.
\end{remark}

\begin{remark}
\LXTL{Our asymptotic theory requires a large number of clusters, which may be unrealistic in many cluster randomized experiments. When $M$ is small, one alternative approach is to use a mixed-effects model by imposing modelling assumptions on the data generating process. Another alternative approach is to use Fisher randomization tests that deliver finite-sample exact $p$-values under the sharp null hypothesis.  See \citet[Section 4]{Zhao2021} for more details. Importantly, the Fisher randomization tests must follow the same treatment assignment rule of cluster rerandomization.}
\end{remark}




\section{Numeric Examples}\label{sec:sim}

\subsection{Simulation}
In this section, we conduct simulation to assess the finite-sample performances of all eight combinations with binary choices for three orthogonal axes, individual-level (\texttt{X}) versus cluster-level (\texttt{C}), the Mahalanobis distance (\texttt{M}) versus the optimal weighted Euclidean distance without orthogonalization (\texttt{W}), and using  regression adjustment (\texttt{.adj}) or not. We also consider two baseline methods: Hajek (\texttt{Haj}) and Horvitz--Thompson (\texttt{HT}) estimators without using cluster rerandomization for better comparison. The names of the methods are written as \texttt{ReMC}, \texttt{ReWC}, \texttt{ReMI}, \texttt{ReWI}, \texttt{ReMC.adj}, \texttt{ReWC.adj}, \texttt{ReMI.adj}, and \texttt{ReWI.adj}, where the first two letters refer to cluster rerandomization and the rest correspond to the aforementioned three axes.

We generate potential outcomes from the following model:
\begin{align*}
  Y_{ij}(z) = g (n_i) + {x}^\top_{ij}\beta_{iz} + \varepsilon _{ij}(z) \quad (i = 1,\ldots,M,\ j = 1,\ldots,n_i,\ z=0,1).
\end{align*}
Here $g (n_i)$ captures the cluster effects on individuals and $\varepsilon _{ij}(z)$'s are independent $ \mathcal{N}(0,16)$. Pre-treatment covariates $x_{ij}$ are generated from a $K$-dimensional normal distribution with mean $ 0 $ and covariance matrix $\Sigma$, where $\Sigma_{ij} = (1-\rho)\delta_{ij} + \rho$, with $\delta_{ij} = 1$ if $i=j$ and $0$ otherwise. We fix $M=100$, with $50$ assigned to the treatment arm and $50$  to the control arm.  We also conduct simulations with relatively small $M$, and assess the validity of such asymptotic inference. The results are shown in the appendix. The size of each cluster is sampled uniformly from $\{m \in \mathbb{N} \mid 4 \leq m \leq 10\}$. The coefficients are generated from $\beta_{iz} = \beta_{z} + U(-0.1, 0.1)$, where $U(\cdot,\cdot)$ denotes uniform distribution. Here each component of $\beta_{1}$ is sampled from $\{0.5\gamma,\gamma,1.5\gamma\}$ with equal probabilities, and $\beta_0 = 2\gamma{1}_{K}-\beta_1$, where $1_{K}$ is a $K$-dimensional all-one vector, and $\gamma$ is chosen to ensure  the proportion of variance contributed by the covariates to be $0.5$. These $K$ covariates are used to generate the potential outcomes, conduct cluster rerandomization, and perform regression adjustment.
We use the scaled cluster totals of individual-level covariates together with cluster size as cluster-level covariates. 
Each cluster rerandomization is conducted with an acceptance rate $\alpha = 0.1\%$. We consider four scenarios with parameters given in Table~\ref{tab:params}.

\begin{table}
  \centering
  \caption{Parameters of four scenarios}\label{tab:params}
  \begin{tabular}{ccccc}
    \toprule
    Scenario & $K$  & $\rho$  & $\gamma$ & $g (n_i)$     \\
    \midrule
    1        & $7$  & $0.8$   & $1$      & $( n_i -7)/2$ \\
    2        & $7$  & $-0.15$ & $5$      & $( n_i -7)/2$ \\
    3        & $12$ & $0.4$   & $0.5$    & $6$           \\
    4        & $12$ & $-0.09$  & $12$      & $6$           \\
    \bottomrule
  \end{tabular}
\end{table}

Once generated, the potential outcomes and covariates are fixed throughout the simulation. We repeat the simulation for each scenario, draw treatment assignments $1000$ times, and report the biases, standard deviations (SD), root mean squared errors (RMSE),  and empirical coverage probabilities (CP) of 95\% confidence intervals, and mean lengths of such intervals (CI Length). \LXTL{In particular, coverage probability is the percentage of confidence intervals that cover the true value of the average treatment effect from the $1000$ replications.} We consider two methods for constructing confidence intervals: the normal-based method given in Theorem~\ref{thm::regression-adjustment-under--cluster-rerandomization} and Remark~\ref{remark::discrepancy} and improved method provided in the appendix.

\begin{figure}
  \centering
  \includegraphics[width=\textwidth]{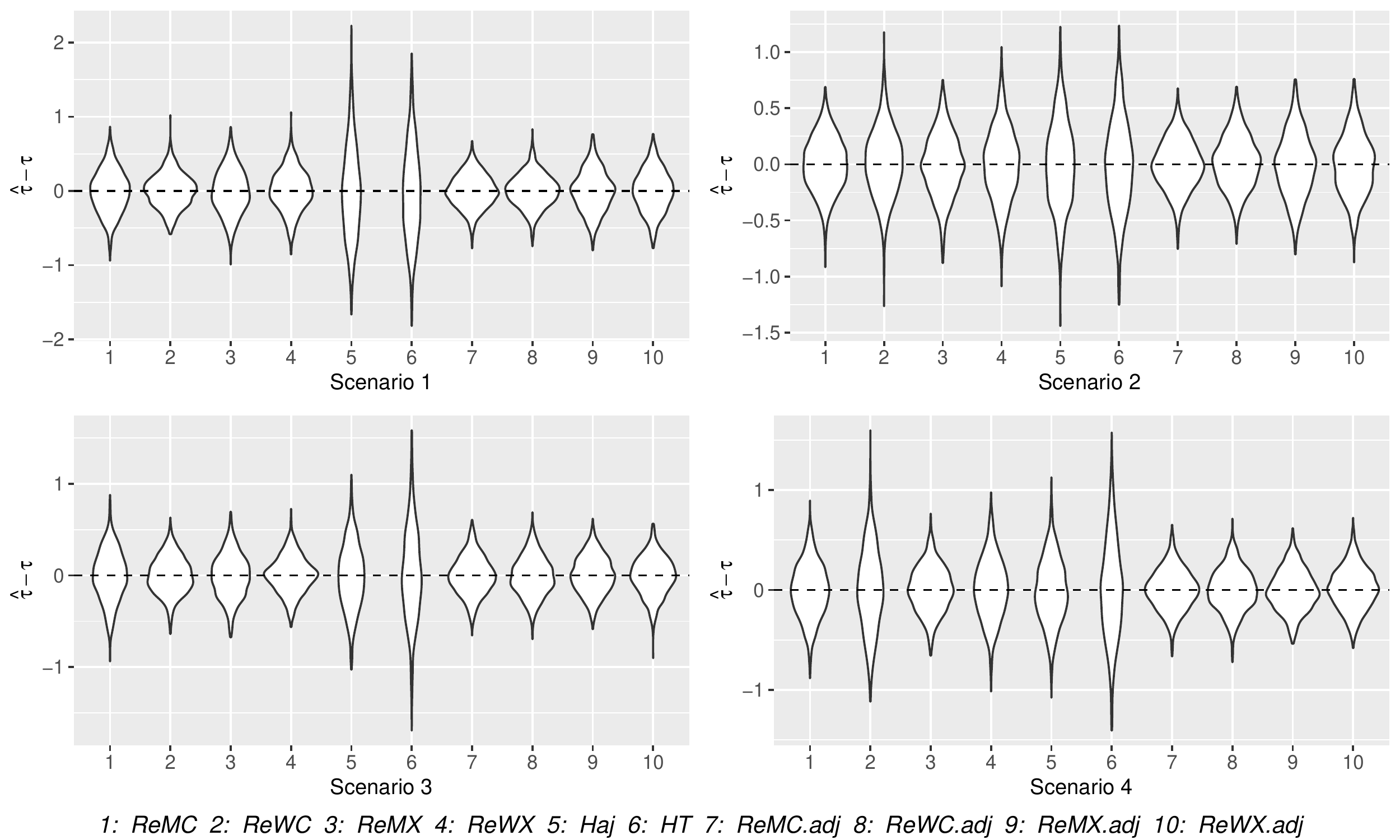}
  \caption{Comparison results in simulation.}  \label{fig:sim}
\end{figure}


\begin{table}[htbp]
  \centering
  \caption{Comparison results in  simulation}\label{tab:sim}
  \resizebox{0.6\textwidth}{!}{
    \begin{tabular}{cccccccc}

      \toprule

\multirow{2}{*}{Method} & \multirow{2}{*}{Bias} & \multirow{2}{*}{SD} & \multirow{2}{*}{RMSE} & \multicolumn{2}{c}{Normal-based}   & \multicolumn{2}{c}{Improved}                          \\ 
& & &  & CP           & CI Length    & CP       & CI Length \\ \toprule\addlinespace[0.5em]
                  \multicolumn{8}{c}{Scenario 1}                           \\ 
                  
    \midrule
      ReMC   & -0.01 & 0.29 & 0.29 & 1.00 & 2.53 & 0.97 & 1.24 \\ 
      ReWC     & 0.00 & 0.23 & 0.23 & 1.00 & 2.53 & 0.98 & 1.07 \\
      ReMX    & -0.01 & 0.30 & 0.30 & 1.00 & 2.53 & 0.97 & 1.28 \\ 
      ReWX      & 0.01 & 0.27 & 0.27 & 1.00 & 2.54 & 0.96 & 1.20 \\ 
      Haj     & -0.02 & 0.62 & 0.62 & 0.96 & 2.53 & -- & -- \\ 
      HT       & -0.03 & 0.62 & 0.62 & 0.95 & 2.53 & -- & -- \\ 
      ReMC.adj & 0.00 & 0.22 & 0.22 & 0.98 & 1.04 & -- & -- \\ 
      ReWC.adj & 0.00 & 0.22 & 0.22 & 0.97 & 1.05 & -- & -- \\ 
      ReMX.adj & 0.00 & 0.27 & 0.27 & 0.98 & 1.29 & 0.96 & 1.19 \\ 
      ReWX.adj & 0.00 & 0.28 & 0.28 & 0.98 & 1.29 & 0.97 & 1.22 \\ 
      
       \toprule\addlinespace[0.5em]
                  \multicolumn{8}{c}{Scenario 2}                           \\ 
        \midrule          
      ReMC     & 0.00 & 0.24 & 0.24 & 1.00 & 2.18 & 0.98 & 1.14 \\ 
      ReWC     & 0.01 & 0.30 & 0.30 & 1.00 & 2.17 & 0.96 & 1.23 \\ 
      ReMX     & -0.01 & 0.27 & 0.27 & 1.00 & 2.17 & 0.97 & 1.21 \\ 
      ReWX     & 0.00 & 0.32 & 0.32 & 1.00 & 2.18 & 0.95 & 1.31 \\ 
      Haj      & -0.01 & 0.38 & 0.38 & 1.00 & 2.17 & -- & -- \\ 
      HT      & -0.01 & 0.41 & 0.41 & 0.99 & 2.16 & -- & -- \\ 
      ReMC.adj & 0.00 & 0.22 & 0.22 & 0.98 & 1.04 &--  & -- \\ 
      ReWC.adj & 0.02 & 0.23 & 0.23 & 0.97 & 1.05 & -- &  --\\ 
      ReMX.adj & -0.01 & 0.28 & 0.28 & 0.98 & 1.29 & 0.96 & 1.20 \\ 
      ReWX.adj & 0.00 & 0.27 & 0.27 & 0.99 & 1.30 & 0.97 & 1.20 \\ 
      
       \toprule\addlinespace[0.5em]
                  \multicolumn{8}{c}{Scenario 3}                           \\ 
         \midrule         
      ReMC     & 0.00 & 0.27 & 0.27 & 1.00 & 2.12 & 0.97 & 1.19 \\ 
      ReWC     & 0.00 & 0.20 & 0.20 & 1.00 & 2.12 & 0.99 & 1.03 \\ 
      ReMX     & -0.01 & 0.24 & 0.24 & 1.00 & 1.58 & 0.97 & 1.07 \\ 
      ReWX     & 0.01 & 0.20 & 0.20 & 1.00 & 1.58 & 0.99 & 1.02 \\ 
      Haj      & -0.01 & 0.35 & 0.35 & 0.97 & 1.58 & -- & -- \\ 
      HT       & 0.01 & 0.50 & 0.50 & 0.95 & 2.10 & -- & --  \\ 
      ReMC.adj & 0.00 & 0.21 & 0.21 & 0.98 & 0.99 & -- & -- \\
      ReWC.adj & -0.01 & 0.21 & 0.21 & 0.98 & 1.01 & -- & -- \\ 
      ReMX.adj & 0.00 & 0.20 & 0.20 & 1.00 & 1.16 & 0.99 & 1.00 \\ 
      ReWX.adj & -0.01 & 0.21 & 0.21 & 1.00 & 1.16 & 0.98 & 1.01 \\

       \toprule\addlinespace[0.5em]
                  \multicolumn{8}{c}{Scenario 4}                           \\ 
      \midrule
      ReMC     & -0.01 & 0.27 & 0.27 & 1.00 & 4.50 & 0.98 & 1.34 \\ 
      ReWC     & 0.01 & 0.41 & 0.41 & 1.00 & 4.49 & 0.97 & 1.76 \\ 
      ReMX     & 0.01 & 0.23 & 0.23 & 1.00 & 4.29 & 0.99 & 1.20 \\ 
      ReWX     & 0.01 & 0.31 & 0.31 & 1.00 & 4.30 & 0.97 & 1.49 \\ 
      Haj      & -0.01 & 0.32 & 0.32 & 1.00 & 4.27 & -- & -- \\
      HT       & 0.01 & 0.49 & 0.49 & 1.00 & 4.48 & -- & -- \\ 
      ReMC.adj & 0.00 & 0.20 & 0.20 & 0.98 & 0.99 & -- & -- \\ 
      ReWC.adj & 0.00 & 0.21 & 0.21 & 0.98 & 0.99 & -- & -- \\
      ReMX.adj & 0.00 & 0.20 & 0.20 & 1.00 & 1.16 & 0.99 & 1.01 \\ 
      ReWX.adj & 0.01 & 0.21 & 0.21 & 1.00 & 1.16 & 0.98 & 1.00 \\ 
      
      \bottomrule
    \end{tabular}
    }
\end{table}

As shown in Table~\ref{tab:sim} and Figure~\ref{fig:sim}, the combination of cluster rerandomization and regression adjustment results in better performance for all four scenarios. However, in some occasions, the combination has small improvement over cluster rerandomization only. This is because using cluster rerandomization only already has significant improvement, as the standard deviations are 30\%--50\% smaller than the two baseline methods.
Moreover,  the biases are negligible for all methods.

In Scenarios 1 and 3, cluster rerandomization schemes based on the optimal weighted Euclidean distances outperform cluster rerandomization schemes based on the Mahalanobis distances. This is due to the positively correlated covariates. On the contrary, in Scenarios 2 and 4, cluster rerandomization schemes based on the Mahalanobis distances are better because the covariates are negatively correlated. We conclude that the correlation structure between covariates can severely influence the efficiency of the balancing criterion.

Moreover, cluster rerandomization with cluster-level covariates outperforms cluster rerandomization with individual-level covariates in Scenarios 1 and 2, whereas cluster rerandomization with  individual-level covariates is slightly better in Scenarios 3 and 4. To illustrate this, if the potential outcomes are not influenced by cluster size like in Scenarios  3 and 4, especially if the potential outcomes have a small variation across units, cluster rerandomization with individual-level covariates could be better. This is because they are based on the Hajek estimator, which partially eliminates the effect of cluster size by using $\tilde{{\varepsilon}}$ instead of $Y$ and fits a more precise model for this particular situation. On the other hand, we recommend cluster rerandomization with cluster-level covariates, especially when the potential outcomes depend heavily on cluster size, because in such settings, the effect of cluster size is not directly accounted for if we use cluster rerandomization with individual-level covariates. Finally, by Corollary~\ref{coro:clusterbetter}, with a sufficient small threshold, cluster rerandomization with cluster-level covariates is more efficient than cluster rerandomization with individual-level covariates in terms of variance reduction.

\subsection{Child growth monitoring data}\label{sec:real}
\cite{Tembo2017Home} conducted a cluster-randomized experiment in developing countries to study the effect of community-based monitoring on child growth faltering. We use part of this dataset to illustrate our theoretical results.  We use two arms of the study, community-based monitoring as treatment and no monitoring as control. In the experiment, $41$ communities were randomly assigned to the treatment arm while $41$ to the control arm. The mean, median, and maximum of cluster size are $3.085$, $3$, and $8$, respectively. We choose $6$ individual-level covariates, including the child age, family background information, and some baseline growth data. We use cluster size and scaled cluster totals of individual-level covariates as cluster-level covariates.  The height of a child is used as the outcome. As the potential outcomes are not completely available, we fit two simple linear regression models to impute the missing potential outcomes, and perform simulation based on the imputed dataset. The overall acceptance rate of all cluster rerandomization schemes are approximately $0.1\%$, computed from the quantile of the empirical distances. We use $4$ covariates at the design stage for cluster rerandomization and $6$ covariates at the analysis stage for regression adjustment.

\begin{figure}
  \centering
  \includegraphics[width=.9\textwidth]{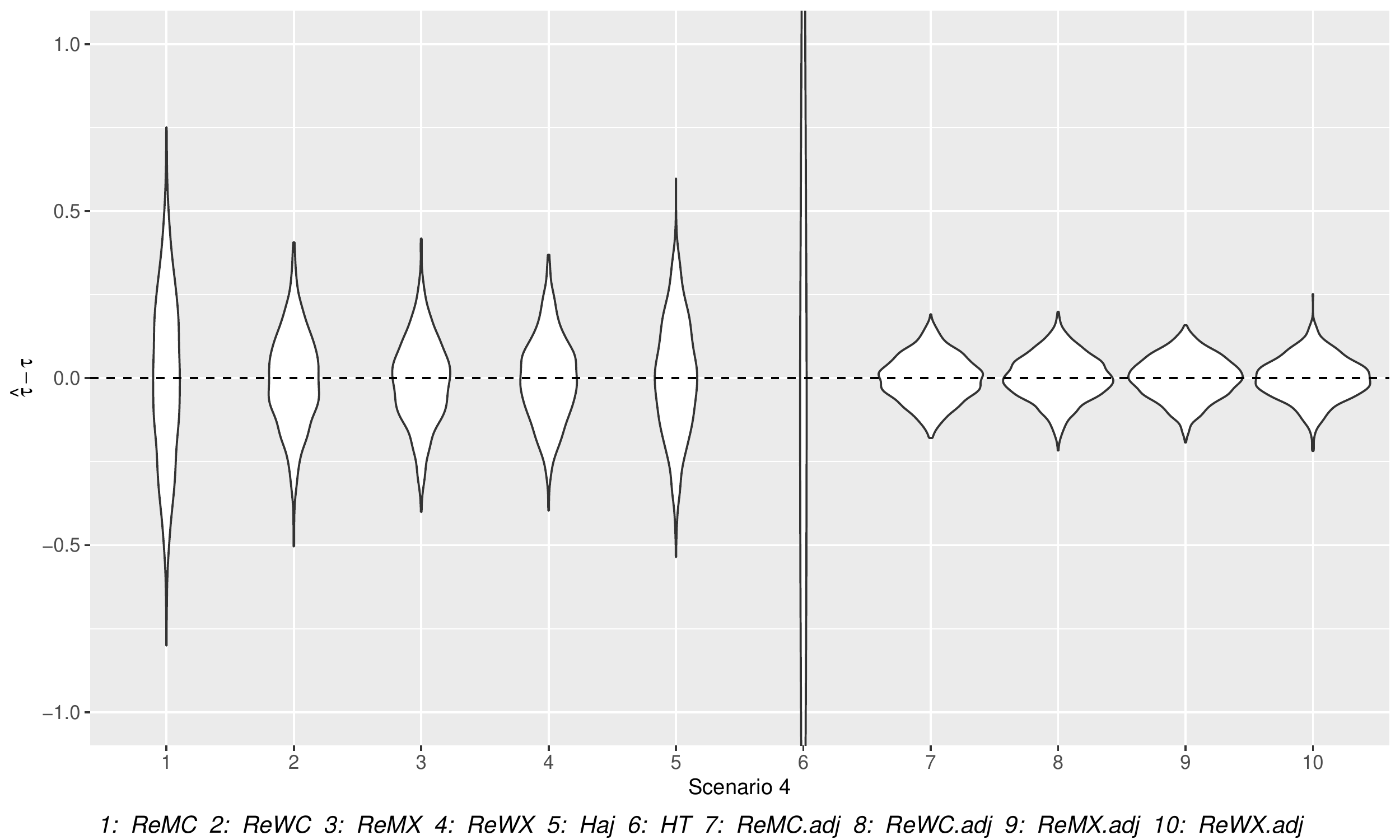}
  \caption{Comparison results on the child monitoring data.}
  \label{fig:real}
\end{figure}

Table~\ref{tab:real} and Figure~\ref{fig:real} show that cluster rerandomization always improves the performance, and cluster rerandomization schemes based on the weighted Euclidean distances are better than cluster rerandomization schemes based on the Mahalanobis distances. In the case when cluster-level covariates are used, the standard deviations under cluster rerandomization schemes using the weighted Euclidean distances are about half of those using the Mahalanobis distances. The combinations of cluster rerandomization and regression adjustment perform the best, and they have almost the same standard deviations and mean squared errors.

\begin{table}
  \centering
  \caption{Comparison results on the child monitoring data}
  \label{tab:real}
  \begin{tabular}{cccccccc}
    \toprule
\multirow{2}{*}{Method} & \multirow{2}{*}{Bias} & \multirow{2}{*}{SD} & \multirow{2}{*}{RMSE} & \multicolumn{2}{c}{Normal-based}   & \multicolumn{2}{c}{Improved}                          \\ 
& & &  & CP           & CI Length    & CP       & CI Length \\

    \midrule
    ReMC     & -0.01 & 0.25 & 0.25 & 1.00 & 4.63 & 0.94 & 0.95 \\ 
    ReWC      & -0.01 & 0.14 & 0.14 & 1.00 & 4.63 & 0.93 & 0.54 \\ 
    ReMX     & -0.01 & 0.13 & 0.13 & 0.99 & 0.72 & 0.94 & 0.51 \\ 
    ReWX    & 0.00 & 0.13 & 0.13 & 0.99 & 0.72 & 0.95 & 0.52 \\ 
    Haj      & 0.00 & 0.17 & 0.17 & 0.96 & 0.72 & --  & -- \\ 
    HT       & -0.05 & 1.20 & 1.20 & 0.94 & 4.60 & -- & -- \\ 
    ReMC.adj & 0.00 & 0.07 & 0.07 & 0.97 & 0.31 &  --& -- \\
    ReWC.adj & 0.00 & 0.07 & 0.07 & 0.97 & 0.31 & -- & -- \\ 
    ReMX.adj & 0.00 & 0.06 & 0.06 & 0.99 & 0.34 & 0.99 & 0.31 \\ 
    ReWX.adj & 0.00 & 0.06 & 0.06 & 0.99 & 0.34 & 0.98 & 0.32 \\ 
    \bottomrule
  \end{tabular}
\end{table}

For this particular dataset, the covariates and outcomes are all positive. All cluster-level covariates and outcomes are strongly correlated with cluster size and there exists strong collinearity between covariates, which helps to explain why cluster rerandomization with the weighted Euclidean distance is much better than cluster rerandomization with the Mahalanobis distance. For cluster rerandomization with cluster-level covariates, the existence of collinearity and the fact that the weighted Euclidean distance puts more weights on more important covariates give rise to a dramatic increase of efficiency. The outcomes in this dataset, children's heights, have a small variation across units, which helps to explain why the Hajek estimator performs much better than the Horvitz--Thompson estimator.

\section*{Acknowledgement}
Dr. Hanzhong Liu was supported by the National Natural Science Foundation of China, Grant No. 12171476. Dr. Peng Ding was partially supported by the U.S. National Science Foundation. 


\bibliographystyle{apalike}
\bibliography{paper-ref}

\begin{thebibliography}{}

\bibitem[Althabe et~al., 2008]{althabe2008behavioral}
Althabe, F., Buekens, P., Bergel, E., Beliz{\'a}n, J.~M., Campbell, M.~K.,
  Moss, N., Hartwell, T., and Wright, L.~L. (2008).
\newblock A behavioral intervention to improve obstetrical care.
\newblock {\em N. Eng. J. Med.}, 358:1929--40.

\bibitem[Alzer, 1997]{alzer1997some}
Alzer, H. (1997).
\newblock On some inequalities for the gamma and psi functions.
\newblock {\em Math. Comput.}, 66:373--89.

\bibitem[Athey and Imbens, 2017]{Athey2017}
Athey, S. and Imbens, G.~W. (2017).
\newblock {\em The Econometrics of Randomized Experiments}.
\newblock North-Holland, Amsterdam.

\bibitem[Das~Gupta et~al., 1972]{gupta1972inequalities}
Das~Gupta, S., Eaton, M.~L., Olkin, I., Perlman, M., Savage, L.~J., and Sobel,
  M. (1972).
\newblock Inequalities on the probability content of convex regions for
  elliptically contoured distributions.
\newblock In {\em Proc. Sixth Berkeley Symp. Math. Statist. Prob.}, volume~2.
  Berkeley, California: University of California Press, pp. 241--65.

\bibitem[de~Hoop et~al., 2012]{de2012best}
de~Hoop, E., Teerenstra, S., van Gaal, B.~G., Moerbeek, M., and Borm, G.~F.
  (2012).
\newblock The “best balance” allocation led to optimal balance in
  cluster-controlled trials.
\newblock {\em J. Clin. Epidemiol.}, 65:132--37.

\bibitem[Dempsey et~al., 2018]{dempsey2018effect}
Dempsey, A.~F., Pyrznawoski, J., Lockhart, S., Barnard, J., Campagna, E.~J.,
  Garrett, K., Fisher, A., Dickinson, L.~M., and O’Leary, S.~T. (2018).
\newblock Effect of a health care professional communication training
  intervention on adolescent human papillomavirus vaccination: a cluster
  randomized clinical trial.
\newblock {\em JAMA Pediatr.}, 172:e180016--e180016.

\bibitem[Dharmadhikari and Joag-Dev, 1988]{dharmadhikari1988unimodality}
Dharmadhikari, S. and Joag-Dev, K. (1988).
\newblock {\em Unimodality, Convexity, and Applications}.
\newblock Elsevier, Amsterdam.

\bibitem[Donner and Klar, 2000]{donner2000design}
Donner, A. and Klar, N. (2000).
\newblock {\em Design and Analysis of Cluster Randomization Trials in Health
  Research}.
\newblock Arnold Publishers Limited, London.

\bibitem[Hayes and Moulton, 2017]{hayes2017cluster}
Hayes, R.~J. and Moulton, L.~H. (2017).
\newblock {\em Cluster Randomised Trials}.
\newblock CRC Press, Florida.

\bibitem[Huber, 1967]{Huber1967}
Huber, P.~J. (1967).
\newblock The behavior of maximum likelihood estimates under nonstandard
  conditions.
\newblock In {\em Proc. Fifth Berkeley Symp. Math. Statist. Prob.}, volume~1.
  Berkeley, California: University of California Press, pp. 221--33.

\bibitem[Johansson et~al., 2021]{johansson2021optimal}
Johansson, P., Rubin, D.~B., and Schultzberg, M. (2021).
\newblock On optimal rerandomization designs.
\newblock {\em J. R. Statist. Soc. B}, 83:395--403.

\bibitem[Li et~al., 2016]{li2016evaluation}
Li, F., Lokhnygina, Y., Murray, D.~M., Heagerty, P.~J., and DeLong, E.~R.
  (2016).
\newblock An evaluation of constrained randomization for the design and
  analysis of group-randomized trials.
\newblock {\em Statist. Med.}, 35:1565--79.

\bibitem[Li et~al., 2017]{li2017evaluation}
Li, F., Turner, E.~L., Heagerty, P.~J., Murray, D.~M., Vollmer, W.~M., and
  DeLong, E.~R. (2017).
\newblock An evaluation of constrained randomization for the design and
  analysis of group-randomized trials with binary outcomes.
\newblock {\em Statist. Med.}, 36:3791--806.

\bibitem[Li and Ding, 2017]{li2017general}
Li, X. and Ding, P. (2017).
\newblock General forms of finite population central limit theorems with
  applications to causal inference.
\newblock {\em J. Am. Statist. Assoc.}, 112:1759--69.

\bibitem[Li and Ding, 2020]{2020Rerandomization}
Li, X. and Ding, P. (2020).
\newblock Rerandomization and regression adjustment.
\newblock {\em J. R. Statist. Soc. B}, 82:241--68.

\bibitem[Li et~al., 2018]{Li9157}
Li, X., Ding, P., and Rubin, D.~B. (2018).
\newblock Asymptotic theory of rerandomization in treatment-control
  experiments.
\newblock {\em Proc. Nat. Acad. Sci.}, 115:9157--62.

\bibitem[Li et~al., 2020]{Li2020factorial}
Li, X., Ding, P., and Rubin, D.~B. (2020).
\newblock Rerandomization in $2^\textup{K}$ factorial experiments.
\newblock {\em Ann. Statist.}, 48:43--63.

\bibitem[Liang and Zeger, 1986]{Liang1986}
Liang, K.-Y. and Zeger, S.~L. (1986).
\newblock Longitudinal data analysis using generalized linear models.
\newblock {\em Biometrika}, 73:13--22.

\bibitem[Lin, 2013]{Lin2013Agnostic}
Lin, W. (2013).
\newblock Agnostic notes on regression adjustments to experimental data:
  Reexamining $\textup{F}$reedman's critique.
\newblock {\em Ann. Appl. Statist.}, 7:295--318.

\bibitem[Middleton and Aronow, 2015]{middleton2015unbiased}
Middleton, J.~A. and Aronow, P.~M. (2015).
\newblock Unbiased estimation of the average treatment effect in
  cluster-randomized experiments.
\newblock {\em Statist. Polit. Policy}, 6:39--75.

\bibitem[Morgan and Rubin, 2012]{morgan2012rerandomization}
Morgan, K.~L. and Rubin, D.~B. (2012).
\newblock Rerandomization to improve covariate balance in experiments.
\newblock {\em Ann. Statist.}, 40:1263--82.

\bibitem[Morgan and Rubin, 2015]{morgankari2015}
Morgan, K.~L. and Rubin, D.~B. (2015).
\newblock Rerandomization to balance tiers of covariates.
\newblock {\em J. Am. Statist. Assoc.}, 110:1412--21.

\bibitem[Moulton, 2004]{moulton2004covariate}
Moulton, L.~H. (2004).
\newblock Covariate-based constrained randomization of group-randomized trials.
\newblock {\em Clin. Trials}, 1:297--305.

\bibitem[Raab and Butcher, 2001]{raab2001balance}
Raab, G.~M. and Butcher, I. (2001).
\newblock Balance in cluster randomized trials.
\newblock {\em Statist. Med.}, 20:351--65.

\bibitem[Raudenbush, 1997]{raudenbush1997statistical}
Raudenbush, S.~W. (1997).
\newblock Statistical analysis and optimal design for cluster randomized
  trials.
\newblock {\em Psychol. Methods}, 2:173--85.

\bibitem[Schochet, 2013]{Schochet2013}
Schochet, P.~Z. (2013).
\newblock Estimators for clustered education \textsc{RCT}s using the
  \textsc{N}eyman model for causal inference.
\newblock {\em J. Educ. Behav. Statist.}, 38:219--38.

\bibitem[Schochet, 2020]{schochet2020analyzing}
Schochet, P.~Z. (2020).
\newblock Analyzing grouped administrative data for {RCTs} using design-based
  methods.
\newblock {\em J. Educ. Behav. Statist.}, 45:32--57.

\bibitem[Su and Ding, 2021]{su2021modelassisted}
Su, F. and Ding, P. (2021).
\newblock Model-assisted analyses of cluster-randomized experiments.
\newblock {\em J. R. Statist. Soc. B}, 83:994--1015.

\bibitem[Tembo et~al., 2017]{Tembo2017Home}
Tembo, Sarah, Levenson, Rachel, Rockers, Peter, C., Fink, and Gunther (2017).
\newblock Home- and community-based growth monitoring to reduce early life
  growth faltering: an open-label, cluster-randomized controlled trial.
\newblock {\em Am. J. Clin. Nutr.}, 106:1070--77.

\bibitem[Turner et~al., 2017a]{turner2017reviewdesign}
Turner, E.~L., Li, F., Gallis, J.~A., Prague, M., and Murray, D.~M. (2017a).
\newblock Review of recent methodological developments in group-randomized
  trials: part 1---design.
\newblock {\em Am. J. Public Health}, 107:907--15.

\bibitem[Turner et~al., 2017b]{turner2017reviewanalysis}
Turner, E.~L., Prague, M., Gallis, J.~A., Li, F., and Murray, D.~M. (2017b).
\newblock Review of recent methodological developments in group-randomized
  trials: part 2---analysis.
\newblock {\em Am. J. Public Health}, 107:1078--86.

\bibitem[White, 1980]{White1980}
White, H. (1980).
\newblock A heteroskedasticity-consistent covariance matrix estimator and a
  direct test for heteroskedasticity.
\newblock {\em Econometrica}, 48:817--38.

\bibitem[Wight et~al., 2002]{wight2002limits}
Wight, D., Raab, G.~M., Henderson, M., Abraham, C., Buston, K., Hart, G., and
  Scott, S. (2002).
\newblock Limits of teacher delivered sex education: interim behavioural
  outcomes from randomised trial.
\newblock {\em Br. Med. J.}, 324:1430--35.

\bibitem[Zhao and Ding, 2021]{Zhao2021}
Zhao, A. and Ding, P. (2021).
\newblock Covariate-adjusted fisher randomization tests for the average
  treatment effect.
\newblock {\em J. Econom.}, page in press.

\end{thebibliography}


\appendix







\section{Additional Simulation}
\label{sec::F}
\LXTL{Let $I_K$ denote a $K \times K$ identity matrix. We generate potential outcomes from the following model:
\begin{equation*}
    Y_{ij}(z) = f(2+\epsilon_{iz}+x_{ij}^\top \beta_{iz})+\varepsilon_{ij}(z) \quad (i=1,\ldots,M,\ j=1,\ldots,n_i,\ z=0,1),
\end{equation*}
where $\beta_{iz} = \beta_z+\varsigma_{iz}$ with the components of $\beta_{z}$ $(z=0,1)$ being independently generated from the  $t$-distribution with three degrees of freedom and $\varsigma_{iz}$ $(i=1,\ldots,M,z=0,1)$ being independently generated from $\mathcal{N}(0,I_K)$. We independently generate $x_{ij}$ $(i=1,\ldots,M,\ j=1,\ldots,n_i)$ from $\mathcal{N}(0,I_K)$ and $\epsilon_{iz}$ $(i=1,\ldots,M,z=0,1)$ from $\mathcal{N}(0,1)$. The error terms $\varepsilon_{ij}(z)$ $(i=1,\ldots,M,\ j=1,\ldots,n_i,\ z=0,1)$ are independently generated from $N(0,\sigma^2_z)$, where $\sigma_z^2$ ($z=0,1$) are chosen so that the $f(2+\epsilon_{iz}+x_{ij}^\top \beta_{iz})$ $(i=1,\ldots,M,\ j=1,\ldots,n_i)$ constitute half of the variation of $Y_{ij}(z)$. We use $(\tilde{x}_{i\cdot},n_i)\in \mathbb{R}^{K+1}$ in cluster rerandomization using Mahalanobis distance and regression adjustment. We view simulation as a full factorial experiment and generate data under all combinations of the following five factors:
\begin{itemize}
    \item Number of clusters $M=\{20+4k\mid k\in\mathbb{N}, \ 0\leq k\leq 15\}$;
    \item Variance of cluster sizes $\textnormal{vn}=\{\textnormal{H},\textnormal{L}\}$. If vn $=$ H, $n_i$'s are generated uniformly from $\{m\in \mathbb{N}\mid 2\leq m \leq 16\}$; if vn $=$ L, $n_i$'s are generated uniformly from $\{m\in \mathbb{N}\mid 4\leq m \leq 10\}$;
    \item Type of outcome-generating functions $\textnormal{fn}=\{\textnormal{linear},\textnormal{nonlinear}\}$. $f(t)=t$ for fn $=$ linear; $f(t)=t^3$ for fn $=$ nonlinear; 
    \item Dimension of covariates $K=\{1,5\}$;
    \item Acceptance rate $\alpha=\{0.001,0.1\}$.
\end{itemize}}

\LXTL{Each of the $16\times 2\times 2\times2\times 2=256$ scenarios is replicated under $100$ random seeds. In each scenario under each random seed, we draw assignments for $1000$ times with half of the clusters being assigned to the treatment arm. We construct $95\%$ normal-based confidence intervals by the method provided in Section~\ref{sec::B} based on the regression residuals. We compute the empirical coverage probabilities (CP) and percentages of reduction in  root mean squared errors (RMSE) of the combination of cluster rerandomization and regression adjustment relative to two benchmark estimators without using cluster rerandomization, i.e., the Horvitz--Thompson estimator $\tauht$ and Hajek estimator $\tauhaj$ under classic cluster randomization.}

\LXTL{As shown in Figure~\ref{fig:cp}, the asymptotic inference does not vary much between vn$=$H and vn$=$L and between fn$=$linear and fn$=$nonlinear. The dimension of covariates plays an important role in the validity of asymptotic inference. When $K=1$, $M=20$ is enough for the validity of the confidence interval under half of the random seeds and $M=24$ is enough for the validity of the confidence interval under most of the random seeds. When $K=5$, $M$ must be greater than $44$ to ensure the validity of asymptotic inference.}

\begin{figure}
  \centering
  \includegraphics[width=.8\textwidth]{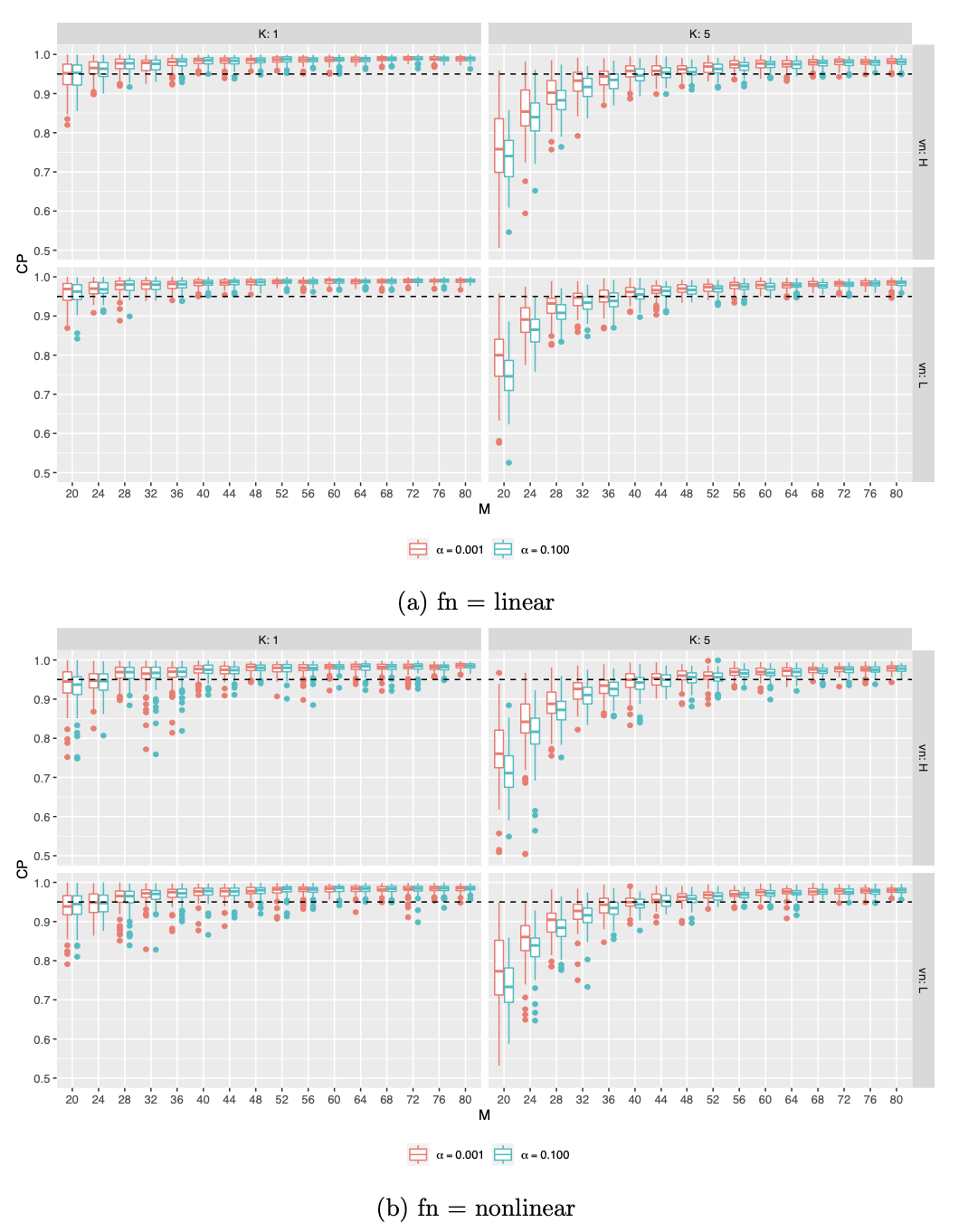}
  \caption{Coverage probabilities for different values of $M$, $K$, and variance of cluster sizes.}
  \label{fig:cp}
\end{figure}

\begin{figure}
  \centering
  \includegraphics[width=.8\textwidth]{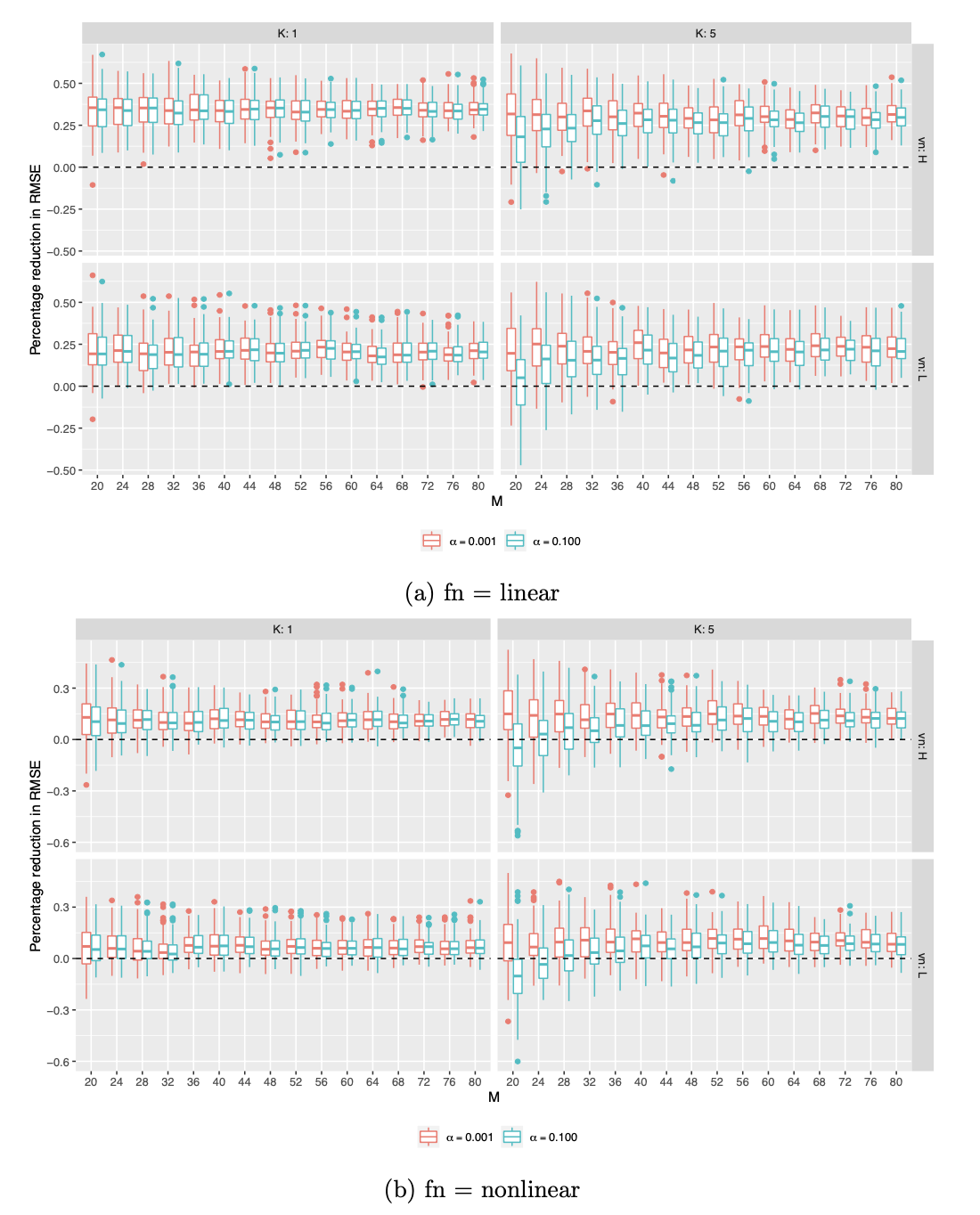}
  \caption{Percentages of reduction in RMSE relative to the Horvitz--Thompson estimator without using cluster rerandomization for different values of $M$, $K$, and variance of cluster sizes.}
  \label{fig:prmse_ht}
\end{figure}

\begin{figure}
  \centering
  \includegraphics[width=.8\textwidth]{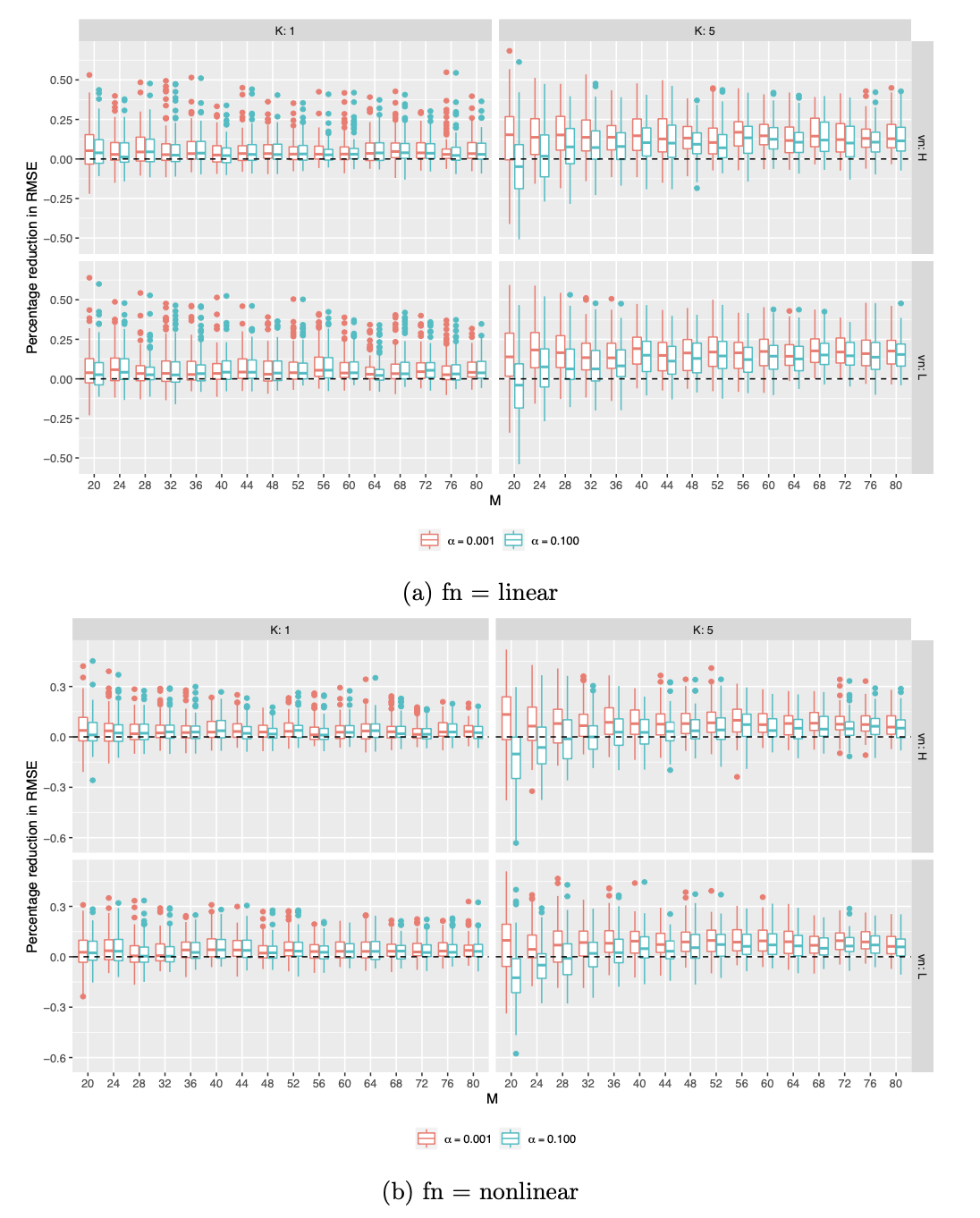}
  \caption{Percentages of reduction in RMSE relative to the Hajek  estimator without using cluster rerandomization for different values of $M$, $K$, and variance of cluster sizes.}
  \label{fig:prmse_haj}
\end{figure}

Figures~\ref{fig:prmse_ht}--\ref{fig:prmse_haj} show the percentages of reduction in RMSE. We can see that the combination of cluster rerandomization and regression adjustment is beneficial in reducing RMSE relative to the two benchmark estimators under most of the random seeds for all values of $M$. The benefit is more evident relative to the Horvitz--Thompson estimator than to the Hajek estimator.

\LXTL{The asymptotic inference is more likely to be valid and the regression adjustment provides more efficiency gains under more balanced assignments with $\alpha=0.001$ than with $\alpha=0.1$. This demonstrates the benefit of cluster rerandomization.}

\section{More detailed results on cluster rerandomization and regression adjustment}\label{sec:apreadj}

\subsection{Asymptotic properties of regression-adjusted estimators under cluster rerandomization}


In general, the covariates used in the analysis stage can be different from those used in the design stage. In the analysis stage,  we use ${w}_{ij}$ for individual-level covariates and ${v}_i$ for cluster-level covariates. Suppose that all these covariates are centered such that
$\sum_{i=1}^M\sum_{j=1}^{n_i}{w}_{ij} = 0$ and $\sum_{i=1}^M{v}_i = 0$.
Define $\tilde{{w}}_{i\cdot} = \sum_{j=1}^{n_i}{w}_{ij}M/N$. Let ${V}$ and $
  {\tilde{W}}$ be the concatenations of ${v}_i$ and $\tilde{  w}_{i\cdot}$ respectively.

With a slight abuse of notation, let $\tauadjhaj$ denote the coefficient of $Z_{ij}$ obtained from the ordinary least squares fit of $Y_{ij}$ on $(1,Z_{ij},{w}_{ij},Z_{ij}{w}_{ij})$ in the individual-level regression adjustment, and $\tauadjht$  denote the coefficient of $Z_i$ obtained from the ordinary least squares fit of $\tilde{Y}_{i\cdot}$ on $(1,Z_i,{v}_{i},Z_i{v}_{i})$ in the cluster-level regression adjustment.


\cite{2020Rerandomization} derived the asymptotic distributions of the regression-adjusted estimators under rerandomization with treatments assigned at the individual level. We extend their results to cluster rerandomization. Define, for $z=0,1$,
\begin{align*}
  {\beta}_{w}(z) =& \arg \min_{{\beta}} \sum_{i=1}^M\sum_{j=1}^{n_i}\left\{Y_{ij}(z)-\Bar{Y}(z)- {w}_{ij}^\top {\beta}  \right\}^2,\\
  {\beta}_{v}(z) =& \arg \min_{\beta} \sum_{i=1}^M\left\{\tilde{Y}_{i\cdot}(z)-\Bar{Y}(z)- {v}_{i}^\top {\beta}  \right\}^2.
\end{align*}

Let $\vtthajb$, $\vtxhajb$, $\vxthajb$, $\vtthajb$, and $(R_{ x}^{\textnormal{adj}})^{2}$ be the quantities defined similarly to those in Proposition~\ref{prop:FCLT} except that we replace $Y_{ij}(z)$ by $Y_{ij}(z)- {w}_{ij}^{\top} {\beta}_{w}(z)$, $z=0,1$. Let $\vtthtb$, $\vtchtb$, $\vcthtb$, $\vtthtb$ and $(R_{c}^{\textnormal{adj}})^{2}$ be the quantities defined similarly to those in Proposition~\ref{prop:FCLT}  except that we replace $\tilde{Y}_{i\cdot}(z)$ by $\tilde{Y}_{i\cdot}(z)- {v}_{i}^{\top} {\beta}_{v}(z)$, $z=0,1$. Their limits exist under regularity conditions given in Section~\ref{sec:regcon}, and we use the same notation to denote their limits if no confusion arises. We have the following result.
\begin{theorem}
  \label{thm:tauhtadjasymp}
  Under regularity conditions,  cluster rerandomization with individual-level covariates satisfies 
  \begin{gather*}
    M^{1/2}(\tauadjhaj-\tau)\mid\mx   \mathrel{\ \dot{\sim}\ } (\vtthajb)^{1/2}\big[\{1-(R^{\textnormal{adj}}_{ x})^{2}\}^{1/2} \epsilon+R^{\textnormal{adj}}_{ x}L_{K,a}\big],\\
    M^{1/2}(\tauadjhaj-\tau)\mid \mathcal{D}_{x}(A_x) \mathrel{\ \dot{\sim}\ } (\vtthajb)^{1/2} \big[
    \{1-(R^{\textnormal{adj}}_{ x})^{2}\}^{1/2} \epsilon+
    R^{\textnormal{adj}}_{ x}{\mu}_x^{\top}{\eta} \mid {\eta}^{\top}\vxxhaj^{1/2}A_x\vxxhaj^{1/2}{\eta} \leq a
    \big],
  \end{gather*}
  and cluster rerandomization with cluster-level covariates satisfies 
  \begin{gather*}
      M^{1/2}(\tauadjht-\tau)\mid\mc  \mathrel{\ \dot{\sim}\ } (\vtthtb)^{1/2}\big[
    \{1-(R^{\textnormal{adj}}_{ c})^{2}\}^{1/2} \epsilon+
    R^{\textnormal{adj}}_{ c}L_{K,a}
    \big],\\
    M^{1/2}(\tauadjht-\tau)\mid \mathcal{D}_{c}(A_c)  \mathrel{\ \dot{\sim}\ } (\vtthtb)^{1/2} \big[
    \{1-(R^{\textnormal{adj}}_{c})^{2}\}^{1/2} \epsilon+
    R^{\textnormal{adj}}_{c}{\mu}_c^{\top}{\eta} \mid {\eta}^{\top}\vccht^{1/2}A_c\vccht^{1/2}{\eta}\leq a
    \big],
  \end{gather*}
  where  $\eta = (\eta_1, \ldots,\eta_K)^\top$,  $\epsilon$, $\eta_1, \ldots,\eta_K$ are independent standard normal random variables, and
  \begin{gather*}
    {\mu}_x^\top =\vtxhajb \vxxhaj^{-1/2}/\{\vtxhajb \vxxhaj^{-1}\vxthajb\}^{1/2},\quad
    {\mu}_c^\top = \vtchtb \vccht^{-1/2}/\{\vtchtb \vccht^{-1}\vcthtb\}^{1/2}.
  \end{gather*}
\end{theorem}

\subsection{The gains from cluster rerandomization and regression adjustment}

To compare the efficiency, \citet{2020Rerandomization} introduced the notion of S-optimality to refer to the coefficient vector ${\beta}$ that could lead to the shortest asymptotic quantile range among all adjusted estimators. For cluster rerandomization with cluster-level covariates, the coefficient vector with respect to S-optimal regression-adjusted average treatment effect estimator coincides with ${\beta}_{v}(z)$ when ${v}_{i} = B{c}_{i}$ or ${c}_{i}=B{v}_{i}$ for certain matrix $B$. For cluster rerandomization with individual-level covariates, it may not be ${\beta}_{w}(z)$, which is the coefficient vector of $w_{ij}$ from the ordinary least squares fit of $Y_{ij}(z)$ on $(1,{w}_{ij})$. Instead, it equals the coefficient vector of ${\tilde{w}}_{i\cdot}$ from the ordinary least squares fit of $\tilde{Y}_{i\cdot}(z)$ on $(1,{\tilde{w}}_{i\cdot})$. \cite{su2021modelassisted} concluded that the individual-level regression adjustment is suboptimal, so we will not discuss the efficiency gains from the individual-level regression adjustment.

\cite{Li2020factorial} proved the efficiency gain for the combination of rerandomization and regression adjustment with treatments assigned at the individual level if we use more or less information  in the regression adjustment. We reexamine their results under cluster rerandomization using cluster-level covariates. We consider two special cases of cluster-level regression adjustment: ${v}_{i} = B{c}_{i}$ or ${c}_{i}=B{v}_{i}$. The former indicates that there is less information at the analysis stage, whereas the latter indicates that there is more information at the analysis stage. For the latter case, Corollary~\ref{coro:anamoreinfo} below shows that regression adjustment can further improve the efficiency. 


\begin{corollary}
Under regularity conditions,
  \label{coro:anamoreinfo}
  if ${c}_{i}=B{v}_{i}$,
  \begin{align*}
    M^{1/2}(\tauadjht-\tau)\mid \mathcal{D}_{c}(A_c)  \mathrel{\ \dot{\sim}\ } (\vttht)^{1/2} (1-R^2_{ v})^{1/2} \epsilon,
  \end{align*}
  where $\epsilon$ is a standard normal random variable, and $R^2_{ v}$ is the squared multiple correlation defined similarly to $R_{  c}^2$ with $ c$ being replaced by $ v$.
\end{corollary}

Corollary~\ref{coro:anamoreinfo} can be interpreted as follows. If we use cluster-level covariates and have more information in the analysis stage, then using the combination of cluster rerandomization  and regression adjustment is asymptotically equivalent to using regression adjustment only, which is no worse than using cluster rerandomization only.

For the case we have less information in the analysis stage, we have the following corollaries.

\begin{corollary}
  \label{coro:analessinfo}
 Under regularity conditions, if ${v}_{i}=B{c}_{i}$,
  \begin{align*}
    M^{1/2}(\tauadjht-\tau)\mid \mathcal{D}_{c}(A_c)   \mathrel{\ \dot{\sim}\ } (\vttht)^{1/2}\big\{ (1-R^2_{ c})^{1/2} \epsilon + R_{ c}{\mu}^{\top}{\eta}\mid {\eta}^{\top}\vccht^{1/2}A_c\vccht^{1/2}{\eta} \leq a\big\},
  \end{align*}
  where $\eta = (\eta_1, \ldots,\eta_K)^\top$, $\epsilon$, $\eta_1, \ldots,\eta_K$ are independent standard normal random variables,  
  $$
    {\mu}^\top =  (\mu_1,\ldots,\mu_K) = \{\vtcht \vccht^{-1}\vctht\}^{-1/2} \vtcht \vccht^{-1/2}H = \mu_c^\top H,
  $$
  and
  $$
    H = I_{K}-\vccht^{1/2}B^{\top}(B\vccht B^{\top})^{-1}B\vccht^{1/2}.
  $$
\end{corollary}

\begin{corollary}
  \label{cor:lessinfocovotho}
  Suppose that the covariates are orthogonalized. If we 
   use the first $J<K$ covariates in the regression adjustment, then under regularity conditions,
  \[
  \vara\{M^{1/2}(\tauadjht-\tau)\mid \mathcal{D}_{c}(A_c) \} \leq  \vara\{M^{1/2}(\tauht-\tau)\mid \mathcal{D}_{c}(A_c) \}
  .\]
\end{corollary}


By Corollary~\ref{coro:analessinfo}, the asymptotic distribution of $M^{1/2}(\tauadjht-\tau)$ under $\mathcal{D}_{c}(A_c)$ is determined by the direction and length of the vector ${\mu}$. For an 
ellipsoidal acceptance region of the cluster rerandomization criterion, some directions lead to variance reduction while others do not. Regression adjustment is equivalent to a projection of ${\mu_c}$ onto a direction determined by $H$. However, such projection might not be in a favorable direction and there is no guarantee of efficiency gain in general. For some special cases as stated in Corollary~\ref{cor:lessinfocovotho}, the asymptotic variance can be reduced. Moreover, if we use cluster rerandomization based on the Mahalanobis distance, the acceptance region is a ball and regression adjustment as a projection  decreases the length of the vector ${\mu_c}$, leading to variance reduction. This echoes the result in \cite{Li2020factorial} that regression adjustment  improves the precision when we have less information in the analysis stage.

The gain from cluster rerandomization can be seen directly from Theorem~\ref{thm:tauhtadjasymp} and Proposition~\ref{prop:concentration-symmetric-unimodality}  where $R^{\textnormal{adj}}_{c}$ and $R^{\textnormal{adj}}_{x}$ are equal to $0$ without cluster rerandomization and  greater than or equal to $0$ with cluster rerandomization. 


\subsection{Estimating the sampling distributions of regression-adjusted estimators}
In the main text, we discuss the validity of $\vhatadjlz$ and $\vhatadjhw$ when $c_i=v_i$ and $x_{ij} = w_{ij}$. Theorem~\ref{thm:vhwvlzop} below states general results without assuming $c_i=v_i$ and $x_{ij} = w_{ij}$.


\begin{theorem}
  \label{thm:vhwvlzop}
  Under regularity conditions, 
  \begin{gather*}
   M \vhatadjlz - \vtthajb\mid \mathcal{D}_{x}(A_x) \mathrel{\geq_{\textup{p}}} 0, \quad
   M \vhatadjhw - \vtthtb \mid \mathcal{D}_{c}(A_c) \mathrel{\geq_{\textup{p}}} 0,
  \end{gather*}
  where a random variable $\geq_{\textup{p}}0$ means that its probability limit is greater than or equal to $0$.
\end{theorem}

By Theorem~\ref{thm:vhwvlzop}, the two types of robust standard errors are asymptotically conservative under cluster rerandomization.

\section{Improvement for variance estimation and confidence intervals}\label{sec::B}
Although $\vhatadjhw$ and $\vhatadjlz$ are  valid under cluster rerandomization, they can be overly conservative partly because  they ignore that the asymptotic distributions have truncated normal components due to cluster rerandomization. We can derive better variance and confidence interval estimators by better estimating the sampling distributions.

In Theorem~\ref{thm:tauhtadjasymp}, we have derived the asymptotic distributions of $M^{1/2}(\tauadjht-\tau)$ and $M^{1/2}(\tauadjhaj-\tau)$ under cluster rerandomization. First, we derive better estimators of $\vtthajb$ and $\vtthtb$.  Define the adjusted potential outcomes as $\tdadjyi(z)=\tilde{Y}_{i\cdot}(z)- {v}_{i}^{\top} {\beta}_{v}(z)$ and $\varepsilon_{ij}^{\textnormal{adj}}(z) = {\varepsilon}_{ij}(z)- {w}_{ij}^{\top} {\beta}_{w}(z)$, for $z=0,1$. Let $\tdadjei(z)$ be the scaled cluster total of ${\varepsilon}_{ij}(z)$. Denote $\tdadjy(z)$ and $\tdadje(z)$ as the vectors of $\tdadjyi(z)$ and $\tdadjei(z)$, respectively. Recall that $\vtthajb$ and $\vtthtb$ are defined similarly as $\vtthaj$ and $\vttht$ except that we replace $\varepsilon_{ij}(z)$ by $\varepsilon_{ij}^{\textnormal{adj}}(z)$ and $\tilde{Y}_{i\cdot}(z)$ by $\tdadjyi(z)$, for $z=0,1$.

Define the variance of the linear projection of $\tdadjy(z)$ on ${C}$ as
\[
  \varf\{\tdadjy(z) \mid {C}\}=\covf\{\tdadjy(z), {C}\}\{\covf({C})\}^{-1} \covf\{{C}, \tdadjy(z)\}, ~~ z=0,1.
\]
Similarly, define
\[
  \varf\{\tdadje(z) \mid \tilde{{X}}\}=\covf\{\tdadje(z),\tilde{{X}}\}\{\covf(\tilde{{X}})\}^{-1} \covf\{\tilde{{X}}, \tdadje(z)\}, ~~z=0,1.
\]
Let $D$ be the vector of residuals of the cluster-level regression adjustment. Let ${u}_{ij}$ be the residual of the individual-level regression adjustment, $\tilde{u}_{i\cdot}$ be its scaled cluster total, and $\tilde{U}$ be the vector of $\tilde{u}_{i\cdot}$. Denote by $\hvarfz(\cdot)$ and $\hcovfz(\cdot)$ the sample variance and covariance under treatment arm $z$. Let $\tilde{G}$ be the union of $\tilde{X}$ and $\tilde{W}$. To simplify the presentation, we assume that $\covf (\tilde{{G}})$ is invertible.

Following the  results of \citet{Li9157}, less conservative estimators of $\vtthajb$ and $\vtthtb$ are
\begin{align*}
    \hatvtthajb &= {e_1^{-1} \hvarf \{\tdadje(1)\} +
  e_0^{-1}\hvarf \{\tdadje(0)\} -
  \hvarf \{\tdadje (1)-\tdadje (0)\mid \tilde{G}\}},\\
 \hatvtthtb &=  e_1^{-1} \hvarf \{\tdadjy(1)\} +
  e_0^{-1}\hvarf \{\tdadjy(0)\} -
  \hvarf \{\tdadjy(1)-\tdadjy(0)\mid C\} ,
\end{align*}
where $\hvarf \{\tdadje(z)\} = \hvarfz(\tilde{U}),$ $\hvarf \{\tdadjy(z)\}=
   \hvarfz (D),$  for $z=0,1,$ and
\begin{align*}
     \hvarf \{\tdadjy(1)-\tdadjy(0)\mid {C} \}&=\big\{ \hcovft (D,{C} )-\hcovfc (D,{C} ) \big\} \{\covf ({C})\}^{-1} \big\{ \hcovft ({C},D )-\hcovfc ({C},D ) \big\},\\
   \hvarf   \{\tdadje (1)-\tdadje (0)\mid \tilde{{G}} \}&= \{ \hcovft ( \tilde{U}, \tilde{{G}} ) - \hcovfc ( \tilde{U}, \tilde{{G}} ) \} \{\covf (\tilde{{G}})\}^{-1} \{ \hcovft ( \tilde{{G}},\tilde{U} ) - \hcovfc ( \tilde{{G}},\tilde{U} )  \}.
\end{align*}

Second, we estimate  $R^{\textnormal{adj}}_{x}$ and $R^{\textnormal{adj}}_{c}$. 
Recall that $\radjx$ and $\radjc$ are defined similarly as $R_{x}$ and $R_{c}$ except that we replace $Y_{ij}(z)$ by $ Y_{ij}(z)- {w}_{ij}^{\top} {\beta}_{w}(z)$ and $\tilde{Y}_{i\cdot}(z)$ by $\tdadjyi(z)=\tilde{Y}_{i\cdot}(z)- {v}_{i}^{\top} {\beta}_{v}(z)$, for $z=0,1$. 
Applying Proposition~1 of \citet{Li9157}, we obtain the following results for $\radjx$ and $\radjc$.
\begin{proposition}\label{thm:mcorr}
  Under regularity conditions, $(\radjx)^2$ or $(\radjc)^2$ can be expressed in terms of the variances of the adjusted potential outcomes and of their projections on $\tilde{{X}}$ or $C$:
  \begin{equation}
    \label{equ:radjx}
    (\radjx)^2  =\frac{
    e_1^{-1}\varf  \{\tdadje(1)\mid \tilde{{X}}\} +
    e_0^{-1} \varf  \{\tdadje(0) \mid \tilde{{X}}\} -
    \varf  \{\tdadje (1)-\tdadje (0)\mid \tilde{{X}}\}
    }{e_1^{-1}\varf  \{\tdadje (1)\} + 
    e_0^{-1} \varf  \{\tdadje (0)\} -
    \varf  \{\tdadje (1)-\tdadje (0)\}
    },
  \end{equation}

  \begin{equation}
    \label{equ:radjc}
    (\radjc)^2
    =\frac{
    e_1^{-1}\varf  \{\tdadjy(1)\mid {C}\} +
    e_0^{-1} \varf  \{\tdadjy(0)\mid {C}\} -
    \varf  \{\tdadjy(1)-\tdadjy(0)\mid {C}\}
    }{
    e_1^{-1}\varf  \{\tdadjy(1)\} + 
    e_0^{-1} \varf  \{\tdadjy(0)\} -
    \varf  \{\tdadjy(1)-\tdadjy(0)\}
    }.
  \end{equation}
\end{proposition}

The denominator of \eqref{equ:radjx} is equal to $\vtthajb$ and the denominator of \eqref{equ:radjc} is equal to $\vtthtb$. Following the  results of \citet{Li9157}, we can estimate $\varf   \{\tdadjy(z)\mid {C} \}$ by $$\hvarf \{\tdadjy(z)\mid {C} \}=
   \hcovfz (D,{C} )\{ \hcovfz ({C}) \}^{-1}\hcovfz ({C},D ).
$$
Similarly, $\varf   \{\tdadjy(1)-\tdadjy(0)\mid {C} \}$ can be estimated by $ \hvarf \{\tdadjy(1)-\tdadjy(0)\mid {C} \}$.
Then we get an estimator for the numerator of \eqref{equ:radjc}. 
Thus, we estimate $(\radjc)^2$ by
$$
 (\hradjc)^2 = \big[e_1^{-1} \hvarf \{\tdadjy(1)\mid {C}\} +
  e_0^{-1}\hvarf \{\tdadjy(0)\mid {C}\} -
  \hvarf \{\tdadjy(1)-\tdadjy(0)\mid {C}\}\big] \Big / {\hatvtthtb}.
$$

For cluster rerandomization using individual-level covariates,  $\varf   \{\tdadje (z)\mid \tilde{{X}} \}$ can be estimated by
\[
  \hvarf \{\tdadje(z)\mid \tilde{{X}} \} = \hcovfz ( \tilde{U}, \tilde{{X}} )
  \{ \hcovfz (\tilde{{X}} ) \}^{-1}\hcovfz (\tilde{{X}},{\tilde{U}} ).
\]
Similarly, 
$\varf   \{\tdadje (1)-\tdadje (0)\mid \tilde{{X}} \}$ can be estimated by 
$$
 \hvarf   \{\tdadje (1)-\tdadje (0)\mid \tilde{{X}} \}= \{ \hcovft ( \tilde{U}, \tilde{{X}} ) - \hcovfc ( \tilde{U}, \tilde{{X}} ) \} \{\covf (\tilde{{X}})\}^{-1} \{ \hcovft ( \tilde{{X}},\tilde{U} ) - \hcovfc ( \tilde{{X}},\tilde{U} )  \}.
$$
Again we derive a consistent estimator for the numerator of \eqref{equ:radjx}. Thus, we estimate $(\radjx)^2$ by
$$
(\hradjx)^2 = \big[{e_1^{-1} \hvarf \{\tdadje(1) \mid \tilde{{X}}\} +
  e_0^{-1}\hvarf \{\tdadje(0)\mid \tilde{{X}}\} -
  \hvarf \{\tdadje (1)-\tdadje (0)\mid \tilde{{X}}\}}\big] \Big/ {\hatvtthajb}.
$$

Finally, we estimate $\mu_x$ and $\mu_c$ defined in Theorem~\ref{thm:tauhtadjasymp}.  Recall that 
\begin{eqnarray}
        {\mu}_x^\top &=& \vtxhajb \vxxhaj^{-1/2}/\{\vtxhajb \vxxhaj^{-1}\vxthajb\}^{1/2},   \label{equ:mux} \\
         {\mu}_c^\top &=& \vtchtb \vccht^{-1/2}/\{\vtchtb \vccht^{-1}\vcthtb\}^{1/2}.   \label{equ:muc}
\end{eqnarray}
The denominators of \eqref{equ:mux} and \eqref{equ:muc} equal the square roots of the numerators of \eqref{equ:radjx} and \eqref{equ:radjc}, respectively. We have derived their consistent estimators. Recall that
\begin{align*}
     \vxthajb &=(\vtxhajb)^{\top}=  e_1 ^{-1}\covf   \{\tilde{{X}},\tdadje (1) \} +  e_0 ^{-1}\covf\{\tilde{{X}},\tdadje (0) \},\\
      \vcthtb  &=(\vtchtb)^{\top}= e_1 ^{-1}\covf   \{{C},\tdadjy(1) \}+  e_0 ^{-1}\covf   \{{C},\tdadjy(0) \}.
\end{align*}
We estimate $\covf   \{\tilde{{X}},\tdadje(z) \}$ by $\hcovfz(\tilde{X}, \tilde{U})$ and $\covf   \{{C},\tdadjy(z) \}$ by $\hcovfz  (C,D )$. 
Plugging in these terms, we derive consistent estimators of $\mu_x$ and $\mu_c$, denoted as $\hat{\mu}_x$ and $\hat{\mu}_c$, respectively. 


Recall that $\xi_k $ is a $K$-dimensional vector with $1$ at the $k$th dimension and $0$ at other dimensions and $I_K$ is a $K \times K$ identity matrix. Let $q_{\zeta}(R^2,\mu,A,a)$ be the $\zeta$th quantile of $(1-R^2)\epsilon+R^2\mu^\top\eta\mid \eta^\top A\eta \leq a$.  Theorem~\ref{thm:interval-conservative} below provides asymptotically conservative confidence intervals.

\begin{theorem}
  \label{thm:interval-conservative}
  Under regularity conditions, 
  \begin{gather*}
    \Bigl[\tauadjhaj+\bigl(\hatvtthajb\bigm/M\bigr)^{1/2}q_{\varsigma/2}\bigl\{(\hradjx)^2, \xi_{1}, I_K,a\bigr\},\ \ \  \tauadjhaj+\bigl(\hatvtthajb\bigm/M\bigr)^{1/2}q_{1-\varsigma/2}\bigl\{(\hradjx)^2, \xi_1,I_K, a\bigr\}\Bigr],\\
    \Bigl[\tauadjht+\bigl(\hatvtthtb\bigm/M\bigr)^{1/2}q_{\varsigma/2}\bigl\{(\hradjc)^2, \xi_1,I_K, a\bigr\}, \ \ \ \tauadjht+\bigl(\hatvtthtb\bigm/M\bigr)^{1/2}q_{1-\varsigma/2}\bigl\{(\hradjc)^2, \xi_1,I_K, ,a\bigr\}\Bigr]
  \end{gather*}
are asymptotically conservative $1-\varsigma$ confidence intervals of $\tau$ under $\mx$ and $\mc$, respectively;
  \begin{gather*}
         \Bigl[\tauadjhaj+\bigl(\hatvtthajb\bigm/M\bigr)^{1/2}q_{\varsigma/2}\bigl\{(\hradjx)^2, \hat{\mu}_x, \vxxhaj^{1/2}A_x\vxxhaj^{1/2},a\bigr\},\qquad\qquad\qquad\qquad\qquad\qquad\qquad\qquad
         \\
         \qquad\qquad\qquad\qquad\qquad\qquad\qquad\qquad\tauadjhaj+\bigl(\hatvtthajb\bigm/M\bigr)^{1/2}q_{1-\varsigma/2}\bigl\{(\hradjx)^2, \hat{\mu}_x,\vxxhaj^{1/2}A_x\vxxhaj^{1/2}, a\bigr\}\Bigr], \\
        \Bigl[\tauadjht+\bigl(\hatvtthtb\bigm/M\bigr)^{1/2}q_{\varsigma/2}\bigl\{(\hradjc\}^2, \hat{\mu}_c,\vccht^{1/2}A_c\vccht^{1/2}, a\bigr\},\qquad\qquad\qquad\qquad\qquad\qquad\qquad\qquad\\ 
        \qquad\qquad\qquad\qquad\qquad\qquad\qquad\qquad\tauadjht+\bigl(\hatvtthtb\bigm/M\bigr)^{1/2}q_{1-\varsigma/2}\bigl\{(\hradjc)^2, \hat{\mu}_c,\vccht^{1/2}A_c\vccht^{1/2}, a\bigr\}\Bigr]
\end{gather*}
 are asymptotically conservative $1-\varsigma$ confidence intervals of $\tau$ under $\dxa$ and $\dca$, respectively.
\end{theorem}

\section{Results for cluster rerandomization with cluster-level covariates that are omitted in Section~\ref{sec:comparison-tier}}\label{sec::C}

We assume that the $c_i$'s are orthogonalized by Gram-Schmidt orthogonalization. Let
$$
  \vccht = \operatorname{diag}(V_{\textnormal{ht},c_1c_1},\ldots,V_{\textnormal{ht},c_Kc_K}),
  \quad
  \vtcht = (V_{\textnormal{ht},\tau c_{1}}, \ldots, V_{\textnormal{ht},\tau c_K}).
$$
Define the squared multiple correlation for the $k$th  covariate as
$
  R_{c_k}^2 =  V_{\textnormal{ht},\tau c_k }^2/( V_{\textnormal{ht},\tau\tau}V_{\textnormal{ht},c_k  c_k }).
$
The optimal cluster rerandomization scheme is
$
  \mathcal{D}_c(A_{c}^{\textnormal{opt}})= \{M{\hat{\tau}_{\textnormal{ht}, c}}^\top A_{c}^{\textnormal{opt}} {\hat{\tau}_{\textnormal{ht}, c}}<a\}
$
with the optimal diagonal matrix
$$
  A_{c}^{\textnormal{opt}} = \operatorname{diag}\{(V_{\textnormal{ht},\tau c_{1}}V_{\textnormal{ht},c_{1}c_{1}}^{-1})^2,\ldots,(V_{\textnormal{haj},\tau c_{k}}V_{\textnormal{haj},c_{K}c_{K}}^{-1})^2\} .
$$
Corollary~\ref{cor:varalphaortho-cluster-covariate} below is  the counterpart of Corollary  \ref{cor:varalphaortho} in the main text. Recall that
\begin{equation} 
    p_{K}=\frac{2\pi}{K+2}\left\{ \frac{2\pi^{K/2}}{K\Gamma (K/2)} \right\}^{-2/K}. \nonumber
\end{equation}

\begin{corollary}
  \label{cor:varalphaortho-cluster-covariate}
  Under regularity conditions with orthogonalized covariates and optimal weighted Euclidean distance, if ${V}_{\textnormal{haj},\tau c_k }{V}_{\textnormal{ht},c_k c_k }^{-1}$'s are nonzero for all $k =1,\ldots,K$, then
  $$
    \vara    \{M^{1/2}(\tauht-\tau) \mid \mathcal{D}_c(A_c^{\textnormal{opt}}) \} =  V_{\textnormal{ht},\tau \tau} \Big\{(1- R_c^2)+ K \Big(\prod_{k =1}^K R_{c_{k}}^2\Big)^{1/K}p_K \alpha^{2/K} +o(\alpha^{2/K}) \Big\},
  $$
  for a small $\alpha$.

\end{corollary}

Similar to Section~\ref{sec:comparison-tier}, we partition the covariates into $L$ tiers with $K_l$ covariates in tier $l$, for $l=1,\ldots,L$ and $\sum_{l=1}^L K_l = K$. Let ${c}_{[l]}$ be the covariates in the $l$th tier. Then ${c} =({c}_{[1]},\ldots,{c}_{[L]})$. Within tier $l$, define $\hat{\tau}_{\textnormal{ht},c_{[l]}}$ as the Horvitz--Thompson estimator for the covariates, $V_{\textnormal{ht},c_{[l]}c_{[l]}}$ as its covariance, $\mathcal{M}_{[l]}$ as the event that the corresponding Mahalanobis distance is smaller than or equal to the threshold $a_{[l]}$, $R_{c_{[l]}}^2$ as the corresponding $R^2$, and
$\alpha_{[l]}  $ as the asymptotic acceptance rate. The overall asymptotic acceptance rate satisfies $\alpha = \prod_{l=1}^L \alpha_{[l]}$ due to the orthogonalization of the covariates. Corollary~\ref{cor:varalpha-tier-cluster-covariate} below is the counterpart of Corollary~\ref{cor:varalphatier} for cluster rerandomization with cluster-level covariates.

\begin{corollary}
  \label{cor:varalpha-tier-cluster-covariate}
  Under regularity conditions,
  $$
    \vara  \{M^{1/2} (\tauht-\tau)\mid \mathcal{M}_{[1]},\ldots,\mathcal{M}_{[L]}  \} =   V_{\textnormal{ht},\tau \tau} \Big\{(1- R_c^2)+ \sum_{l=1}^L R_{c_{[l]}}^2 p_{K_{l}}\alpha_{[l]}^{2/K_{l}}  +o(\alpha_{[l]}^{2/K_l}) \Big\},
  $$
  for small $\alpha_{[l]}$, $l=1,\ldots,L$.
\end{corollary}

Given $\alpha= \prod_{l=1}^L \alpha_{[l]} > 0$,  to minimize the second term in the asymptotic variance in Corollary \ref{cor:varalpha-tier-cluster-covariate}, we must choose
$$
  \alpha_{[l]}=\left({c_0 {R}_{c_{[l]}}^2 p_{K_{l}}}/{K_{l}}\right)^{-K_{l}/2} \quad (l = 1,\ldots,L),
$$
for some positive constant $c_0$.  Theorem~\ref{thm:wtedbetter-cluster-covariate} below  is the counterpart of Theorem~\ref{thm:wtedbetter} for cluster rerandomization with cluster-level covariates.

\begin{theorem}
  \label{thm:wtedbetter-cluster-covariate}
  Under regularity conditions with orthogonalized covariates, the following inequality holds for any $\alpha_{[1]},\ldots,\alpha_{[L]}$ satisfying $\alpha= \prod_{l=1}^L \alpha_{[l]}$:
  \begin{align*}
    \sum_{l=1}^L {R}_{c_{[l]}}^2 p_{K_{l}}\alpha_{[l]}^{2/K_l} \geq  K\Big(\prod_{k=1}^K R_{c_{k}}^2\Big)^{1/K}p_K \alpha^{2/K}.
  \end{align*}
\end{theorem}

Theorem~\ref{thm:wtedbetter-cluster-covariate}  quantifies  the superiority of cluster rerandomization scheme based on the optimal weighted Euclidean distance when the cluster-level covariates are orthogonalized.

\section{A counterexample for Remark~\ref{remark::weighted-maha}}\label{sec::D}

Naturally, we want to compare the optimal $A_{x}$ or $A_c$ with the one used in the Mahalanobis distance. Unfortunately, there could not be a universal result since $\nu_x(A_{x}^{\textnormal{opt}})$ or $\nu_c(A_{c}^{\textnormal{opt}})$ could be smaller than, equal to, or larger than $\nu_x (V_{\textnormal{haj},xx}^{-1})=1$ or $\nu_c(V^{-1}_{\textnormal{ht},cc}) = 1$, respectively, as shown in the example below.
\begin{example}
  Suppose that
  \[
    V_{\textnormal{haj},xx}=
    \begin{bmatrix}
      4      & \ \  \delta \\
      \delta & \ \  4
    \end{bmatrix},\quad V_{\textnormal{haj},\tau x}=  (1, 1 ).
  \]
  Then
  \begin{equation*}
    \nu_x (A_x^{\textnormal{opt}})/\nu_x (V_{\textnormal{haj},xx}^{-1}) = \left(\frac{4-\delta}{4+\delta}\right)^{1/2}.
  \end{equation*}
  When $\delta > 0$, i.e., there is a positive correlation between covariates, $\nu_x (A_x^{\textnormal{opt}})$ is smaller than $\nu_x (V_{\textnormal{haj},xx}^{-1})$. When $\delta < 0$, the opposite result holds.
\end{example}

\section{Proofs}\label{sec::E}



\subsection{Regularity conditions}\label{sec:regcon}
This section provides the  regularity conditions for  the theoretical results. These conditions are required for using the finite population central limit theorem. Let $|| \cdot ||_\infty$ denote the $\ell_\infty$ norm of a vector.


The regularity conditions below are used for deriving the asymptotic properties of $\tauht$ under  classic cluster  randomization.

\begin{condition}\label{cond:1}
  As $M \rightarrow \infty$, $ e_1 $ has a limit in $ (0,1)$. 
\end{condition}

\begin{condition}
  \label{cond:2}
  As $M \rightarrow \infty$, 
  \begin{enumerate}
    \item[(i)] the finite population variance $\varf  \{\tilde{Y}(z)\}$ $(z=0,1)$ and  covariance  $\covf  \{\tilde{Y}(0),\tilde{Y}(1)\}$ have finite limiting values, and the limit of $\vttht$ is positive;
    \item[(ii)] $\max_{z=0,1} \max _{1 \leq i \leq M} \{\tilde{Y}_{i\cdot}(z)-\bar{Y}(z) \}^2 / M \rightarrow 0$.
  \end{enumerate}
\end{condition}

Condition~\ref{cond:1} indicates that we have enough clusters assigned to both the treatment and control arms. Condition \ref{cond:2} restricts the moments of the potential outcomes.

The asymptotic properties of $\tauhaj$ under classic cluster randomization  require Condition~\ref{cond:1} and Condition~\ref{cond:3} below. 

\begin{condition}\label{cond:3}
  As $M \rightarrow \infty$, 
  \begin{enumerate}
    \item[(i)] the finite population variance, $\varf  \{ \tilde{\varepsilon }(z)\}$ $(z=0,1)$, and covariances, $\covf \{\tilde{\varepsilon }(0),\tilde{\varepsilon }(1)\}$ and $\covf \{\tilde{\varepsilon }(z),\tilde{\omega}\}$, have finite limiting values, and the limit of $\vtthaj$ is positive;
    \item[(ii)] $ M^{-1} \max _{1 \leq i \leq M}(\tilde{\omega}_{i}-1)^2  \to 0$ and $ M^{-1} \max_{z=0,1}  \max _{1 \leq i \leq M}\big\{\tdei(z)\big\}^2  \to 0$.
  \end{enumerate}
\end{condition}

The asymptotic properties of $\tauht$ in Theorems~\ref{thm:asymp2}--\ref{thm:optimdiag} and \ref{thm:wtedbetter-cluster-covariate}, Corollaries~\ref{coro:clusterbetter} and \ref{cor:varalphaortho-cluster-covariate}--\ref{cor:varalpha-tier-cluster-covariate}, and Propositions~\ref{prop:FCLT}--\ref{prop:concentration-symmetric-unimodality}
require Conditions~\ref{cond:1}--\ref{cond:2} and Condition~\ref{cond:4} below.



\begin{condition}\label{cond:4}
  As $M \rightarrow \infty$,  
  \begin{enumerate}
    \item[(i)] the finite population covariances, $\covf ({C})$ and $\covf  \{{C},\tilde{Y}(z)\}$ $(z=0,1)$, and $A_c$ have finite limiting values, and the limits of $\covf ({C})$ and $A_c$ are nonsingular;
    \item[(ii)]  $ M^{-1} \max _{1 \leq i \leq M} \|{{c}}_{i} \|_{\infty}^{2}  \rightarrow 0$.
  \end{enumerate}
\end{condition}

The asymptotic properties of $\tauhaj$ in Theorems~\ref{thm:asymp2}--\ref{thm:wtedbetter}, Corollaries~\ref{coro:clusterbetter}--\ref{cor:varalphatier}, and Propositions~\ref{prop:FCLT}--\ref{prop:concentration-symmetric-unimodality} require Conditions~\ref{cond:1}, \ref{cond:3}, and Condition~\ref{cond:5} below. 

\begin{condition}\label{cond:5}
  As $M \rightarrow \infty$, 
  \begin{enumerate}
    \item[(i)] the finite population  covariances,  $\covf (\tilde{{X}})$, $\covf  \{\tilde{{X}},\tilde{\varepsilon } (z)\}$ $(z=0,1)$, and $\covf  (\tilde{{X}},\tilde{\omega})$, and $A_x$ have finite limiting values, and the limits of $\covf (\tilde{{X}})$ and $A_x$ are nonsingular;
    \item[(ii)]  $M^{-1} \max _{1 \leq i \leq M}\|\tilde{{x}}_{i\cdot}\|_{\infty}^2  \rightarrow 0$.
  \end{enumerate}
\end{condition}



The asymptotic properties of $\tauadjht$ in Theorems~\ref{thm::regression-adjustment-under--cluster-rerandomization} and \ref{thm:tauhtadjasymp}--\ref{thm:interval-conservative}, Corollaries~\ref{coro:anamoreinfo}--\ref{cor:lessinfocovotho}, and Proposition~\ref{thm:mcorr}
require Conditions~\ref{cond:1}, \ref{cond:2}, \ref{cond:4}, and Condition~\ref{cond:6} below.

\begin{condition}\label{cond:6}
  As $M \rightarrow \infty$, 
  \begin{enumerate}
    \item[(i)] $M^{-1}\max_{1\leq i\leq M}\left\|{v}_i\right\|^2_{\infty} = o(1)$;
    \item[(ii)]  the finite population covariances, $\covf ({V})$, $\covf  \{{V},\tilde{Y}(z)\}$ ($z=0,1$), and $\covf  ({V},{C})$, have finite limiting values, and the limit of $\covf ({V})$ is nonsingular;
  \end{enumerate}
\end{condition}


The asymptotic properties of $\tauadjhaj$  in Theorems~\ref{thm::regression-adjustment-under--cluster-rerandomization} and \ref{thm:tauhtadjasymp}--\ref{thm:interval-conservative} and Proposition~\ref{thm:mcorr} require Conditions~\ref{cond:1}, \ref{cond:3}, \ref{cond:5}, and Condition~\ref{cond:7} below.

\begin{condition}
  \label{cond:7}
  As $M \rightarrow \infty$, 
  \begin{enumerate}
    \item[(i)] $N^{-1} \sum_{i=1}^M \sum_{j=1}^{n_i} {w}_{i j} \varepsilon _{ij}(z)$ ($z=0,1$) converges to a finite vector, and $N^{-1} \sum_{i=1}^M \sum_{j=1}^{n_i} {w}_{i j} {w}_{i j}^{\mathrm{T}}$ converges to a finite and
          invertible matrix;
    \item[(ii)] $\max_{1\leq i\leq M}\tilde{\omega}_i = o(M^{1/3})$, $N^{-1} \sum_{i=1}^M \sum_{j=1}^{n_i} \left\|{w}_{i j}\right\|_{\infty}^{4}=O(1) $,  $M^{-1}\max_{1\leq i\leq M}\left\|{\tilde{{w}}}_{i\cdot}\right\|^2_{\infty} = o(1)$, $N^{-1} \sum_{i=1}^M \sum_{j=1}^{n_i} \{Y_{i j}(z)\}^{4}=O(1) \ (z=0,1)$;
    \item[(iii)] the finite population covariances, $\covf (\tilde{{W}})$, $\covf  \{\tilde{{W}},\tilde{{\varepsilon} } (z)\}$ ($z=0,1$), $\covf  \{\tilde{{W}},\tilde{\omega}\}$, and $\covf  (\tilde{{W}}, \tilde{{X} } )$, have finite limiting values, and the limit of $\covf(\tilde{{W}})$ is nonsingular.
  \end{enumerate}
\end{condition}


\subsection{Some preliminary results on weak convergence}

Before proving the main results, we introduce some preliminary results  on weak convergence derived by \cite{Li9157}. Consider the combination $(\hat{\tau},\hat{\tau}_{\star})$ as the estimator of the average treatment effects for the outcomes and covariates, for example, $(\tauht,\taucht)$ and $(\tauhaj,\tauxhaj)$. Consider a balance criterion $\phi(M^{1/2}\hat{\tau}_{\star}, A)$ which is a binary function, and $A$ is the matrix used to conduct cluster rerandomization.  Denote the event that a treatment assignment is acceptable under the balance criterion $\phi( M^{1/2}\hat{\tau}_{\star}, A)$ as  $\phi(M^{1/2}\hat{\tau}_{\star}, A)=1$,  and let  $\mathfrak{G}=\left\{{\mu}: \phi\left({\mu}, A\right)=1\right\}$  be the acceptance region for  $M^{1/2}\hat{\tau}_{\star}$. We write $A>0$ if $A$ is strictly positive definite.
Suppose that Condition~\ref{cond:criterion} below on $\phi(\cdot, \cdot)$ holds.
\begin{condition}
  \label{cond:criterion}
  \begin{enumerate}
    \item[(i)] $\phi(\cdot, \cdot)$ is almost surely continuous;
    \item[(ii)] if ${T}_\star \sim \mathcal{N}\left( 0 , V_{\star\star}\right)$, then $\textnormal{pr}\left\{\phi\left({T}_\star, A\right)=1\right\}>0$ for any $A>0$, and $\var\left\{{T}_\star \mid \phi\left({T}_\star, A\right)=1\right\}$ is a continuous function of
          $V_{\star \star} ;$

    \item[(iii)] $\phi\left({\mu}, A\right)=\phi\left(-{\mu}, A\right)$, for all ${\mu}$ and $A>0.$
  \end{enumerate}

\end{condition}

In Condition~\ref{cond:criterion}, (i) and (ii) impose smoothness constraints on $\phi$ and prevent the acceptance region from being a set of measure zero; (iii) is a symmetry consideration to ensure that relabel the treatment and control units does not change the  balance criterion. Cluster rerandomization schemes based on the Mahalanobis distances and weighted Euclidean distances satisfy Condition~\ref{cond:criterion}. We denote the variance of $M^{1/2}(\hat{\tau},\hat{\tau}_{\star})$ as $V$, and we use the same notation to represent its limit:
$$
  {V}=\begin{bmatrix}
    V_{\tau\tau}  & V_{\tau \star} \\
    V_{\star\tau} & V_{\star\star}
  \end{bmatrix}.
$$
Assume that $V$ is strictly positive definite. Let $A_{\infty}$ be the limit of $A$, and $\mathfrak{G}_{\infty}=\left\{{\mu}: \phi\left({\mu}, {A}_{\infty}\right)=1\right\}$ be the limit of $\mathfrak{G} $.  Proposition~\ref{prop:weakconv} below is a direct result of Proposition A1 in \cite{Li9157}.

\begin{proposition}
  \label{prop:weakconv}
  Under Condition~\ref{cond:criterion}, as $M \rightarrow \infty$, if
  $$
    M^{1/2} \left(\begin{array}{c}
        \hat{\tau}-\tau \\
        \hat{\tau}_{\star}
      \end{array}\right)  \mathrel{\ \dot{\sim}\ } \mathcal{N}(0,V),
  $$
  then
  \begin{gather*}
    M^{1/2} \left(\begin{array}{c}
        \hat{\tau}-\tau \\
        \hat{\tau}_{\star}
      \end{array}\right)\biggm | M^{1/2} \hat{\tau}_{\star} \in \mathfrak{G} \mathrel{\ \dot{\sim}\ } \left(\begin{array}{c}
        {T} \\
        {{T}}_{\star}
      \end{array}\right)\biggm | {{T}}_{\star} \in \mathfrak{G},\\
  \end{gather*}
  where $({T},{{T}_{\star}^\top})^{\top}\sim \mathcal{N}({0},{V})$.
\end{proposition}

\subsection{Proof of Proposition~\ref{prop:FCLT}}
We need the following Lemma from \cite{su2021modelassisted}.
\begin{lemma}
  \label{lem:tauhaj-approx}
  Under Conditions \ref{cond:1} and \ref{cond:3},
  \begin{align*}
    M^{1/2}\Bigl[\tauhaj-\Big\{M_1^{-1}\sum_{i=1}^N Z_i\tdei-M_0^{-1}\sum_{i=1}^N (1-Z_i)\tdei\Big\}\Bigr]=o_{\textup{p}}(1).
  \end{align*}
\end{lemma}
\citet{su2021modelassisted} proved Lemma~\ref{lem:tauhaj-approx}  under slightly stronger conditions. After carefully examining their proof, we found that it can be easily generalized to the case of our weaker conditions, so we omit the proof of Lemma~\ref{lem:tauhaj-approx}.

\begin{proof}[of Proposition~\ref{prop:FCLT}]
  Applying the vector-form finite population central limit theorem developed by \cite{li2017general} to the scaled cluster-total potential outcomes $\tdyi(z)$, $z=0,1$, and cluster-level covariates ${c}_i$, we have, under Conditions~\ref{cond:1}, \ref{cond:2}, and \ref{cond:4},
  \begin{equation}\label{eq:tauhtasymp}
    M^{1/2}
    \left(
    \begin{array}{c}
        \tauht  -\tau \\
        \taucht
      \end{array}
    \right)
    \mathrel{\ \dot{\sim}\ }  \mathcal{N}  \left( 0 ,
    \begin{bmatrix}
        \vttht & \vtcht \\
        \vctht & \vccht
      \end{bmatrix} \right). \nonumber
  \end{equation}


We could get a similar result for the joint asymptotic distribution of  $(\tauhaj,\tauxhaj^\top) ^\top$  after replacing $\tilde{Y}_{i\cdot}(z)-\bar{Y}(z)$ by $\tdei(z) = \tilde{Y}_{i\cdot}(z)-\tilde{\omega}_i\bar{Y}(z)$ and applying Lemma~\ref{lem:tauhaj-approx}. Specifically, applying Lemma~\ref{lem:tauhaj-approx} to each element of $M^{1/2}
    (
    \tauhaj  -\tau,
    \tauxhaj ^\top  )^{\top}$, we have, under Conditions~\ref{cond:1} and \ref{cond:3},
  \[
    M^{1/2}
    \left(
    \begin{array}{c}
        \tauhaj  -\tau \\
        \tauxhaj
      \end{array}
    \right)
    = M^{1/2}
    \Bigg(\begin{gathered}
        \scalebox{1}{$M_1^{-1}\sum\nolimits_{i=1}^M Z_i\tdei(1)-	M_0^{-1}\sum\nolimits_{i=1}^M (1-Z_i)\tdei(0)$} \\
        \scalebox{1}{$M_1^{-1}\sum\nolimits_{i = 1}^M Z_i\tdxi -M_0^{-1}\sum\nolimits_{i=1}^M   (1-Z_i)\tdxi$}
      \end{gathered} \Bigg) + o_{\textup{p}} (1).
  \]
  Under Conditions~\ref{cond:1}, \ref{cond:3}, and \ref{cond:5}, applying the vector-form finite population central limit theorem to $\tdei(z)$,  $z=0,1$, and $\tdxi$, we have
  \[M^{1/2}\Bigg(\begin{gathered}
        \scalebox{1}{$M_1^{-1}\sum\nolimits_{i=1}^MZ_i\tdei(1)-M_0^{-1}\sum_{i=1}^M (1-Z_i)\tdei(0)$} \\
        \scalebox{1}{$M_1^{-1}\sum_{i = 1}^M   Z_i\tdxi-M_0^{-1}\sum_{i=1}^M(1-Z_i)\tdxi$}
      \end{gathered}\Bigg)
    \mathrel{\ \dot{\sim}\ }\mathcal{N}\left(0,\Bigg[\begin{aligned}
        \vtthaj  \  & \vtxhaj \\
        \vxthaj  \  & \vxxhaj
      \end{aligned}\Bigg]\right).\]
  By Slutsky's theorem, the conclusion holds.
\end{proof}




\subsection{Proof of Theorem~\ref{thm:asymp2}}
Proposition~\ref{cor:weakconv} below follows from Propositions~\ref{prop:FCLT} and \ref{prop:weakconv}.
\begin{proposition}
  \label{cor:weakconv}
  Under Conditions~\ref{cond:1}, \ref{cond:3}, and \ref{cond:5},
  \begin{gather*}
    M^{1/2}
    \begin{pmatrix}
      \tauhaj-\tau \\
      \hat{\tau}_{\textnormal{haj},x}
    \end{pmatrix}
    \biggm | M^{1/2} \hat{\tau}_{\textnormal{haj},x} \in \mathfrak{G}_x
    \mathrel{\ \dot{\sim}\ }
    \begin{pmatrix}
      {T} \\
      T_{x}
    \end{pmatrix}
    \biggm | T_{x} \in \mathfrak{G}_{x},
  \end{gather*}
  and under Conditions~\ref{cond:1}, \ref{cond:2}, and \ref{cond:4},
  \begin{gather*}
    M^{1/2}
    \left(
    \begin{array}{c}
        \hat{\tau}_{\textnormal{ht}}-\tau \\
        \hat{\tau}_{\textnormal{ht},c}
      \end{array}
    \right)
    \biggm | M^{1/2} \hat{\tau}_{\textnormal{ht},c} \in \mathfrak{G}_c
    \mathrel{\ \dot{\sim}\ }
    \left(
    \begin{array}{c}
        {T} \\
        {{T}}_{c}
      \end{array}
    \right)
    \biggm | {{T}}_{c} \in \mathfrak{G}_{c},
  \end{gather*}
  where $({T},{{T}_{x}^\top})^{\top}\sim \mathcal{N}({0},{V}_{\textnormal{haj}})$ and $({T},{{T}_{c}^\top})^{\top}\sim \mathcal{N}({0},{V}_{\textnormal{ht}})$ with
$${V}_{\textnormal{haj}} = \Bigg[\begin{aligned}
        \vtthaj  \  & \vtxhaj \\
        \vxthaj  \  & \vxxhaj
      \end{aligned}\Bigg], \quad  {V}_{\textnormal{ht}}  =
       \begin{bmatrix}
        \vttht & \vtcht \\
        \vctht & \vccht
      \end{bmatrix}.
 $$ 
\end{proposition}

\begin{proof}[of Theorem~\ref{thm:asymp2}]

  The asymptotic distribution of $M^{1/2}(\tauht  -\tau) \mid \mc$ is a direct result of applying Theorem~1 of \cite{Li9157} to $\tdyi(z)$, $z=0,1$, and ${c}_i$.

  To derive the asymptotic distribution of $ M^{1/2}(\tauhaj  -\tau) \mid \mx$, similar to the proof of Proposition~\ref{lem:tauhaj-approx}, we have $\{ M^{1/2}(\tauhaj-\tau),M^{1/2}\tauxhaj,A_x \}$ has the same asymptotic distribution as $\{ T,{T}_x, A_x \}$ with $A_x = \vxxhaj^{-1}$.
  Therefore,  $M^{1/2}(\tauhaj-\tau) \mid M \tauxhaj^{\top} A_x \tauxhaj \leq a$ has the same asymptotic distribution as $T \mid {T}_x^{\top} A_x {T}_x \leq a$. Applying Theorem~1 of \cite{Li9157} to $\tdei(z)$, $z=0,1$, and $\tdxi$, we have
  \[
    T \mid {T}_x^{\top} A_x{T}_x\leq a  \mathrel{\ \dot{\sim}\ } \vttht ^{1/2}\{(1-R_{ x}^{2})^{1/2}  \epsilon+R_{ x}  L_{K, a}\}.
  \]
  The conclusion follows immediately.

  %
  %

\end{proof}

\subsection{Proof of Corollary \ref{coro:clusterbetter}}
\begin{proof}[of Corollary~\ref{coro:clusterbetter}]
  For the cluster-level covariates ${c}_i=(n_i,\tdxi)$,  we have
  \begin{gather*}
    \vttht  = ( e_1  e_0 )^{-1}\varf \big\{ e_1 \tilde{Y}(0)+ e_0 \tilde{Y}(1)\big\},\\
    \vttht  R^2_{ c} = ( e_1  e_0 )^{-1}\covf \big\{ e_1 \tilde{Y}(0)+ e_0 \tilde{Y}(1),{C}\big\}\covf ({C})^{-1}\covf \big\{{C}, e_1 \tilde{Y}(0)+ e_0 \tilde{Y}(1)\big\}.
  \end{gather*}
  Therefore, $e_1 e_0  \vttht (1-R_{ c}^{2} )$ is  the variance of the residuals when we regress $ e_1 \tilde{Y}(0)+ e_0 \tilde{Y}(1)$ on $C$ with intercept, i.e.,
  $$
    e_1 e_0  \vttht \left (1-R_{ c}^{2}\right) = \min_{\beta\in \mathbb{R}^{K+1},\ \beta_0\in \mathbb{R}}\varf  \left\{ e_1 \tilde{Y}(0)+ e_0 \tilde{Y}(1)-{1}_M\beta_0-(N/M\tilde{\omega},\tdx)\beta\right\},
  $$
  where ${1}_M$ is an $M$-dimensional all-one vector. Similarly, we have
  \begin{eqnarray*}
    e_1 e_0 \vtthaj \left (1-R_{ x}^{2}\right) &=&
    \min_{\beta\in \mathbb{R}^K,\ \beta_0\in \mathbb{R}}\varf  \left\{ e_1 \tde (0)+ e_0 \tde (1)-{1}_M \beta_0-\tdx\beta\right\} \\
    &=&\min_{\beta\in \mathbb{R}^K,\ \beta_0 \in \mathbb{R}}\varf  \left[ e_1 \tilde{Y}(0)+ e_0 \tilde{Y}(1)-\{ e_1 \bar{Y}(0)+ e_0 \bar{Y}(1)\}\tilde{\omega}-{1}_M\beta_0-\tdx\beta\right],
  \end{eqnarray*}
  where the second equality follows from
  $
    \tdei (z) = \tdyi(z)-\tilde{\omega}_i\bar{Y}(z).
  $
  Therefore,
  $$
    \vttht \left (1-R_{ c}^{2}\right) \leq \vtthaj \left (1-R_{ x}^{2}\right),
  $$
  which completes the proof.
\end{proof}

\subsection{Proof of Theorem~\ref{thm:asymquadra}}
\begin{proof}[of Theorem~\ref{thm:asymquadra}]
  By Proposition~\ref{prop:FCLT},
  $$
    M^{1/2} ( \tauht-\tau, \taucht^\top ) \mathrel{\ \dot{\sim}\ } \left(T, {T}_c^\top \right) ,\quad M^{1/2} ( \tauhaj-\tau, \tauxhaj^\top) \mathrel{\ \dot{\sim}\ } \left(T, {T}_x^\top \right),
  $$
  where $(T,{T}_c^\top)^\top \sim \mathcal{N}(0,V_{\textnormal{ht}})$ and $(T,{T}_x^\top)^\top\sim \mathcal{N}(0,V_{\textnormal{haj}})$.
  We only give the proof for cluster rerandomization with individual-level covariates, as the proof for cluster rerandomization with cluster-level covariates is similar.
  By Proposition~\ref{cor:weakconv}, we have
  \begin{eqnarray}
    M^{1/2}(\hat{\tau}_{\textnormal{haj},x}-\tau)\mid \mathcal{D}_{x}(A_x)
    & = &
    M^{1/2}(\hat{\tau}_{\textnormal{haj},x}-\tau)\mid ( M^{1/2}\hat{\tau}_{\textnormal{haj},x})^{\top}A_x(M^{1/2}\hat{\tau}_{\textnormal{haj},x})\leq a \nonumber \\
    & \mathrel{\ \dot{\sim}\ } & T \mid {T}_x^{\top}A_x{T}_x\leq a. \nonumber
  \end{eqnarray}
  We decompose $T$ as 
  \[ T = T-\vtxhaj \vxxhaj^{-1}{T}_x + \vtxhaj \vxxhaj^{-1}{T}_x.\]
  Note that $T-\vtxhaj \vxxhaj^{-1}{T}_x$ is independent of ${T}_x$, and thus, independent of $\vtxhaj \vxxhaj^{-1}{T}_x$ and ${T}_x^{\top} A_x{T}_x$. Therefore, 
  $$ T-\vtxhaj \vxxhaj^{-1}{T}_x \mid {T}_x^{\top} A_x{T}_x \leq a \sim T-\vtxhaj \vxxhaj^{-1}{T}_x.$$
  By the definition of $R_x^2$, we have $R_x^2 = \vtxhaj\vxxhaj^{-1}\vxthaj/\vtthaj$, and therefore
  $$\var ( T-  \vtxhaj \vxxhaj^{-1}{T}_x ) = \vtthaj-\vtxhaj \vxxhaj^{-1}\vxthaj = (1-R_x^2)\vtthaj.$$
  Thus,
  $$T-\vtxhaj \vxxhaj^{-1}{T}_x \sim \vtthaj^{1/2}(1-R_x^2)^{1/2}\epsilon.$$
  Since $\vxxhaj^{-1/2}{T}_x$ is a $K$-dimensional standard normal random vector,
  \begin{align*}
         & \vtxhaj \vxxhaj^{-1}{T}_x \mid {T}_x^{\top}A_x{T}_x\leq a                                                                                          \\
    \sim & \vtxhaj\vxxhaj^{-1/2} ( \vxxhaj^{-1/2}{T}_x ) \mid ( {T}_x^{\top}\vxxhaj^{-1/2} ) ( \vxxhaj^{1/2}A_x\vxxhaj^{1/2} ) ( \vxxhaj^{-1/2}{T}_x ) \leq a \\
    \sim & \vtxhaj\vxxhaj^{-1/2} {\eta} \mid {\eta}^{\top}  \vxxhaj^{1/2}A_x\vxxhaj^{1/2} {\eta}\leq a                                                        \\
    \sim & \vtthaj^{1/2} (   R_x{\mu}_x^{\top}{\eta} ) \mid
    {\eta}^{\top}\vxxhaj^{1/2}A_x\vxxhaj^{1/2} {\eta}\leq a,
  \end{align*}
  where ${\eta}$ is a $K$-dimensional standard normal random vector and is independent of $\epsilon$ and
  $${\mu}_x^{\top}=(\vtxhaj\vxxhaj^{-1} \vxthaj)^{-1/2}\vtxhaj \vxxhaj^{-1/2}.$$
  Therefore,
  \[ \vtthaj^{1/2} \big\{  (1-R_x^2 )^{1/2}\epsilon +   R_x{\mu}_x^{\top}{\eta} \mid
    {\eta}^{\top}{\vxxhaj^{1/2}A_x\vxxhaj^{1/2}}_x{\eta}\leq a  \big\} \mathrel{\ \dot{\sim}\ } T \mid {T}_x^{\top}A_x{T}_x\leq  a,\]
  which implies that
  \[
    M^{1/2}(\tauhaj-\tau) \mid \mathcal{D}_{x}(A_x)
    \mathrel{\ \dot{\sim}\ }  \vtthaj^{1/2} \big\{  (1-R_x^2 )^{1/2}\epsilon +   R_x{\mu}_x^{\top}{\eta} \mid
    {\eta}^{\top}\vxxhaj^{1/2}A_x\vxxhaj^{1/2}{\eta}\leq a \big\} .
  \]

\end{proof}

\subsection{Proof of Proposition~\ref{prop:concentration-symmetric-unimodality}}

First, we introduce some useful lemmas. Lemma~\ref{lem9157} below is from \citet[][Lemma A8]{Li9157}.

\begin{lemma}\label{lem9157}
  If $A$ and $B$ are independent, and are both symmetric around zero and unimodal, then $A + B$ is symmetric around zero and unimodal.
\end{lemma}

Recall that the asymptotic distributions in Theorem~\ref{thm:asymquadra} are composed of two independent components: a normal random variable and a truncated normal random variable. To simplify the presentation, we say that a random variable is symmetric unimodal if it is symmetric around zero and unimodal. Clearly, the normal random variable is symmetric  unimodal.  To prove the symmetric unimodality of the truncated normal component, we follow the proof of  \citet[][Proposition 2]{Li2020factorial} using properties of  symmetric unimodality of random vectors. The definition of symmetric unimodal random vector follows from \citet{dharmadhikari1988unimodality} as an extension of univariate case. 
\begin{definition}
For a set $\mathcal{B}$ of distributions on $\mathbb{R}^{K}$, we say that $\mathcal{B}$ is closed convex if it satisfies two properties: (i) for any distributions $\nu_{1}, \nu_{2} \in \mathcal{B}$ and for any $\lambda \in(0,1)$, the distribution $(1-\lambda) \nu_{1}+\lambda \nu_{2} \in \mathcal{B}$, and (ii) a distribution $\nu \in \mathcal{B}$ if there exist a sequence of distributions in $\mathcal{B}$ converging weakly to $\nu$.
\end{definition}

 For any set $\mathcal{C}$ of distributions, the closed convex hull of $\mathcal{C}$ is the smallest closed convex set containing $\mathcal{C}$. A compact convex set in Euclidean space $\mathbb{R}^{K}$ is called a convex body if it has a nonempty interior. A set $\mathcal{K} \subset \mathbb{R}^{K}$ is symmetric if $\mathcal{K}=\{-a: a \in \mathcal{K}\}$. 
 
 \begin{definition}
 A distribution on $\mathbb{R}^{K}$ is symmetric unimodal if it is in the closed convex hull of $\mathcal{U}$, where $\mathcal{U}$ is the set of all uniform distributions on symmetric convex bodies in $\mathbb{R}^{K}$.
 \end{definition}

Lemmas~\ref{lem:symmetric-unimodal} and \ref{lem:log-concave-density} below are from \citet[][Lemma A6]{Li2020factorial} and \citet[][Lemma A8]{Li2020factorial} respectively.

\begin{lemma}\label{lem:symmetric-unimodal}
  If $\psi\in \mathbb{R}^{K}$ is a symmetric unimodal random vector, then for any non-random vector $b\in \mathbb{R}^{K}$, $b^\top\psi$ is a symmetric unimodal random variable.
\end{lemma}

\begin{lemma}\label{lem:log-concave-density}
  If a random vector in $\mathbb{R}^{K}$ has a log-concave density, then it is symmetric unimodal.
\end{lemma}

Lemma~\ref{lem:unimodal-symmetric-truncated-normal} below is  new  which shows the symmetric unimodality of a truncated normal random vector.
\begin{lemma}\label{lem:unimodal-symmetric-truncated-normal}
  The truncated normal random variable $\rho b^\top\eta\mid \eta^\top A\eta\leq a$ is symmetric unimodal, where ${\eta} \sim \mathcal{N}(0,I_K)$, $\rho$ is a non-random scalar, ${b}$ is a non-random vector, and $A$ is a non-random symmetric positive definite matrix.
\end{lemma}
\begin{proof}[of Lemma~\ref{lem:unimodal-symmetric-truncated-normal}]
  By Lemma~\ref{lem:symmetric-unimodal}, it suffices to show that $\eta\mid \eta^\top A\eta\leq a$ is symmetric unimodal. The density function of $\eta\mid \eta^\top A\eta\leq a$ is
  \begin{align*}
    g({x})=\frac{1\left\{x^{\top}A x \leq a\right\}}{\textnormal{pr} \left(\eta^\top A\eta\leq a\right)}(2 \pi)^{-K / 2} \exp \left(-x^{\top} x / 2\right),
  \end{align*}
  so
  \begin{align*}
    \log  g(x) = \begin{cases}-\log \left\{ \textnormal{pr} \left(\eta^\top A\eta\leq a\right)\right\}-(K / 2) \log (2 \pi)-x^\top x / 2, & \text{ if } x^\top A x \leq a, \\ -\infty, & \text { otherwise. }\end{cases}
  \end{align*}
  It is straightforward to show that $\log g(x)$ is concave. By Lemma~\ref{lem:log-concave-density},
  $\eta\mid \eta^\top A\eta\leq a$ is symmetric unimodal.
\end{proof}

Now we can prove the first part of Proposition~\ref{prop:concentration-symmetric-unimodality}.
\begin{proof}[of Proposition~\ref{prop:concentration-symmetric-unimodality}]
  We only give the proof for cluster rerandomization with individual-level covariates, as the proof for cluster rerandomization with cluster-level covariates is similar. By Theorem~\ref{thm:asymquadra}, the asymptotic distribution of $M^{1/2}(\tauhaj-\tau) \mid M\tauxhaj^{\top}A_x\tauxhaj\leq a$ is $A+B$ with
  $$
    A = \vtthaj^{1/2}\left (1- R_x^2\right)^{1/2}\epsilon, \quad B=\vtthaj^{1/2}{R}_x {\mu}_x^{\top}{\eta}\mid {\eta}^{\top}\vxxhaj^{1/2}A_x\vxxhaj^{1/2}{\eta}\leq a,
  $$
  where $\epsilon$, $\eta_k \ (k=1,\ldots,K)$ are independent standard normal random variables and $\eta = (\eta_1,\ldots,\eta_K)^\top$. $A$ and $B$ are independent and symmetric unimodal where the symmetric unimodality of $B$ follows directly from Lemma~\ref{lem:unimodal-symmetric-truncated-normal}. 
  Thus, by Lemma~\ref{lem9157}, $A+B$ is symmetric unimodal.

Next, we prove the second part of Proposition~\ref{prop:concentration-symmetric-unimodality}. We need the following two Lemmas. The first one is from  \citet[][Theorem 2.1]{gupta1972inequalities}.

\begin{lemma}
  \label{lem:sapp3}
  Let
  $$
    \Sigma=\left [\begin{array}{ll}
        \Sigma_{11} & \Sigma_{12}  \\
        \Sigma_{21} & \sigma_{p p}
      \end{array}\right]
  $$
  be a $p \times p$ positive definite matrix with $\Sigma_{11}$ being a $(p-1) \times (p-1)$ matrix. Let ${x}=$ $\left (x_{1}, \cdots, x_{p}\right)^{\top}$ be a random vector with density function $\left\lvert \Sigma_{\lambda}\right \rvert^{-1 / 2} f\left ({x} \Sigma_{\lambda}^{-1} {x}^{\top}\right),$ where
  $$
    \Sigma_{\lambda}=\left [\begin{array}{ll}
        \Sigma_{11}         & \lambda \Sigma_{12} \\
        \lambda \Sigma_{21} & \sigma_{p p}
      \end{array}\right], \quad 0 \leq \lambda \leq 1.
  $$
  If $E$ is a convex symmetric set in $R^{p-1}$, then ${\operatorname{pr}}\big\{\left (x_{1}, \cdots, x_{p-1}\right) \in E,\left\lvert x_{p}\right \rvert \leq h\big\}$ is non-decreasing in $\lambda.$
\end{lemma}

The second one is new.

\begin{lemma}
  \label{lem:sapp4}
  ${\operatorname{pr}}\bigl\{\lvert (1-\rho^2)^{1/2}\epsilon+\rho {b}^{\top}{\eta} \rvert<c\mid {\eta}^{\top}A{\eta}\leq a\bigr\}$ is a non-decreasing function of $\rho$, where $\epsilon\sim \mathcal{N}(0,1)$ and ${\eta} \sim \mathcal{N}(0,I_K)$ are independent, and ${b}$ is a non-random vector satisfying ${b}^{\top} {b} = 1$.
\end{lemma}

\begin{proof}[of Lemma~\ref{lem:sapp4}]
  Simple calculation gives
  $$
   \begin{pmatrix}
        (1-\rho^2)^{1/2}\epsilon+\rho {b}^{\top}{\eta}  \\
    A^{1/2}{\eta}
   \end{pmatrix}\sim \mathcal{N}\left( 0,\begin{bmatrix}
        1                      & \rho b^\top A^{1/2} \\
        \rho A^{1/2}b & A
      \end{bmatrix}\right).
  $$
  By Lemma~\ref{lem:sapp3}, ${\operatorname{pr}}\big\{\lvert (1-\rho^2)^{1/2}\epsilon+\rho {b}^{\top}{\eta} \rvert<c,  {\eta}^{\top}A{\eta}\leq a\big\}$ is a non-decreasing function of $\rho$. The same is true for the conditional probability ${\operatorname{pr}}\big\{\lvert (1-\rho^2)^{1/2}\epsilon+\rho {b}^{\top}{\eta} \rvert<c\mid {\eta}^{\top}A{\eta}\leq a\big\}$.
\end{proof}

The second half of Proposition~\ref{prop:concentration-symmetric-unimodality} follows directly from Lemma~\ref{lem:sapp4}.

\end{proof}


\subsection{Proof of Theorem~\ref{thm:varalpha}}

Before proving Theorem~\ref{thm:varalpha}, we introduce two useful lemmas.

\begin{lemma}
  \label{lem:sapp9} Let $s_K={2 \pi^{K/2}} / {\Gamma (K/2)}$ be the surface of a $K$-dimensional unit sphere and $B(a)  = \left\{{x}=(x_1,\ldots,x_K)\mid \|{x}\|_2^2 \leq a^2\right\}$. Then
  \[\int_{B(a) } x_1^2~\textup{d}x_1\cdots \textup{d}x_K=\frac{1}{(K+2)K}a^{K+2}s_K. \]
\end{lemma}
\begin{proof}[of Lemma~\ref{lem:sapp9}]
  By symmetry,
  $$
    \int_{B(a) } x_1^2~\textup{d}x_1\cdots \textup{d}x_K = K^{-1} \int_{B(a) }  (x_1^2 + \cdots + x_K^2) ~\textup{d}x_1\cdots \textup{d}x_K .
  $$
  Because the integrand on the right-hand side is the radius, the above formula reduces to
  $$
    \frac{s_K} { K } \int_{0}^a r^2 r^{K-1} \textup{d}r =    \frac{1}{(K+2)K}a^{K+2}s_K.  
  $$
  \end{proof}
  %


\begin{lemma}
  \label{lem:sapp10}
  Under Conditions~\ref{cond:1}, \ref{cond:3}, and \ref{cond:5}, the threshold $a$ in the cluster rerandomization scheme $\mathcal{D}_x(A_x)$ satisfies
  \begin{align*}
    a  & = (2\pi) \det( \vxxhaj )^{1/K}\det(A_x)^{1/K}\left(\frac{s_K}{K}\right)^{-2/K}\alpha^{2/K}+o(\alpha^{2/K}).
  \end{align*}
  Under Conditions~\ref{cond:1}, \ref{cond:2}, and \ref{cond:4}, the threshold $a$ in the cluster rerandomization scheme $\mathcal{D}_c(A_c)$ satisfies
  \begin{align*}
a  &= (2\pi)\det( \vccht )^{1/K} \det(A_c)^{1/K}\left(\frac{s_K}{K}\right)^{-2/K}\alpha^{2/K}+o(\alpha^{2/K}).
  \end{align*}
\end{lemma}

\begin{proof}[of Lemma~\ref{lem:sapp10}]

  We only give the proof for cluster rerandomization with individual-level  covariates, as the proof for cluster rerandomization with cluster-level covariates is similar.
  Denote the asymptotic distribution of $(M^{1/2}(\tauhaj-\tau), M^{1/2}\tauxhaj^\top)$ as $\left({T}_{\infty}, {{T}}_{x,\infty}^\top \right) $, where ${{T}}_{x,\infty}$ is a $K$-dimensional normal random vector. By the property of normal random vector, we have
  $${{T}}_{x,\infty}^{\top}{A}_x{{T}}_{x,\infty} \sim \sum_{k=1}^{K}\lambda_k \eta_k^2 ,$$
  where $\lambda_1,\dots,\lambda_K$ are the eigenvalues of $\vxxhaj^{1/2}{A}_x\vxxhaj^{1/2}$, and $\eta_1,\dots, \eta_K$  are independent standard normal random variables.
  Let $\phi (x)$ be the density function of a standard normal random variable, then
  \begin{align*}
    {\operatorname{pr}}\Big (\sum_{k=1}^{K}\lambda_k \eta_k^2 \leq a \Big)
     & =\int_{\sum_{k=1}^{K} \lambda_k x_k^2 \leq a}~\prod_{k=1}^K\phi (x_k)~\textup{d}x_1\cdots \textup{d}x_K                                                           \\
     & =(2\pi)^{-K/2}\Big(\prod_{k=1}^K \lambda_k \Big)^{-1/2}\int_{B (a^{1/2})} \exp\Big (2^{-1} {\sum_{k=1}^{K} x_k^2/\lambda_k}\Big)\textup{d}x_1\cdots \textup{d}x_K \\
     & = (2\pi)^{-K/2}\Big(\prod_{k=1}^K \lambda_k\Big)^{-1/2}\int_{B (a^{1/2})} \{ 1+O(r^2) \} \textup{d}x_1\cdots \textup{d}x_K.
  \end{align*}
  The last line follows from $\min_{k=1,\dots,K} \lambda_k>0$ and the Taylor expansion, 
  \[
    \exp\Big ( 2^{-1} \sum_{k=1}^{K} x_k^2/\lambda_k \Big ) = 1+O\Big  (2^{-1} \sum_{k=1}^{K} x_k^2/\lambda_k\Big ) = 1+O (r^2),\quad r^2 = \sum_{k=1}^K x_k^2.
  \]
  By the volume formula of a $K$-dimensional ball,
  \begin{equation}\label{eq:**}
    \alpha={\operatorname{pr}}\Big (\sum_{k=1}^{K}\lambda_k \eta_k^2 \leq a \Big) = (2\pi)^{-K/2}\Big  (\prod_{k=1}^K \lambda_k\Big)^{-1/2} \frac{s_K}{K}a^{K/2} + O(a^{K/2+1}).
  \end{equation}
  As $\lambda_k,~k=1,\ldots,K$, are eigenvalues of $\vxxhaj^{1/2}{A}_x\vxxhaj^{1/2}$, we have
  \begin{align*}
    \prod_{k=1}^K \lambda_k = \det \big( \vxxhaj^{1/2}{A}_x\vxxhaj^{1/2} \big)  = \det( \vxxhaj ) \det ( {A}_x ).
  \end{align*}
  Substitute into \eqref{eq:**} and reorganize the terms, we get the desired result.

\end{proof}

\begin{proof}[of Theorem~\ref{thm:varalpha}]
  We only give the proof for cluster rerandomization with individual-level  covariates, as the proof for cluster rerandomization with cluster-level covariates is similar.  As $\left({T}_{\infty}, {{T}}_{x,\infty}^\top \right) $ is the asymptotic distribution of $(M^{1/2}(\tauhaj-\tau), M^{1/2}\tauxhaj^\top)$, we have
  $$\vara\{M^{1/2}(\tauhaj-\tau)\mid  M\tauxhaj^{\top}A_x\tauxhaj\leq a \} =\var({T}_{\infty} \mid {{T}}_{x,\infty}^\top{A}_x{{T}}_{x,\infty}).$$
  Referring to the proof of Theorem~\ref{thm:asymquadra},
  $$
    {T}_{\infty} \mid {{T}}_{x,\infty}^\top{A}_x{{T}}_{x,\infty} \leq a \sim V_{\textnormal{haj},\tau \tau}^{1/2} \big\{  (1-R_{x}^2 )^{1/2}\epsilon +   R_{x}\mu_x^{\top}{\eta} \mid
    {\eta}^{\top}V_{\textnormal{haj},xx}^{1/2}A_{x}V^{1/2}_{\textnormal{haj},xx}{\eta}\leq a \big\},
  $$
  where
  $${\mu}_x^{\top}=(\mu_1,\ldots,\mu_K)=\vtxhaj\vxxhaj^{-1/2}/(\vtxhaj\vxxhaj^{-1}\vxthaj)^{1/2}, \quad {\eta}=(\eta_{1},\ldots,\eta_{K})^{\top},$$
  and $\epsilon,\eta_1, \ldots,\eta_K$ are independent standard normal random variables.

  We perform the eigenvalue decomposition as follows:
  $
    V_{\textnormal{haj},xx}^{1/2}A_{x}V^{1/2}_{\textnormal{haj},xx} = P_{x} {\Lambda}_x P_{x}^{\top}
  $, where ${\Lambda}_x = \operatorname{diag}(\lambda_1,\ldots,\lambda_K)$. Then $\mu_x^{\top}{\eta}\mid
    {\eta}^{\top}V_{\textnormal{haj},xx}^{1/2}A_{x}V^{1/2}_{\textnormal{haj},xx}{\eta}\leq a$ has the same distribution as
  $
    \mu_x^{\top}P_x{\eta} \mid
    {\eta}^{\top} {\Lambda}_x{\eta}\leq a
  $.
  Let $(\zeta_1,\ldots,\zeta_K) = \mu_x^{\top}P_x $.

  First, we derive the asymptotic variance of $\sum_{k=1}^{K} \zeta_k\eta_k \mid \sum_{k=1}^{K}\lambda_k \eta_k^2 \leq a$ with $\min_{1\leq k \leq K}\lambda_k>0$.
  Note that
  $$\var  \Big(\eta_1 \mid \sum_{k=1}^{K}\lambda_k \eta_k^2 \leq a \Big) = \alpha^{-1} \int_{\sum_{k=1}^{K} \lambda_k x_k^2 \leq a}x_1^2\prod_{k=1}^K\phi (x_k)~\textup{d}x_1\cdots \textup{d}x_K,~\text{where}~ \alpha={\operatorname{pr}}\Big (\sum_{k=1}^{K}\lambda_k \eta_k^2 \leq a \Big),$$
  and
  \begin{equation*}
    \begin{split}
      \int_{\sum_{k=1}^{K} \lambda_k x_k^2 \leq a}x_1^2\prod_{k=1}^K\phi (x_k)~dx_1\cdots \textup{d}x_K &= (2\pi)^{-K/2}\Big (\prod_{k=1}^K \lambda_k\Big)^{-1/2}\lambda_1^{-1}\int_{B ({a}^{1/2})} \{ x_1^2+O(r^4) \} ~\textup{d}x_1\cdots \textup{d}x_K.
    \end{split}
  \end{equation*}
  The above approximation follows from the Taylor expansion,
  \[
    x_1^2\exp\Big (2^{-1} \sum_{k=1}^{K} x_k^2/\lambda_k\Big) = x_1^2+O\Big (2^{-1} x_1^2\sum_{k=1}^{K} x_k^2/\lambda_k\Big) = x_1^2+O (r^4),\quad \text{where }~ r^2 = \sum_{k=1}^K x_k^2.
  \]
  By Lemma~\ref{lem:sapp9},
  \[
    \int_{\sum_{k=1}^K \lambda_k x_k^2 \leq a}x_1^2\prod_{k=1}^K\phi (x_k)~\textup{d}x_1\cdots \textup{d}x_K= (2\pi)^{-K/2}\Big (\prod_{k=1}^K \lambda_k\Big)^{-1/2}\lambda_1^{-1} \frac{1}{(K+2)K}a^{K/2+1}s_K+O(a^{K/2+2}).
  \]
  Substitute $a$ with the result of Lemma~\ref{lem:sapp10}, we have
  \[
    \var  \Bigl (\eta_1 \,\Bigm |\, \sum_{k=1}^{K}\lambda_k \eta_k^2 \leq a\Bigr ) = p_K \lambda_1^{-1}\Bigl (\prod_{k=1}^K \lambda_k\Bigr )^{1/K}\alpha^{2/K}+o(\alpha^{2/K}).
  \]
  Similarly, for $k=1,\dots,K$,
  \[
    \var  \Bigl (\eta_k \,\Bigm |\, \sum_{k=1}^{K}\lambda_k \eta_k^2 \leq a\Bigr ) = p_K \lambda_k^{-1}\Bigl (\prod_{k=1}^K \lambda_k\Bigr )^{1/K}\alpha^{2/K}+o(\alpha^{2/K}).
  \]
  Moreover, for $m\neq n$, $\eta_m$ and $\eta_n$ are conditional uncorrelated, then
  \begin{align*}
    E \Bigl(\eta_m\eta_n\,\Bigm |\, \sum_{k=1}^{K}\lambda_k \eta_k^2 \leq a\Bigr)
     & = E\biggl\{\eta_mE\Bigl(\eta_n\,\Bigm | \,\eta_m, \sum_{k=1}^{K}\lambda_k \eta_k^2 \leq a\Bigr)\biggm | \sum_{k=1}^{K}\lambda_k \eta_k^2 \leq a\biggr\} \\
     & =E\Bigl(\eta_m\times0\Bigm |\sum_{k=1}^{K}\lambda_k \eta_k^2 \leq a\Bigr) = 0.
  \end{align*}
  The last line follows from the fact that given $\eta_m$ and a symmetric ellipsoidal acceptance region, the conditional distribution of $\eta_n$ is symmetric.
  Therefore,
  \begin{eqnarray}
    \var  \Bigl (\sum_{k=1}^K \zeta_k\eta_k \Bigm | \sum_{k=1}^{K}\lambda_k \eta_k^2 \leq a\Bigr ) & = & \sum_{k=1}^K \zeta_k^2\var  \Bigl (\eta_k \,\Bigm |\, \sum_{k=1}^{K}\lambda_k \eta_k^2 \leq a\Bigr ) \nonumber \\
    &= & p_K \Bigl (\sum_{k=1}^K\frac{\zeta_k^2}{\lambda_k}\Bigr )\Bigl (\prod_{k=1}^K \lambda_k\Bigr )^{1/K}\alpha^{2/K}+o(\alpha^{2/K}) . \nonumber
  \end{eqnarray}
  Note that
  \begin{gather*}
    \Big (\prod_{k=1}^K \lambda_k\Big)^{1/K} =  \det( \vxxhaj )^{1/K}\det( A_x )^{1/K},\quad \operatorname{diag}(\lambda_1,\ldots,\lambda_K)={P}_x^{\top}\vxxhaj^{1/2}{A}_x{V}^{1/2}_{xx}{P}_x,\\ \quad(\zeta_1,\ldots,\zeta_K)=\vtxhaj\vxxhaj^{-1/2}P_x/(\vtxhaj\vxxhaj^{-1}\vxthaj)^{1/2}.
  \end{gather*}
  Thus,
  \begin{align*}
    \sum_{k=1}^K \frac{\zeta_k^2}{\lambda_k} & = \vtxhaj\vxxhaj^{-1/2}{P}_x ({P}_x^{\top}\vxxhaj^{1/2} {A}_x\vxxhaj^{1/2}{P}_x)^{-1}{P}_x^{\top}\vxxhaj^{-1/2}\vxthaj/\vtxhaj\vxxhaj^{-1}\vxthaj \\
                                           & =\vtxhaj\vxxhaj^{-1} {A}_x^{-1}\vxxhaj^{-1}\vxthaj/(\vtxhaj\vxxhaj^{-1}\vxthaj).
  \end{align*}
  Therefore,
  \begin{align*}
    \var({T}_{\infty}|{{T}}_{x,\infty}^\top{A}_x{{T}}_{x,\infty}\leq a ) = & \vtthaj \big\{  (1-R^2_{x} ) +   R^2_{x}\var({\mu}_x^{\top}P_x{\eta}\mid
    {\eta}^{\top}{\Lambda}_x{\eta}\leq a) \big\}                                                                                                            \\
    =                                                               & \vtthaj \{  (1-R^2_{x} ) +   {R}_x^2p_K {\nu}_x(A_x) \alpha^{2/K} + o(\alpha^{2/K}) \}.
  \end{align*}
\end{proof}

\subsection{Proof of Theorem~\ref{thm:optimdiag}}
\begin{proof}[of Theorem~\ref{thm:optimdiag}]
  It is convenient to fix $\prod_{k=1}^K w_k = 1$ and minimize the term that is related to $w_k$ for $k=1,\ldots,K$. Denote $(b_1,\ldots,b_K)^{\top}=\vxxhaj^{-1}\vxthaj $. We only need to minimize
  $$
    \vtxhaj\vxxhaj^{-1} {A}_x^{-1}\vxxhaj^{-1}\vxthaj = \sum_{k=1}^K b^2_k w_k^{-1}.
  $$
  Using the inequality of arithmetic and geometric means, we have $$\sum_{k=1}^K b^2_k w_k^{-1} \geq K\Bigl(\prod_{k=1}^K  b^2_k w_k^{-1}\Bigr)^{1/K} = K\Bigl(\prod_{k=1}^K  b^2_k\Bigr)^{1/K}.$$
  The equality holds if and only if
  $$
    \frac{b_1^2}{w_1}=\cdots=\frac{b_K^2}{w_K},$$ 
which implies $ w_k = c_0 b_k^2=c_0 (\vtxhaj\vxxhaj^{-1}{\xi}_k)^2$ for some constant $c_0>0$. Then within all positive diagonal matrix, ${\nu}_x(A_x)$ reach its minimum at $A_x^{\textnormal{opt}} =\operatorname{diag}(w_1,\ldots,w_K)$ with $ w_k = c_0(\vtxhaj \vxxhaj^{-1}{\xi}_k)^2$ for some constant $c_0>0$. The proof for ${\nu}_c(A_c)$ is similar.

\end{proof}

\subsection{Proof of Corollary~\ref{cor:varalphaortho}}
\begin{proof}[of Corollary~\ref{cor:varalphaortho}]
It suffices to show that the optimal weights $w_k = c_0{R}_{x_k}^2/{V}_{\textnormal{haj},x_k x_k}$, $k=1,\ldots,K$, for some  constant $c_0>0$. Denote  $V_{\textnormal{haj},xx} = \operatorname{diag}(V_{\textnormal{haj},x_1x_1},\ldots,V_{\textnormal{haj},x_Kx_K})$. Recall that $ \sum_{k=1}^K {R}^2_{x_k} = R^2_{x} $ with 
  $$
    R_{x_k}^2 =  V_{\textnormal{haj},\tau x_k }^2/( V_{\textnormal{haj},\tau\tau}V_{\textnormal{haj},x_k  x_k }).
$$
  Then
  
  \begin{eqnarray*}
  {V}_{\textnormal{haj},\tau x}{V}_{\textnormal{haj},xx}^{-1} &= & ({V}_{\textnormal{haj},\tau x_{1}}{V}_{\textnormal{haj},x_{1}x_{1}}^{-1},\ldots,{V}_{\textnormal{haj},\tau x_{K}}{V}^{-1}_{\textnormal{haj},x_{K}x_{K}}) \\
  &=& ({R}_{x_1}{V}_{\textnormal{haj},\tau\tau}^{1/2}/{V}_{\textnormal{haj},x_1x_1}^{1/2},\ldots,{R}_{x_K}{V}_{\textnormal{haj},\tau\tau}^{1/2}/{V}_{\textnormal{haj},x_Kx_K}^{1/2}).
  \end{eqnarray*}
  We minimize the term related to ${A}_{x}$ in ${\nu}_{x}(A_{x})$. Note that 
  \begin{align*}
     & {V}_{\textnormal{haj},\tau x}{V}_{\textnormal{haj},xx}^{-1}{A}_{x}^{-1}{V}_{\textnormal{haj},xx}^{-1}{V}_{\textnormal{haj},x\tau}\det( A_{x} )^{1/K}\det( V_{\textnormal{haj},xx} )^{1/K} \\
     & = \sum_{k=1}^K {R}^2_{x_k}{V}_{\textnormal{haj},\tau\tau}/({V}_{\textnormal{haj},x_k x_k}w_k)\Bigl(\prod_{k=1}^K w_k{V}_{\textnormal{haj},x_k x_k}\Bigr)^{1/K}                             \\
     & \geq K{V}_{\textnormal{haj},\tau\tau} \Bigl(\prod_{k=1}^K {R}_{x_k}^2\Bigr)^{1/K},
  \end{align*}
  where the last inequality follows from the inequality of arithmetic and geometric means.
  Therefore,
  $$
    {\nu}_{x} \big(A_{x}\big) \geq K{V}_{\textnormal{haj},\tau\tau} \Bigl(\prod_{i=1}^K {R}_{x_k}^2\Bigr)^{1/K} \Bigm/ ({V}_{\textnormal{haj},\tau x}{V}^{-1}_{\textnormal{haj},xx}{V}_{\textnormal{haj},x\tau}) = K \Bigl(\prod_{i=1}^K {R}_{x_k}^2\Bigr)^{1/K} \Bigm/ R^2_{x}.
  $$
  The equality holds if and only if
  $$
    {R}^2_{x_1}/({V}_{\textnormal{haj},x_1x_1}w_1) = \cdots ={R}^2_{x_K}/({V}_{\textnormal{haj},x_Kx_K}w_K),
  $$
  which implies $w_k = c_0{R}_{x_k}^2/{V}_{\textnormal{haj},x_k x_k}$ for some  constant $c_0>0$.
\end{proof}
\subsection{Proof of Corollary~\ref{cor:varalphatier}}

\begin{proof}[of Corollary~\ref{cor:varalphatier}]
  By Theorem~3 of \cite{Li9157},
  $$
    M^{1/2}(\tauhaj-\tau) \mid \mathcal{M}_{{[1]}},\ldots,\mathcal{M}_{{[L]}} \mathrel{\xrightarrow{d}} {V}_{\textnormal{haj},\tau\tau}^{1/2}\Big\{(1-R^2_{x})^{1/2}\epsilon + \sum_{l=1}^L {R}_{x_{[l]}} L_{K_l,a_{[l]}} \Big\},
  $$
  where $\epsilon$ is a standard normal random variable, and $\epsilon,L_{K_1,a_{[1]}},\ldots,L_{K_l,a_{[l]}}$ are jointly independent.
  Applying our Theorems~\ref{thm:asymp2} and \ref{thm:varalpha} with $K=K_l$, $\alpha = \alpha_{[l]}$, $\vtthaj=1$, $A_x = \vxxhaj^{-1}$, and $R^2_{x} = 1$, we have
  $$
    \var(L_{K_l,a_l}) = p_{K_l}\alpha_{[l]}^{2/K_l} + o(\alpha_{[l]}^{2/K_l}).
  $$
  Therefore,
  $$
    \vara \bigl\{M^{1/2}(\tauhaj-\tau) \mid \mathcal{M}_{[1]},\ldots,\mathcal{M}_{[L]} \bigr\} =  {V}_{\textnormal{haj},\tau\tau}\Bigl\{(1-R^2_{x})+ \sum_{l=1}^L p_{K_l}{R}_{x_{[l]}}^2 \alpha_{[l]}^{2/K_l} + o(\alpha_{[l]}^{2/K_l})\Bigr\}.
  $$
 The leading term in $\vara \bigl\{M^{1/2}(\tauhaj-\tau) \mid \mathcal{M}_{[1]},\ldots,\mathcal{M}_{[L]} \bigr\}$ has the following lower bound:
  \begin{align*}
    {V}_{\textnormal{haj},\tau\tau}\bigl\{(1-R^2_{x})+ \sum_{l=1}^L p_{K_l}{R}_{x_{[l]}}^2 \alpha_{[l]}^{2/K_l}\bigr\} & \geq{V}_{\textnormal{haj},\tau\tau}\biggl[(1-R^2_{x}) +  K \biggl\{\prod_{l=1}^L\Bigl(\frac{p_{K_l} {R}_{x_{[l]}}^2\alpha_{[l]}^{2/K_l}}{K_l}\Bigr)^{K_l}\biggr\}^{1/K} \biggr] \\
                                                                                                               & ={V}_{\textnormal{haj},\tau\tau}\biggl[(1-R^2_{x}) +  K \biggl\{\prod_{l=1}^L\Bigl(\frac{p_{K_l} {R}_{x_{[l]}}^2}{K_l}\Bigr)^{K_l}\biggr\}^{1/K}\alpha^{2/K} \biggr] .
  \end{align*}
  The equality holds if and only if
  $$
    p_{K_1} {R}_{x_{[1]}}^2\alpha_{[1]}^{2/K_1}/{K_1} = \cdots = p_{K_L} {R}_{x_{[L]}}^2\alpha_{[L]}^{2/K_L}/{K_L},
  $$
  which implies that for some constant $c_0 > 0$,
  $$
    \alpha_{[l]}=\left({c_0{R}_{x_{[l]}}^2 p_{K_{l}}}/{K_{l}}\right)^{-K_{l}/2} \quad (l=1,\ldots,L).
  $$
\end{proof}

\subsection{Proof of Theorem~\ref{thm:wtedbetter}}
To prove Theorem~\ref{thm:wtedbetter}, we need Lemma~\ref{lem:pk} below.
\begin{lemma}\label{lem:pk}
  $$
    p_{K}=\frac{2\pi}{K+2}\left\{ \frac{2\pi^{K/2}}{K\Gamma (K/2)} \right\}^{-2/K} \quad (K=1,2,\ldots)
  $$
  is decreasing in $K$.
\end{lemma}

\begin{proof}[of Lemma~\ref{lem:pk}]
  First, 
  \[
    \begin{aligned}
      \log p_{K} & =\log (2 \pi)-\log (K+2)-2 \{ \log (2/K) + (K/2)\log \pi-\log \Gamma (K/2) \}/K \\
                 & = \log (2 \pi) -\log \{ \pi (K+2) \}+ (2/K)\log \Gamma (K/2+1).
    \end{aligned}
  \]
  Let $x=K/2$ and we consider the function
  \[
    f (x)=-\log \{ 2\pi (x+1) \}+ (1/x) \log \Gamma (x+1).
  \]
  It suffices to prove that $f (x)$ is decreasing in $x$ for $x\geq 0$.
  Simple calculation gives
  \[
    f'(x)=-1/(x+1)+\psi (x+1)/x- ( 1 /x^{2} ) \log \Gamma (x+1),
  \]
 where $\psi (x)=\textup{d} \log \Gamma (x)/\textup{d}x$ is the Digamma function. Let
  \[
    g (x)=x^{2}f'(x)=-x^{2}/(x+1)+x\psi (x+1)-\log \Gamma (x+1),
  \]
   We only need to show that $g (x)\leq 0$ for $x\geq 0$.
  In fact,
  \[
    \begin{aligned}
      g'(x) & =-x (x+2)/(x+1)^{2}+\psi (x+1)+x\psi' (x+1)-\psi (x+1) \\
            & =-x (x+2)/(x+1)^{2}+x\psi'(x+1).
    \end{aligned}
  \]
By \citet[][Theorem 1]{alzer1997some}, $\log x-\psi (x)-1/x$ is an increasing function for $x>0$. Thus, we have $1/x+1/x^2-\psi' (x)$ is positive for $x>0$, which implies that $\psi' (x+1)<(x+2)/(x+1)^{2}$ for $x\geq 0$. Therefore, $g'(x) \leq 0$ and $g (x)\leq g (0)=0$ for $x\geq 0$.
\end{proof}

\begin{proof}[of Theorem~\ref{thm:wtedbetter}]
  By Lemma~\ref{lem:pk}, $p_K$ is decreasing in $K$. Therefore,
  \begin{align*}
    \sum_{l=1}^L R_{x_{[l]}}^2 p_{K_{l}}\alpha_{[l]}^{2/K} & \geq \biggl\{\prod_{l=1}^L\Bigl(\frac{p_{K_l} {R}_{x_{[l]}}^2}{K_l}\Bigr)^{K_l}\biggr\}^{1/K}\alpha^{2/K} \geq  p_K\biggl\{\prod_{l=1}^L\Bigl(\frac{ {R}_{x_{[l]}}^2}{K_l}\Bigr)^{K_l}\biggr\}^{1/K}\alpha^{2/K} \\
                                                           & \geq Kp_K\Bigl(\prod_{k=1}^K R_{x_k}^2\Bigr)^{1/K} \alpha^{2/K}.
  \end{align*}
  The last inequality follows from the inequality of geometric and arithmetic means, and
  $$
    {R}_{x_{[l]}}^2 =   \sum_{k\in [l]} {R}^2_{x_k}, \quad \Bigl(\frac{ {R}_{x_{[l]}}^2}{K_l}\Bigr)^{K_l} \geq \prod_{k\in [l]} {R}^2_{x_k}.
  $$
\end{proof}


\subsection{Proof of Theorem~\ref{thm:tauhtadjasymp}}
Define
  $$
    \tauhtb =  M_1^{-1}{\sum_{i=1}^M Z_i \{\tilde{Y}_{i\cdot}(1)-{v}_i^{\top} {\beta}_{v}(1) \}}-M_0^{-1}{\sum_{i=1}^M (1-Z_i) \{\tilde{Y}_{i\cdot}(0)- {v}_i^{\top} {\beta}_{v}(0)\}},
  $$
  and
  $$
    \tauhajb = N_1^{-1}{\sum_{i,j} Z_{ij} \{{Y}_{ij}(1)- {w}_{ij}^{\top} {\beta}_{w}(1) \}}-N_0^{-1}{\sum_{i,j} (1-Z_i) \{Y_{ij}(0)- {w}_{ij}^{\top} {\beta}_{w}(0)\}}.
  $$

  Lemma~\ref{sup:lembeta} below shows that $\tauht$ and $\tauhtb$ have the same asymptotic distributions under cluster rerandomization. The same is true for $\tauhaj$ and $\tauhajb$.

  \begin{lemma}
    \label{sup:lembeta}
    If Conditions~\ref{cond:1}, \ref{cond:2}, \ref{cond:4}, \ref{cond:6} hold, then
    $$
      M^{1/2}(\tauadjht -\tauhtb)|\mathcal{D}_{{c}}(A_c)=o_{\textup{p}}(1). 
    $$
    If Condition~\ref{cond:1}, \ref{cond:3}, \ref{cond:5}, \ref{cond:7} hold,  then
    $$
      M^{1/2}(\tauadjhaj-\tauhajb)|\mathcal{D}_{{x}}(A_x)=o_{\textup{p}}(1). 
    $$
  \end{lemma}

  \begin{proof}[of Lemma~\ref{sup:lembeta}]
    \label{sup:propbetas}
    By Lemma~A7 of \cite{su2021modelassisted}, $M^{1/2}(\tauadjht -\tauhtb) = o_{\textup{p}}(1)$, which implies that, for any $\delta>0$,
    $$
      \operatorname{pr}\left\{|M^{1/2}(\tauadjht -\tauhtb)|>\delta\right\}=o(1).
    $$
    Proposition~\ref{prop:FCLT} implies that
    $$
      \lim\limits_{M \rightarrow \infty}\operatorname{pr}\{\mathcal{D}_{c}(A_c)\} >0.
    $$
    Therefore,
    $$
      \operatorname{pr}\left\{ |M^{1/2}(\tauadjht -\tauhtb)|>\delta \mid {D}_{c}(A_c)\right\} \leq \operatorname{pr}\left\{ |M^{1/2}(\tauadjht -\tauhtb)|>\delta \right\}\Bigm/\operatorname{pr}\{{D}_{c}(A_c)\}= o(1),
    $$
    which implies that
    $$
      M^{1/2}(\tauadjht -\tauhtb)\mathrel{|}\mathcal{D}_{c}(A_c) = o_{\textup{p}}(1).
    $$
    Similarly, we can prove the second part of Lemma~\ref{sup:lembeta}.
    
  \end{proof}

\begin{proof}[of Theorem~\ref{thm:tauhtadjasymp}]
  The asymptotic results for $\tauadjht$ follow from  Lemma~\ref{sup:lembeta} and Theorems~\ref{thm:asymp2} and \ref{thm:asymquadra} with $\tauht = \tauhtb$. The asymptotic results for $\tauadjhaj$ follow from  Lemma~\ref{sup:lembeta} and Theorems~\ref{thm:asymp2} and \ref{thm:asymquadra} with $\tauhaj = \tauhajb $.

\end{proof}

\subsection{Proof of Corollary~\ref{coro:anamoreinfo}}
\begin{proof}[of Corollary~\ref{coro:anamoreinfo}]
  By Lemma~\ref{sup:lembeta}, $M^{1/2}(\tauadjht-\tau) \mathrel{\ \dot{\sim}\ } M^{1/2}(\tauhtb-\tau)$ under cluster rerandomization. When ${c}_i = B{v}_i$, $M^{1/2}(\tauhtb-\tau)$ is uncorrelated with $M^{1/2}\taucht$. Applying Theorem~\ref{thm:asymquadra}, the conclusion follows immediately.
\end{proof}


\subsection{Proof of Theorem~\ref{thm:vhwvlzop}}
\begin{proof}[of Theorem~\ref{thm:vhwvlzop}]
  By Theorems 2 and 4 of \cite{su2021modelassisted}, we have

  \begin{align*}
   M \vhatlz & \mathrel{\xrightarrow{\textup{p}}} e_1^{-1} \varf  \{\tdadje (1)\} +
    e_0^{-1}\varf  \{\tdadje (0)\}\geq \vtthajb,                                 \\
   M \vhathw & \mathrel{\xrightarrow{\textup{p}}}e_1^{-1} \varf  \{\tdadjy(1)\} +
    e_0^{-1}\varf  \{\tdadjy(0)\}\geq \vtthtb .
  \end{align*}
  As
  $
    \lim\limits_{M \rightarrow \infty}\operatorname{pr}\{{D}_{\textnormal{*}}(A_\textnormal{*})\} >0, \ \text{for } * = x,c,
  $
  similar to the proof of Lemma~\ref{sup:lembeta}, we have
  \begin{align*}
   M \vhatlz - \vtthajb\mid \mathcal{D}_{x}(A_x) \mathrel{\geq_{\textup{p}}} 0, \quad
    M \vhathw - \vtthtb\mid \mathcal{D}_{c}(A_c) \mathrel{\geq_{\textup{p}}} 0.
  \end{align*}
\end{proof}

\subsection{Proof of Theorem~\ref{thm::regression-adjustment-under--cluster-rerandomization}}

\begin{proof}[of Theorem~\ref{thm::regression-adjustment-under--cluster-rerandomization}]
   When $v_i=c_i$ and $x_{ij}=w_{ij}$, we have $\radjc = 0$, $\radjx \geq 0$, and
  \[
    \vtthajb \{ 1- (\radjx)^2 \} = \vtthaj(1-R^2_x), \quad  \vtthtb \{ 1 - (\radjc)^2 \} = \vttht(1-R^2_c).
  \]
  By Theorem~\ref{thm:tauhtadjasymp}, we have
  $$
   M^{1/2}(\tauadjht-\tau)\mid \mathcal{D}_{c}(A_c)  \mathrel{\ \dot{\sim}\ } (\vttht)^{1/2}
    (1-R^2_{c}) \epsilon,
  $$
  and
  $$
   M^{1/2}(\tauadjhaj-\tau)\mid \mathcal{D}_{x}(A_x) \mathrel{\ \dot{\sim}\ } (\vtthajb)^{1/2} \big[
    \{1-(R^{\textnormal{adj}}_{ x})^{2}\}^{1/2} \epsilon+
    R^{\textnormal{adj}}_{ x}{\mu}_x^{\top}{\eta} \mid {\eta}^{\top}\vxxhaj^{1/2}A_x\vxxhaj^{1/2}{\eta} \leq a
    \big].
  $$
Moreover, by Theorem~\ref{thm:vhwvlzop}, $M \vhatadjlz$ and $M \vhatadjhw$ are conservative estimators of $\vtthajb$ and $\vtthtb$, respectively. The conservativeness of the Wald-type confidence intervals follows from the conservativeness of $M\vhatadjlz$ and $M\vhatadjhw$ and Proposition~\ref{prop:concentration-symmetric-unimodality}.
  Therefore, (i) and (ii) hold.

  When ${c}_{i}=( n_i ,\tilde{{x}}_{i \cdot }^{\top})^{\top}$,  Corollary~\ref{coro:clusterbetter} implies that $$\vtthajb\{ 1- (\radjx)^2 \} \geq \vtthtb \{ 1 - (\radjc)^2\}.$$
  Thus, the normal component of the asymptotic distribution of the cluster-level regression adjustment is more concentrated than that of the individual-level regression adjustment. Moreover, the asymptotic distribution of the individual-level regression adjustment has a truncated normal component while the cluster-level regression adjustment does not. Thus, (iii) holds.
\end{proof}

\subsection{Proof of Corollary~\ref{coro:analessinfo}}

\begin{proof}[of Corollary~\ref{coro:analessinfo}]
  By Lemma~\ref{sup:lembeta}, $M^{1/2}(\tauadjht-\tau) \mathrel{\ \dot{\sim}\ } M^{1/2}(\tauhtb-\tau)$ under cluster rerandomization. By Proposition~1 of \cite{2020Rerandomization},
  \[
    \tauhtb-\tau = \tauht-\tau-\operatorname{cov}(\tauht,\tauvht)\{\operatorname{cov}(\tauvht)\}^{-1}\tauvht.
  \]
  Since ${v}_i = B{c}_i$, then $\tauvht = B\taucht$. As $M\vara(\taucht) = \vccht$ and $M\cova(\tauht,\taucht) = \vtcht$, substituting them into the above equation, we have
  \begin{align*}
    \tauhtb-\tau  \mathrel{\ \dot{\sim}\ } & \tauht-\tau -  \vtcht B^{\top}(B\vccht B^{\top})^{-1}B \taucht                                                                \\
    =              & \tauht-\tau -\vtcht\vccht^{-1}\taucht+\vtcht\vccht^{-1}\taucht- \vtcht B^{\top}(B\vccht B^{\top})^{-1}B \taucht           \\
    =              & \{\tauht-\tau -\vtcht\vccht^{-1}\taucht\}                                                                                   \\
                   & + \vtcht\vccht^{-1/2}\left\{I_K-\vccht^{1/2}B^{\top}(B\vccht B^{\top})^{-1}B\vccht^{1/2}\right\}\vccht^{-1/2}\taucht.
  \end{align*}
  Note that
  \[
    \vara\{M^{1/2}(\tauht-\tau -\vtcht\vccht^{-1}\taucht)\} = (1-R^2_{ c})\vttht, \]
  \[ \cova\{M^{1/2}(\tauht-\tau -\vtcht\vccht^{-1}\taucht),M^{1/2}\taucht\} = 0.
  \]
  Therefore, by Proposition~\ref{cor:weakconv}, 
  we have
  \begin{align*}
     & M^{1/2}(\tauhtb-\tau)|\mathcal{D}_c(A_c)  \mathrel{\ \dot{\sim}\ } \vttht^{1/2}(1-R^2_{{c}})^{1/2}\epsilon + \\ &\quad \quad \quad  \vtcht\vccht^{-1/2}\left\{I_K-\vccht^{1/2}B^{\top}(B\vccht B^{\top})^{-1}B\vccht^{1/2}\right\}{\eta} \mid {\eta}^{\top}\vccht^{1/2} A_c\vccht^{1/2}{\eta}\leq a,
  \end{align*}
  which concludes the proof.
\end{proof}

\subsection{Proof of Corollary~\ref{cor:lessinfocovotho}}
\begin{proof}[of Corollary~\ref{cor:lessinfocovotho}]
Let
  $
    {\mu}_c^\top = (\mu_1,\ldots,\mu_K)=\vtcht \vccht^{-1/2}/\{\vtcht \vccht^{-1}\vctht\}^{1/2}.
  $
  It suffices to compare the variances of
  $$
  (\vttht)^{1/2}\big\{ (1-R^2_{ c})^{1/2} \epsilon + R_{ c}{\mu}_c^{\top}H{\eta}\mid {\eta}^{\top}\vccht^{1/2}A_c\vccht^{1/2}{\eta} \leq a\big\}
  $$
  and
  $$
  (\vttht)^{1/2}\big\{ (1-R^2_{ c})^{1/2} \epsilon + R_{ c}{\mu}_c^{\top}{\eta}\mid {\eta}^{\top}\vccht^{1/2}A_c\vccht^{1/2}{\eta} \leq a\big\},
  $$
  where
  $$
    H =
    \begin{bmatrix}
      0 & 0                      \\
      0 & I_{(K-J) \times (K-J)}
    \end{bmatrix}.
  $$
  Due to the conditional uncorrelatedness of $\eta_m$ and $\eta_n$ for $m\not=n$, the difference between the variances of the above two distributions is equal to
  $$
    \vttht {R}_{c}^2 \sum_{k=1}^J \mu^2_k\operatorname{var}(\eta_k\mid{\eta}^{\top}\vccht^{1/2}A_c\vccht^{1/2}{\eta} \leq a)\geq 0,
  $$
  which concludes the  proof.

\end{proof}

\subsection{Proof of Theorem~\ref{thm:interval-conservative}}

Let $\hbetav (1)$ and $\hbetav (0)$ be the coefficients of $v$ in the least squares fit of $\tdyi$ on $(1,v_i)$ under treatment and control arms, respectively.
Let $\hbetaw(1)$ and $\hbetaw (0)$ be the coefficients of $w$ in the least squares fit of $Y_{ij}$ on $(1,w_{ij})$ under treatment and control arms, respectively. Lemmas~\ref{lem:beta-consistensy}--\ref{lem:sample-cov-consistensy} below  show the consistency of the ordinary least squares coefficients, sample means, variances, and covariances under cluster rerandomization. Their unconditional versions can be found in \cite{su2021modelassisted} and \cite{Li2020factorial}. The proof is similar to that of Lemma~\ref{sup:lembeta}, so we omit. Let $\bary_z$ be the sample mean of $Y_{ij}$ under treatment arm $z$.
\begin{lemma}
If Conditions~\ref{cond:1}, \ref{cond:2}, \ref{cond:4}, \ref{cond:6} hold, then
  \label{lem:beta-consistensy}
  \begin{align*}
    \hbetav (z)-\betav (z) \mid \dca \xrightarrow{\textnormal{p}} 0.
  \end{align*}
  If Condition~\ref{cond:1}, \ref{cond:3}, \ref{cond:5}, \ref{cond:7} hold,  then
  \begin{align*}
          \hbetaw (z)-\betaw(z) \mid \dxa \xrightarrow{\textnormal{p}} 0.
  \end{align*}
\end{lemma}

\begin{lemma}
  \label{lem:mean-consistensy}
  If Condition~\ref{cond:1}, \ref{cond:3}, \ref{cond:5}, \ref{cond:7} hold,  then
  \begin{align*}
    \bary_z-\bar{Y}(z) \mid \dxa \xrightarrow{\textnormal{p}} 0.
  \end{align*}
\end{lemma}

\begin{lemma}
  \label{lem:sample-cov-consistensy}
  If Conditions~\ref{cond:1}, \ref{cond:2}, \ref{cond:4}, \ref{cond:6} hold, then
  $$
  \hvarfz\{\tde(z)\} - \var_{\textnormal{f}}\{\tde(z)\}\mid \dxa \xrightarrow{\textnormal{p}} 0,  \quad \hcovfz(\tilde{\omega}) - \cov_{\textnormal{f}}(\tilde{\omega})\mid \dxa \xrightarrow{\textnormal{p}} 0, 
$$  
$$
\hcovfz(\tdx)-\covf(\tdx)\mid\dxa \xrightarrow{\textnormal{p}} 0,\quad \hcovfz(\tilde{W})-\covf(\tilde{W})\mid\dxa \xrightarrow{\textnormal{p}} 0,
$$
$$
    \hcovfz ( \tdx,\tilde{W} ) -\covf(\tdx,\tilde{W})\mid\dxa \xrightarrow{\textnormal{p}} 0,    \quad \hcovfz\{ \tdx,\tde(z)\}-\covf\{\tdx,\tde(z)\}\mid\dxa \xrightarrow{\textnormal{p}} 0,
    $$
    $$
    \hcovfz(\tdx,\tilde{\omega})-\covf(\tdx,\tilde{\omega})\mid\dxa \xrightarrow{\textnormal{p}} 0, \quad   \hcovfz\{\tilde{W},\tde(z)\}-\covf\left\{\tilde{W},\tde(z)\right\}\mid\dxa \xrightarrow{\textnormal{p}} 0,
    $$
    $$
    \hcovfz(\tilde{W},\tilde{\omega})-\covf(\tilde{W},\tilde{\omega})\mid\dxa \xrightarrow{\textnormal{p}} 0, \quad \hcovfz\{\tilde{\omega},\tde(z)\}-\covf\left\{\tilde \omega,\tde(z)\right\}\mid\dxa \xrightarrow{\textnormal{p}} 0.
    $$
  If Condition~\ref{cond:1}, \ref{cond:3}, \ref{cond:5}, \ref{cond:7} hold,  then
  $$
  \hvarfz\{\tdy(z)\} - \var_{\textnormal{f}}\{\tdy(z)\}\mid \dca \xrightarrow{\textnormal{p}} 0,  \quad \hcovfz(V) - \cov_{\textnormal{f}}(V)\mid \dca \xrightarrow{\textnormal{p}} 0, 
  $$
  $$
    \hcovfz\{\tdy(z),V\}-\covf\{\tdy(z),V\}\mid\dca \xrightarrow{\textnormal{p}} 0,\quad \hcovfz(C)-\covf(C)\mid\dca \xrightarrow{\textnormal{p}} 0 ,
  $$
  $$
  \hcovfz (C,V) -\covf(C,V)\mid\dca \xrightarrow{\textnormal{p}} 0,  \quad \hcovfz\{C,\tdy(z)\}-\covf\{C,\tdy(z)\}\mid\dca \xrightarrow{\textnormal{p}} 0.
  $$

\end{lemma}


By the property of projection, Lemma~\ref{lem:lower-bound-projection} below holds.

\begin{lemma}
  \label{lem:lower-bound-projection}
  \begin{align*}
        {\cov_{\textnormal{f}}  } \{ \tdadje(1)-\tdadje(0), \tilde{G} \} \{\covf (\tilde{G})\}^{-1}  {\cov_{\textnormal{f}}  } \{ \tilde{G},\tdadje(1)-\tdadje(0) \} &\leq \varf\{\tdadje(1)-\tdadje(0)\},\\
                {\cov_{\textnormal{f}}  } \{ \tdadjy(1)-\tdadjy(0), C \} \{\covf (C)\}^{-1}  {\cov_{\textnormal{f}}  } \{ C,\tdadjy(1)-\tdadjy(0) \} &\leq \varf\{\tdadjy(1)-\tdadjy(0)\}.
  \end{align*}
\end{lemma}

\begin{proof}[of Theorem~\ref{thm:interval-conservative}]
Recall that
\begin{align*}
    (\hradjc)^2 &= \big[e_1^{-1} \hvarf \{\tdadjy(1)\mid {C}\} +
  e_0^{-1}\hvarf \{\tdadjy(0)\mid {C}\} -
  \hvarf \{\tdadjy(1)-\tdadjy(0)\mid {C}\}\big]\bigm /{\hatvtthtb},\\
    (\hradjx)^2 &= \big[{e_1^{-1} \hvarf \{\tdadje(1) \mid \tilde{{X}}\} +
  e_0^{-1}\hvarf \{\tdadje(0)\mid \tilde{{X}}\} -
  \hvarf \{\tdadje (1)-\tdadje (0)\mid \tilde{{X}}\}}\big]\bigm /{\hatvtthajb}.
\end{align*}
First, we prove that the denominators, $\hatvtthajb$ and $\hatvtthtb$, are conservative. Recall that
\begin{align*}
    \hatvtthajb &= {e_1^{-1} \hvarf \{\tdadje(1)\} +
  e_0^{-1}\hvarf \{\tdadje(0)\} -
  \hvarf \{\tdadje (1)-\tdadje (0)\mid \tilde{G}\}},\\
 \hatvtthtb &=  e_1^{-1} \hvarf \{\tdadjy(1)\} +
  e_0^{-1}\hvarf \{\tdadjy(0)\} -
  \hvarf \{\tdadjy(1)-\tdadjy(0)\mid C\}.
\end{align*}
Note that
$$
    \hvarf\{ \tdadjy(z)\} =  \hvarfz\{ \tdy(z)-\hbetav^\top(z) V\},
$$
$$
\hvarf\{\tdadje(z)\} =  \hvarfz\{\tdy(z)-\tilde{\omega}\bary_z-\hbetav^\top(z) \tilde{W}\} = \hvarfz\bigl[\tde(z)-\tilde{\omega}\{\bary_z-\bar{Y}(z)\}-\hbetav^\top(z) \tilde{W}\bigr].
$$
Therefore, by Lemmas~\ref{lem:beta-consistensy}--\ref{lem:sample-cov-consistensy}, we have
  \begin{align*}
    \hvarf\{ \tdadjy(z)\}\mid \dca \xrightarrow{\textnormal{p}} \var_{\textnormal{f}}\{ \tdadjy(z)\},\quad \hvarf\{\tdadje(z)\} \mid \dxa \xrightarrow{\textnormal{p}} \var_{\textnormal{f}}\{\tdadje(z)\}.
  \end{align*}
Recall that
$$
   \hvarf   \{\tdadje (1)-\tdadje (0)\mid \tilde{{G}} \}= \{ \hcovft ( \tilde{U}, \tilde{{G}} ) - \hcovfc ( \tilde{U}, \tilde{{G}} ) \} \{\covf (\tilde{{G}})\}^{-1} \{ \hcovft ( \tilde{{G}},\tilde{U} ) - \hcovfc ( \tilde{{G}},\tilde{U} )  \}.\\
$$
As $\tilde{{G}}$ is the union  of $\tilde{X}$ and $\tilde{W}$,
$$
\hcovfz( \tilde{U}, \tilde{{G}} )\mid \dxa \xrightarrow{\textnormal{p}}\cov_{\textnormal{f}} \{ \tdadje(z), \tilde{{G}} \},
$$
which implies that
\begin{align*}
   \hvarf   \{\tdadje (1)-\tdadje (0)\mid \tilde{{G}} \}\mid \dxa \xrightarrow{\textnormal{p}}
{\cov_{\textnormal{f}}  } \{ \tdadje(1)-\tdadje(0), \tilde{{G}}\} \{\covf (\tilde{{G}})\}^{-1}  {\cov_{\textnormal{f}}  } \{ \tilde{{G}},\tdadje(1)-\tdadje(0) \}.
\end{align*}

By Lemma~\ref{lem:lower-bound-projection}, the right-hand side is smaller than or equal to $ \varf\{\tdadje(1)-\tdadje(0)\}$. Similarly, the probability limit of $\hvarf \{\tdadjy(1)-\tdadjy(0)\mid C\}$ is smaller than or equal to ${\varf} \{\tdadjy(1)-\tdadjy(0)\}$. 
Therefore, $\hatvtthajb$ is a conservative estimator of $\vtthajb$, and $\hatvtthtb$ is a conservative estimator of $\vtthtb$.

Second, we prove that  the estimated numerators for the multiple correlation coefficients are consistent.
Note that $$ \hcovf\{ \tdadjy(z),C\} =  \hcovfz\{ \tdy(z)-\hbetav^\top(z) V,C\},$$
$$
    \hcovf\{\tdadje(z),\tilde{X}\} =  \hcovfz\{\tdy(z)-\tilde{\omega}\bary_z-\hbetav^\top(z) \tilde{W},\tilde{X}\} = \hcovfz\left[\tde(z)-\tilde{\omega}\{\bary_z-\bar{Y}(z)\}-\hbetav^\top(z) \tilde{W},\tilde{X}\right].
$$
  Therefore, by Lemmas~\ref{lem:beta-consistensy}--\ref{lem:sample-cov-consistensy}, we have
  \begin{align*}
    \hcovf\{\tdadjy(z),C\}\mid \dca \xrightarrow{\textnormal{p}} \covf\{\tdadjy(z),C\},\quad \hcovf\{\tdadje(z),\tilde{X}\} \mid \dxa \xrightarrow{\textnormal{p}} \covf\{\tdadje(z),\tilde{X}\}.
  \end{align*}
Together with the following property by Lemma~\ref{lem:sample-cov-consistensy},
  \begin{align*}
    \hcovfz( \tdx ) -\covf(\tdx)\mid\dxa \xrightarrow{\textnormal{p}} 0, & \quad \hcovfz ( C )-\covf(C)\mid\dca \xrightarrow{\textnormal{p}} 0,
  \end{align*}
  we have
  \begin{align*}
    e_1^{-1} \hvarf \{\tdadje(1) \mid \tilde{{X}}\} +
    e_0^{-1} \hvarf \{\tdadje(0)\mid \tilde{{X}}\} -
    \hvarf \{\tdadje (1)-\tdadje (0)\mid \tilde{{X}}\}\mid\dxa \xrightarrow{\textnormal{p}} \\ e_1^{-1} \varf \{\tdadje(1) \mid \tilde{{X}}\} +
    e_0^{-1}\varf \{\tdadje(0)\mid \tilde{{X}}\} -
    \varf \{\tdadje (1)-\tdadje (0)\mid \tilde{{X}}\},
  \end{align*}
  \begin{align*}
      e_1^{-1} \hvarf \{\tdadjy(1)\mid {C}\} +
  e_0^{-1}\hvarf \{\tdadjy(0)\mid {C}\} -
  \hvarf \{\tdadjy(1)-\tdadjy(0)\mid {C}\} \mid \dca  \xrightarrow{\textnormal{p}}\\  e_1^{-1} {\var_{\textnormal{f}}} \{\tdadjy(1)\mid {C}\} +
  e_0^{-1}{\varf} \{\tdadjy(0)\mid {C}\} -
  {\varf} \{\tdadjy(1)-\tdadjy(0)\mid {C}\}.
  \end{align*}
Thus, the estimated numerators for the multiple correlation coefficients are consistent. Similarly, we can prove the consistency of $\hat{\mu}_x$ and $\hat{\mu}_c$. The conclusion follows immediately.

\end{proof}

\subsection{Proof of Corollaries~\ref{cor:varalphaortho-cluster-covariate} and \ref{cor:varalpha-tier-cluster-covariate}, and Theorem~\ref{thm:wtedbetter-cluster-covariate}}

Cluster rerandomization procedures with individual-level covariates and cluster-level covariates have the same form of asymptotic distributions. The proof is just a slight modification of the proof of Corollaries~\ref{cor:varalphaortho}, \ref{cor:varalphatier}, and Theorem~\ref{thm:wtedbetter}.

\end{singlespace}
\end{document}